\documentclass[11pt,a4paper]{scrartcl}

\usepackage[a4paper,verbose]{geometry}
\pdfoutput=1

\usepackage[utf8]{inputenc}
\usepackage[T1]{fontenc} %
\usepackage{libertine}

\usepackage{amsfonts}
\usepackage{amssymb}
\usepackage{amsthm}
\usepackage{amsmath}
\usepackage{caption}
\usepackage{etoolbox}
\usepackage{stmaryrd}
\usepackage{ifthen}
\usepackage{suffix}
\usepackage{verbatim} %

\usepackage{xfrac}

\usepackage{bm}
\usepackage{physics}
\usepackage[safe]{tipa} %
\usepackage{mathtools} %
\usepackage{color}
\usepackage{wrapfig}
\usepackage{tikz}
\usepackage{tikzit}
\usetikzlibrary{cd, babel}

\usepackage[framemethod=tikz]{mdframed}

\newcommand*{\secref}[1]{\S\ref{#1}}

\usepackage[backend=biber,style=numeric-comp,sorting=none,date=short,maxcitenames=2]{biblatex}

\usepackage{hyperref}

\newcommand*{\email}[1]{%
  \normalsize\href{mailto:#1}{\texttt{#1}}\par
}

\bibliography{bibliography}

\newcommand*{\citep}[1]{\parencite{#1}}
\newcommand*{\citet}[1]{\textcite{#1}}

\usepackage[capitalize]{cleveref} %

\usepackage{graphicx}
\usepackage{pict2e}

\def\mdash{{\hbox{-}}}

\makeatletter
\newcommand{\adjunction}{\@ifstar\named@adjunction\normal@adjunction}
\newcommand{\normal@adjunction}[4]{%
  #1\colon #2%
  \mathrel{\vcenter{%
    \offinterlineskip\m@th
    \ialign{%
      \hfil$##$\hfil\cr
      \longrightharpoonup\cr
      \noalign{\kern-.3ex}
      \smallbot\cr
      \longleftharpoondown\cr
    }%
  }}%
  #3 \noloc #4%
}
\newcommand{\named@adjunction}[4]{%
  #2%
  \mathrel{\vcenter{%
    \offinterlineskip\m@th
    \ialign{%
      \hfil$##$\hfil\cr
      \scriptstyle#1\cr
      \noalign{\kern.1ex}
      \longrightharpoonup\cr
      \noalign{\kern-.3ex}
      \smallbot\cr
      \longleftharpoondown\cr
      \scriptstyle#4\cr
    }%
  }}%
  #3%
}
\newcommand{\longrightharpoonup}{\relbar\joinrel\rightharpoonup}
\newcommand{\longleftharpoondown}{\leftharpoondown\joinrel\relbar}
\newcommand\noloc{%
  \nobreak
  \mspace{6mu plus 1mu}
  {:}
  \nonscript\mkern-\thinmuskip
  \mathpunct{}
  \mspace{2mu}
}
\newcommand{\smallbot}{%
  \begingroup\setlength\unitlength{.15em}%
  \begin{picture}(1,1)
  \roundcap
  \polyline(0,0)(1,0)
  \polyline(0.5,0)(0.5,1)
  \end{picture}%
  \endgroup
}
\makeatother

\let\op=\relax

\let\innerprod=\relax
\newcommand{\innerprod}[2]{\left\langle #1 {,} \: #2 \right\rangle}

\def\op{\ensuremath{^{\,\mathrm{op}}}}

\newcommand{\nn}{{\mathbb{N}}}
\newcommand{\rr}{{\mathbb{R}}}
\newcommand{\Tt}{{\mathbb{T}}}

\newcommand{\Cat}[1]{\mathbf{#1}}
\newcommand{\cat}[1]{\mathcal{#1}}
\newcommand{\Fun}[1]{\mathsf{#1}}

\newcommand{\Kl}{\mathcal{K}\mspace{-2mu}\ell}

\DeclareMathOperator*{\E}{\mathbb{E}}

\renewcommand{\d}{\mathrm{d}}

\newcommand{\para}{\Cat{Para}}

\newcommand{\Ea}{{\mathcal{E}}}
\newcommand{\Fa}{{\mathcal{F}}}

\newcommand{\Ha}{{\mathcal{H}}}

\newcommand{\Pa}{{\mathcal{P}}}

\DeclareMathOperator{\id}{\mathsf{id}}

\newcommand{\xto}[1]{\xrightarrow{#1}}

\makeatletter
\newcommand{\mathoverlap}[2]{\mathpalette\mathoverlap@{{#1}{#2}}}
\newcommand{\mathoverlap@}[2]{\mathoverlap@@{#1}#2}
\newcommand{\mathoverlap@@}[3]{\ooalign{$\m@th#1#2$\crcr\hidewidth$\m@th#1#3$\hidewidth}}
\makeatother

\newcommand{\klcirc}{\bullet} %
\newcommand*{\smallklcirc}{\raisebox{0.18ex}{\scalebox{0.66}{$\klcirc$}}}
\newcommand{\klto}{\mathoverlap{\rightarrow}{\smallklcirc\,}}

\newcommand{\xklto}[1]{\mathoverlap{\xrightarrow{#1}}{\smallklcirc\,}}

\def\lenscirc{\baro}
\newcommand{\lensto}{\mathrel{\ooalign{\hfil$\mapstochar\mkern5mu$\hfil\cr$\to$\cr}}}
\newcommand{\xlensto}[1]{\mathoverlap{\xrightarrow{#1}}{\raisebox{0.375ex}{\scalebox{1.0}[0.33]{$|$}}\,}}

\makeatletter
\providecommand*{\xmapstofill@}{%
  \arrowfill@{\mapstochar\relbar}\relbar\rightarrow
}
\providecommand*{\xmapsto}[2][]{%
  \ext@arrow 0395\xmapstofill@{#1}{#2}%
}

\makeatletter
\def\slashedarrowfill@#1#2#3#4#5{%
  $\m@th\thickmuskip0mu\medmuskip\thickmuskip\thinmuskip\thickmuskip
   \relax#5#1\mkern-7mu%
   \cleaders\hbox{$#5\mkern-2mu#2\mkern-2mu$}\hfill
   \mathclap{#3}\mathclap{#2}%
   \cleaders\hbox{$#5\mkern-2mu#2\mkern-2mu$}\hfill
   \mkern-7mu#4$%
}
\def\rightslashedarrowfill@{%
  \slashedarrowfill@\relbar\relbar\mapstochar\rightarrow}
\newcommand\xslashedrightarrow[2][]{%
  \ext@arrow 0055{\rightslashedarrowfill@}{#1}{#2}}
\makeatother

\theoremstyle{definition}
\newtheorem{defn}{Definition}[section]
\newtheorem{notation}[defn]{Notation}

\newtheorem{ex}[defn]{Example}

\newtheorem{rmk}[defn]{Remark}

\newtheorem*{rmk*}{Remark}

\newtheorem{prop}[defn]{Proposition}
\newtheorem{prop*}{Proposition}

\newtheorem{lemma}[defn]{Lemma}
\newtheorem{thm}[defn]{Theorem}
\newtheorem{cor}[defn]{Corollary}

\newtheorem*{thm*}{Theorem}
\newtheorem*{cor*}{Corollary}

\definecolor{darkblue}{rgb}{0,0,0.7} 
\newcommand{\red}[1]{{\color{red} #1}}

\usepackage{authblk}

\author{Toby St. Clere Smithe}
\affil{University of Oxford \\ \& \\ Topos Institute \\ \email{toby@topos.institute}}

\hypersetup{
  pdftitle={Active Inference II},
  pdfauthor={Toby St Clere Smithe}
}

\tikzstyle{xshiftu}=[shift = {(#1, 0)}]
\tikzstyle{yshiftu}=[shift = {(0, #1)}]

\tikzstyle{dot}=[inner sep=0.25mm,minimum width=1mm,minimum height=1mm,draw,shape=circle,text depth=-0.2mm]
\tikzstyle{white dot}=[dot,fill=white, draw=black]
\tikzstyle{action}=[dot,fill=white,scale=0.667,inner sep=0.5mm]
\tikzstyle{copier}=[dot,fill=white,scale=2.0]
\tikzstyle{black copier}=[dot,fill=black,scale=2.0]

\tikzstyle{box}=[fill=white, draw=black, shape=rectangle]
\tikzstyle{medium box}=[fill=white, draw=black, shape=rectangle, minimum width=1.5cm, minimum height=0.66cm]
\tikzstyle{arrow box}=[fill=white, draw, shape=rectangle,minimum height=5mm,yshift=-0.5mm,minimum width=5mm]

\tikzstyle{effect}=[regular polygon, regular polygon sides=3,draw]
\tikzstyle{state0}=[regular polygon, regular polygon sides=3,draw,shape border rotate=0]
\tikzstyle{state90}=[regular polygon, regular polygon sides=3,draw,shape border rotate=90]
\tikzstyle{state180}=[regular polygon, regular polygon sides=3,draw,shape border rotate=180]
\tikzstyle{state270}=[regular polygon, regular polygon sides=3,draw,shape border rotate=270]

\tikzstyle{scalar}=[diamond,draw,inner sep=1pt]

\tikzstyle{discarder}=[my ground,draw,inner sep=0pt,minimum width=4.2pt,minimum height=11.2pt,anchor=input,rotate=90]
\tikzstyle{discarder0}=[my ground,draw,inner sep=0pt,minimum width=4.2pt,minimum height=11.2pt,anchor=input,rotate=0]

\tikzstyle{pointy1}=[->]
\tikzstyle{midpoint1}=[-, {postaction={decorate,decoration={markings, mark=at position .5 with {\arrow{>}}}}}]
\tikzstyle{midpointy1pointy}=[->, {postaction={decorate,decoration={markings, mark=at position .5 with {\arrow{>}}}}}]
\tikzstyle{dashed1}=[-, dashed]
\tikzstyle{dotted1}=[-, dotted]
\tikzstyle{dash-pointy}=[->, dashed]

\input{strings.tikzdefs}

\ifexternalizetikz\tikzexternaldisable\fi
\newsavebox\sbground
\savebox\sbground{%
  \begin{tikzpicture}[baseline=0pt]
    \draw (0,-.1ex) to (0,.85ex)
    node[ground IEC,draw,anchor=input,inner sep=0pt,
    minimum width=3.15pt,minimum height=8.4pt,rotate=90] {};
  \end{tikzpicture}%
}

\newsavebox\sbcopier
\savebox\sbcopier{%
  \begin{tikzpicture}[baseline=0pt]
    \node[copier,scale=0.7] (a) at (0,3.8pt) {};
    \draw (a) -- +(-90:.21);
    \draw (a) -- +(45:.21);
    \draw (a) -- +(135:.21);
  \end{tikzpicture}}

\ifexternalizetikz\tikzexternalenable\fi

\newsavebox\bsbcopier
\savebox\bsbcopier{%
  \begin{tikzpicture}[baseline=0pt]
    \node[black copier,scale=0.7] (a) at (0,3.8pt) {};
    \draw (a) -- +(-90:.21);
    \draw (a) -- +(45:.21);
    \draw (a) -- +(135:.21);
  \end{tikzpicture}}
\newcommand{\bcopier}{\mathord{\usebox\bsbcopier}}
\ifexternalizetikz\tikzexternalenable\fi

\usepackage{mathabx}

\usepackage{quiver}

\newcommand{\deloop}[1]{\mathbf{B}#1}

\newcommand{\Sum}{\sum\limits}

\def\cctx{\mathbb{C}\Fun{tx}}

\newcommand{\efb}[3]{\llparenthesis\,{#1}\,|\,{#2}\,|\,{#3}\,\rrparenthesis}
\newcommand{\blhom}[1]{\ldbrack{#1}\rdbrack}

\newcommand{\BLens}[1]{\Cat{BayesLens}_{\cat{#1}}} %
\newcommand{\PSGame}[1]{\Cat{PSGame}_{#1}}

\def\PSpc{\Pa\Cat{\mdash Spc}}
\def\PFdCartSpc{\Pa\Cat{\mdash FdCartSpc}}
\def\FdCartSpc{\Cat{FdCartSpc}}

\def\Poly{\Cat{Poly}}

\newcommand{\pMCoalgT}[3]{\Cat{Coalg}_{#2}^{#3}({#1})}

\newcommand{\PMCoalgT}[1]{\pMCoalgT{p}{\cat{C}}{#1}}
\newcommand{\MCoalgT}[1]{\pMCoalgT{-}{\cat{C}}{#1}}
\newcommand{\Coalg}{\Cat{Coalg}}

\makeatletter
\newcommand{\Hier}{\Cat{Hier}}
\newcommand{\HierE@nostar}{{\Cat{Hier}|_{\cat{E}}}}
\newcommand{\HierE@star}[1]{{\Cat{Hier}_{#1}|_{\cat{E}}}}
\newcommand{\HierE}{\@ifstar{\HierE@star}{\HierE@nostar}}
\makeatother

\def\ltri{\triangleleft}

\makeatletter
\DeclareRobustCommand\Equiv{\mathrel{%
  \mathchoice
    {\Equiv@\textfont\displaystyle{.43}}
    {\Equiv@\textfont\textstyle{.43}}
    {\Equiv@\scriptfont\scriptstyle{.45}}
    {\Equiv@\scriptscriptfont\scriptscriptstyle{.5}}
}}
\newcommand{\Equiv@}[3]{%
  \rlap{\raisebox{#3\fontdimen5#12}{$\m@th#2 = $}}%
  \raisebox{-#3\fontdimen5#12}{$\m@th#2 = $}%
}
\makeatother

\def\Comon{\Cat{Comon}}

\def\gauss{\Cat{Gauss}}

\def\DiffSys{\Cat{DiffSys}}
\def\DiffHier{\Cat{DiffHier}}

\def\CartDiffSys{\Cat{CartDiffSys}}
\def\CartDiffHier{\Cat{CartDiffHier}}

\def\H{\Cat{H}}

\def\eP{\Cat{P}}

\def\Laplace{\Fun{L}}
\def\Hebb{\Fun{H}}

\date{\today}
\title{Compositional Active Inference II: \\
  \LARGE Polynomial Dynamics. Approximate Inference Doctrines.}

\begin{document}

\maketitle

\begin{abstract}
  We develop the compositional theory of active inference by introducing \textit{activity}, functorially relating statistical games to the dynamical systems which play them, using the new notion of approximate inference doctrine.
  In order to exhibit such functors, we first develop the necessary theory of dynamical systems, using a generalization of the language of polynomial functors to supply compositional interfaces of the required types: with the resulting polynomially indexed categories of coalgebras, we construct monoidal bicategories of differential and dynamical ``hierarchical inference systems'', in which approximate inference doctrines have semantics.
  We then describe ``externally parameterized'' statistical games, and use them to construct two approximate inference doctrines found in the computational neuroscience literature, which we call the `Laplace' and the `Hebb-Laplace' doctrines: the former produces dynamical systems which optimize the posteriors of Gaussian models; and the latter produces systems which additionally optimize the parameters (or `weights') which determine their predictions.
\end{abstract}

\section{Introduction}

In the first paper in this series \parencite{Smithe2021Compositional1}, we introduced a compositional framework in which to make sense of the `statistical games' played by adaptive and cybernetic systems, with a view to generalizing and contextualizing the free energy principle that lies at the heart of theories of active inference \parencite{Parr2022Active}.
Yet, these statistical games are but one aspect of an active adaptive system, and if a theory of active inference is to be a theory of anything, then it must also acknowledge \textit{activity}!
As a starting point, the framework of statistical games accounts for systems that are open to their environment, and whose predictive performance is accordingly contextual, but the next step --- and the step taken in this paper --- is to animate these statistical games, constructing dynamical systems that play these games, and that can be correspondingly embodied in a changing world.
The behaviours of these model systems can then be compared with observations of natural adaptive systems, and the models can then be refined accordingly.

It is a remarkable fact that our most infamous natural adaptive system, the mammalian brain, seems in part to exemplify the hierarchical bidirectional structure of statistical games:
certain neural circuits in sensory cortex exhibit forward-looking predictions alongside backward-looking corrections that together can be modelled as a kind of dynamical Bayesian inference process, and which appear to couple together to approximate hierarchically structured Bayesian networks \parencite{Bastos2012Canonical}.
Understanding this resemblance is one of the principal motivations for this work.

Since the brain is best understood as an `open' (\textit{i.e.}, embodied and interacting) dynamical system, this resemblance seems to imply a functorial relationship between a category of statistical models on the one hand and a category of open dynamical systems on the other:
the functor would take an appropriately defined statistical model or statistical game, and return a dynamical system that could be understood as playing the game (or inverting the model); the functoriality of this relationship would ensure that the compositional (including hierarchical) structure of the model would be recapitulated in the compositional structure of the resulting dynamical system.

Exhibiting functors of this type, which collectively we call \textit{approximate inference doctrines}, is the task of Section \ref{sec:doctrines}, and indeed we find that the aforementioned neural circuit models arise precisely in this way.
Not only does this explain the mathematical origin of the structure of these circuits, but it simplifies the job of modelling, as one no longer needs to perform a complicated computation for each model:
instead, it is sufficient to obtain the dynamics for each factor of the model, and compose them according to the rules of the category.
(In this paper, we focus on functors \textit{from} statistical models \textit{to} dynamical systems.
One claim of the free energy framework is that it furnishes a universal way to understand adaptive dynamical systems in terms of Bayesian inference \parencite{Friston2019free}, suggesting functors in the opposite direction which we might hypothesize to be appropriately adjoint.
Understanding this relationship is the subject of future work.)

\paragraph{Overview of this paper}

Before we can exhibit any such functors, we need to lay the appropriate mathematical groundwork.
For our purposes, there are two overlapping aspects: a mathematical language in which to talk about stochastic interacting systems; and a definition of open dynamical system that can be expressed in this language and that can be cast into the the relevant compositional form.

In \secref{sec:poly} therefore, we introduce the category of polynomial functors as our choice of language for interaction.
We think of a polynomial as playing a formal role akin to that of the notion of Markov blanket in the informal active inference literature, as it defines the shape or boundary or interface of a type of system; morphisms of polynomials describe how information flows between the boundaries of coupled systems.
In \secref{sec:poly-stoch}, we generalize the usual category of polynomials in order to capture stochastic interactions and the flow of probabilistic information.

Then, in \secref{sec:dyn}, we turn our attention to dynamics.
We begin the section by defining a general notion of dynamical system on an interface using the language of polynomials.
We then package these systems up into categories indexed by polynomials: each category represents a collection of ways that an interface may be animated.
Subsequently, in \secref{sec:dyn-hier}, we bring these categories together with the category of polynomials itself to construct a new collection of categories of hierarchical bidirectional dynamical systems which have the necessary compositional structure to define approximate inference doctrines; then, in \secref{sec:dyn-diff}, we present corresponding categories of \textit{differential} systems, which often form a useful intermediate step on the way to dynamical systems, and show how to obtain dynamical systems from them.

Finally, in \secref{sec:doctrines}, we introduce approximate inference doctrines, concentrating on two that are neuroscientifically relevant.
We begin the section by introducing two pieces of auxiliary technology: categories of Gaussian channels (\secref{sec:doctrines-gauss}, to capture the two neuroscientific doctrines); and parameterized statistical games (\secref{sec:doctrines-para}, to capture parameter learning like synaptic plasticity).
This puts us in the position at last to define two doctrines: the Laplace doctrine (\secref{sec:doctrines-laplace}) for Gaussian channels; and the Hebb-Laplace doctrine (\secref{sec:doctrines-hebb-laplace}) for parameterized Gaussian channels, where not only is the model inverted but the parameters are learnt, too.

\section{Polynomial functors: a language for interacting systems} \label{sec:poly}

In order to be considered \textit{adaptive}, a system must have something to adapt to.
This `something' is often what we call the system's \textit{environment}, and we say that the system is \textit{open} to its environment.
The interface or boundary separating the system from its environment can be thought of as `inhabited' by the system: the system is embodied by its interface of interaction; the interface is animated by the system.
In this way, the system can affect the environment, by changing the shape or configuration of its interface\footnote{
Such changes can be very general: consider for instance the changes involved in producing sound (\textit{e.g.}, rapid vibration of tissue) or light (\textit{e.g.}, connecting a luminescent circuit, or the molecular interactions involved therein).
}; through the coupling, these changes are propagated to the environment.
In turn, the environment may impinge on the interface: its own changes, mediated by the coupling, arrive at the interface as immanent signals; and the type of signals to which the system is alive may depend on the system's configuration (as when an eye can only perceive if its lid is open).
Thus, information flows across the interface.

The mathematical language capturing this kind of inhabited interaction is that of \textit{polynomial functors}, which we adopt following \textcite{Spivak2021Polynomial}.
Informally, a polynomial functor is determined by a type or set of possible configurations, along with, for each possible configuration, a corresponding type or set of possible immanent signals (`inputs').
We will often write $p$ to denote a polynomial, $p(1)$ its possible configurations, and for each $i:p(1)$, $p[i]$ for the corresponding inputs.

In this section, we introduce the basic theory of polynomial functors; in the following subsection, we extend the theory to allow for more general kinds of interaction, to allow for explicitly probabilistic information flows.
Taking a broader view, in this paper we only make use of a fragment of the richness of polynomial interaction: just enough to build open and hierarchical dynamical systems that can perform inference within a single system.
Later in this series, we will expand our use of the language to treat multiple interacting active inference systems, to provide something like a theory of ``polynomial life'', building on our earlier work \parencite{Smithe2021Polynomial}.
Now, however, we begin by introducing the formal definition of the classical category of polynomial functors.

\begin{defn}
  Let $\cat{E}$ be a locally Cartesian closed category (such as $\Cat{Set}$), and denote by $y^A$ the representable copresheaf $y^A := \cat{E}(A, -) : \cat{E} \to \cat{E}$.
  A \emph{polynomial functor} $p$ is a coproduct of representable functors, written $p := \sum_{i : p(1)} y^{p_i}$, where $p(1) : \cat{E}$ is the indexing object.
  The category of polynomial functors in $\cat{E}$ is the full subcategory $\Poly_\Ea \hookrightarrow [\cat{E}, \cat{E}]$ of the \(\cat{E}\)-copresheaf category spanned by coproducts of representables.
  A morphism of polynomials is therefore a natural transformation.
\end{defn}

\begin{rmk} \label{rmk:poly-bundles} %
  Every polynomial functor $P : \cat{E} \to \cat{E}$ corresponds to a bundle $p : E \to B$ in $\cat{E}$, for which $B = P(1)$ and for each $i : P(1)$, the fibre $p_i$ is $P(i)$.
  We will henceforth elide the distinction between a copresheaf $P$ and its corresponding bundle $p$, writing $p(1) := B$ and $p[i] := p_i$, where $E = \sum_i p[i]$.
  A natural transformation $f : p \to q$ between copresheaves therefore corresponds to a map of bundles.
  In the case of polynomials, by the Yoneda lemma, this map is given by a `forwards' map $f_1 : p(1) \to q(1)$ and a family of `backwards' maps $f^\# : q[f_1(\mdash)] \to p[\mdash]$ indexed by $p(1)$, as in the left diagram below.
  Given $f : p \to q$ and $g : q \to r$, their composite $g \circ f : p \to r$ is as in the right diagram below.
  \begin{equation*}
    \begin{tikzcd}
      E & {f^*F} & F \\
      B & B & C
      \arrow["{f^\#}"', from=1-2, to=1-1]
      \arrow[from=1-2, to=1-3]
      \arrow["q", from=1-3, to=2-3]
      \arrow["p"', from=1-1, to=2-1]
      \arrow[from=2-1, to=2-2, Rightarrow, no head]
      \arrow["{f_1}", from=2-2, to=2-3]
      \arrow[from=1-2, to=2-2]
      \arrow["\lrcorner"{anchor=center, pos=0.125}, draw=none, from=1-2, to=2-3]
    \end{tikzcd}
    \qquad\qquad
    \begin{tikzcd}
      E & {f^*g^*G} & G \\
      B & B & D
      \arrow["{(gf)^\#}"', from=1-2, to=1-1]
      \arrow[from=1-2, to=1-3]
      \arrow["r", from=1-3, to=2-3]
      \arrow["p"', from=1-1, to=2-1]
      \arrow[from=2-2, to=2-1, Rightarrow, no head]
      \arrow["{g_1 \circ f_1}", from=2-2, to=2-3]
      \arrow[from=1-2, to=2-2]
      \arrow["\lrcorner"{anchor=center, pos=0.125}, draw=none, from=1-2, to=2-3]
    \end{tikzcd}
  \end{equation*}
  where $(gf)^\#$ is given by the \(p(1)\)-indexed family of composite maps
  $r[g_1(f_1(\mdash))] \xto{f^\ast g^\#} q[f_1(\mdash)] \xto{f^\#} p[\mdash]$.
\end{rmk}

We now recall a handful of useful facts about polynomials and their morphisms, each of which is explained in \textcite{Spivak2021Polynomial} and summarized in \textcite{Spivak2022reference}.

\begin{prop}
  Polynomial morphisms $p \to y$ correspond to sections $p(1) \to \sum_i p[i]$ of the corresponding bundle $p$.
\end{prop}

\begin{prop}
  There is an embedding of \(\cat{E}\) into \(\Poly_\Ea\) given by taking objects \(X : \cat{E}\) to the linear polynomials \(Xy : \Poly_\Ea\) and morphisms \(f : X\to Y\) to morphisms \((f,\id_X) : Xy \to Yy\).
\end{prop}

\begin{prop} \label{prop:poly-tensor}
  There is a symmetric monoidal structure \((\otimes, y)\) on \(\Poly_\Ea\) that we call \textnormal{tensor}, and which is given on objects by \(p\otimes q := \sum_{i:p(1)}\sum_{j:q(1)} y^{p[i]\times q[j]}\) and on morphisms \(f := (f_1, f^\#) : p\to p'\) and \(g := (g_1, g^\#) : q\to q'\) by \(f\otimes g := (f_1 \times g_1, f^\#\times g^\#)\).
\end{prop}

\begin{prop}
  $(\Poly_\Ea, \otimes, y)$ is symmetric monoidal closed, with internal hom denoted $[{-},{=}]$.
  Explicitly, we have $[p,q] = \sum_{f:p\to q} y^{\sum_{i:p(1)} q[f_1(i)]}$.
  Given an object $A : \cat{E}$, we have $[Ay, y] \cong y^A$.
\end{prop}

\begin{prop}
  The composition of polynomial functors $q \circ p : \cat{E}\to\cat{E}\to\cat{E}$ induces a monoidal structure on $\Poly_\Ea$, which we denote $\ltri$, and call `composition' or `substitution'.
  Its unit is again $y$.
  Famously, $\ltri$-comonoids correspond to categories and their comonoid homomorphisms are cofunctors \parencite{Ahman2016Directed}.
  If $\Tt$ is a monoid, then the comonoid structure on $y^\Tt$ corresponds witnesses it as the category $\deloop{\Tt}$.
  Monomials of the form $Sy^S$ can be equipped with a canonical comonoid structure witnessing the codiscrete groupoid on $S$.
\end{prop}

\subsection{Generalized polynomials for stochastic feedback} \label{sec:poly-stoch}

The category of polynomial functors $\Poly_{\cat{E}}$ introduced above for a locally Cartesian closed category $\cat{E}$ can be considered as a category of `deterministic' polynomial interaction; notably, morphisms of polynomials, which encode the coupling of systems' interfaces, do not explicitly incorporate any kind of randomness or uncertainty.
Even if the universe is deterministic, however, the finiteness of systems and their general inability to perceive the totality of their environments make it a convenient modelling choice to suppose that systems' interactions may be uncertain; this will be useful not only in allowing for stochastic interactions between systems, but also to define stochastic dynamical systems `internally' to a category of polynomials.

To reach the desired generalization, we begin by recalling that $\Poly_{\cat{E}}$ is equivalent to the category of Grothendieck lenses for the self-indexing of $\cat{E}$ \parencite{Spivak2020Poly,Spivak2021Polynomial}: $\Poly_{\cat{E}} \cong \int \cat{E}/-\op$, where the opposite is taken pointwise on each $\cat{E}/B$; this is the formal basis for Remark \ref{rmk:poly-bundles}.
We define our categories of generalized polynomials from this perspective, by considering categories indexed by their ``deterministic subcategories'': this allows us to define categories of Grothendieck lenses which behave like $\Poly_{\cat{E}}$ (when restricted to the deterministic case), but also admit uncertain inputs.

\begin{notation}
  Suppose $\cat{C}$ is a symmetric monoidal category.
  We write $\Comon(\cat{C})$ to denote the subcategory of commutative comonoids and comonoid homomophisms in $\cat{C}$.
\end{notation}

\begin{ex}
  Suppose $\Pa:\cat{E}\to\cat{E}$ is a probability monad\footnote{
  By `probability monad', we mean a monad $\Pa$ on $\cat{E}$ taking each object $X$ to an object $\Pa X$ that behaves like a `space of probability distributions on $X$'.
  The monad multiplication performs a `weighted average' of distributions, and the monad unit returns the point or `Dirac delta' distribution on each element.
  For more information on and a number of examples of probability monads, we refer the reader to \textcite{Jacobs2018probability}.
  We will often write $\Pa$ to denote a generic probability monad.
  } on $\cat{E}$.
  Then every object in $\Kl(\Pa)$ is equipped with a canonical comonoid structure (the copy-discard structure \parencite[{\S 2}]{Cho2017Disintegration}), and $\Cat{Comon}\bigl(\Kl(\Pa)\bigr)$ is the wide subcategory of `deterministic' channels.
  Intuitively, this follows almost by definition: a deterministic process is one that has no informational side-effects; that is to say, whether we copy a state before performing the process on each copy, or perform the process and then copy the resulting state, or whether we perform the process and then marginalize, or just marginalize, makes no difference to the resulting state.
  This is just what it means for the process to be a comonoid homomorphism; in other words, deterministic processes introduce no new correlations.
  In fact, $\Cat{Comon}\bigl(\Kl(\Pa)\bigr) \cong \cat{E}$.
\end{ex}

With these ideas in mind, we make the following definitions.

\begin{defn} \label{def:poly-c-idx}
  Suppose $(\cat{C},\otimes,I)$ is a copy-delete category such that $\Comon(\cat{C})$ is finitely complete and $I$ is terminal in $\Comon(\cat{C})$.
  Define an indexed category $\Fun{P} : \Comon(\cat{C})\op\to\Cat{Cat}$ as follows.
  For each object $B:\Comon(\cat{C})$, the category $\Fun{P}(B)$ has as objects the homomorphisms $E\to B$ of $\Comon(\cat{C})$ such that for any other homomorphism $A\to B$, the pullback $A\times_B E$ satisfies the universal property in $\cat{C}$.
  Given a morphism $f : C\to B$, the functor $\Fun{P}(f) : \Fun{P}(B)\to \Fun{P}(C)$ is given by pullback: $\Fun{P}(f) := f^*$; this is well-defined by the universal property.
\end{defn}

\begin{defn} \label{def:poly-c}
  Suppose each functor $\Fun{P}(f) : \Fun{P}(B)\to \Fun{P}(C)$ has a left adjoint, denoted $\Sigma_f$.
  We define the category $\Poly_{\cat{C}}$ of polynomials in $\cat{C}$ to be the category of $\Fun{P}$-lenses: $\Poly_{\cat{C}} := \int \Fun{P}\op$, where the opposite is taken pointwise.
\end{defn}

\begin{ex}
  When $\cat{C}$ is any locally Cartesian closed category such as $\Cat{Set}$, equipped with its Cartesian monoidal structure, Definition \ref{def:poly-c-idx} recovers its self-indexing and hence $\Poly_{\cat{C}}$ is the usual category of polynomials in $\cat{C}$.
\end{ex}

\begin{ex} \label{ex:poly-m-func}
  Suppose $\cat{E}$ is a finitely complete category and $M$ is a monoidal monad on $\cat{E}$.
  Denote by $\iota$ the identity-on-objects inclusion $\cat{E} \hookrightarrow \Kl(M)$ given on morphisms by post-composing with the unit $\eta$ of the monad structure.
  Setting $\cat{C} = \Kl(M)$, we find that for $B : \cat{E}$, $\Fun{P}(B)$ is the full subcategory of $\Kl(M)/B$ on those objects $\iota p : E\klto B$ which correspond to maps $E \xto{p} B \xto{\eta_B} MB$ in the image of $\iota$.
  Given a morphism $f : C\to B$ in $\cat{E}$, the functor $\Fun{P}(f)$ takes objects $\iota p : E\klto B$ to $\iota(f^*p) : f^*E\klto C$ where $f^*p$ is the pullback of $p$ along $f$ in $\cat{E}$, included into $\Kl(M)$ by $\iota$.
  Now suppose that $\alpha$ is a morphism $(E, \iota p:E\klto B) \to (F, \iota q:F\klto B)$ in $\Fun{P}(B)$, and note that since we must have $\iota q\klcirc \alpha = \iota p$, $\alpha$ must correspond to a family of maps $\alpha_x : p[x] \to Mq[x]$ for $x : B$.
  Therefore, $\Fun{P}(f)(\alpha)$ can be defined pointwise as $\Fun{P}(f)(\alpha)_y := \alpha_{f(y)} : p[f(y)] \to Mq[f(y)]$ for $y : C$.
\end{ex}

\begin{notation} \label{not:poly-c}
  For any such monoidal monad $M$ where $\cat{E}$ has dependent sums, we will write $\Poly_M$ as shorthand denoting the corresponding generalized category of polynomials $\Poly_{\Kl(M)}$.
  Since every category $\cat{C}$ corresponds to a trivial monad which we can also denote by $\cat{C}$, this notation subsumes that of Definition \ref{def:poly-c}.
\end{notation}

\begin{rmk}
  We can think of $\Poly_M$ as a dependent version of the category of $M$-monadic lenses, in the sense of \textcite[{\S 3.1.3}]{Clarke2020Profunctor}.
\end{rmk}

Unwinding Example \ref{ex:poly-m-func} further, we find that the objects of $\Poly_M$ are the same polynomial functors as constitute the objects of $\Poly_{\cat{E}}$.
The morphisms $f : p \to q$ are pairs $(f_1, f^\#)$, where $f_1 : B\to C$ is a map in $\cat{E}$ and
$f^\#$ is a family of morphisms $q[f_1(x)]\klto p[x]$ in $\Kl(M)$, making the following diagram commute:

\begin{equation*}
  \begin{tikzcd}
    \sum_{x:B} Mp[x] & \sum_{b:B} q[f_1(x)] & \sum_{y:C} q[y] \\
    B & B & C
    \arrow["{f^\#}"', from=1-2, to=1-1]
    \arrow[from=1-2, to=1-3]
    \arrow["q", from=1-3, to=2-3]
    \arrow["{\eta_B}^* p"', from=1-1, to=2-1]
    \arrow[from=2-1, to=2-2, Rightarrow, no head]
    \arrow["{f_1}", from=2-2, to=2-3]
    \arrow[from=1-2, to=2-2]
    \arrow["\lrcorner"{anchor=center, pos=0.125}, draw=none, from=1-2, to=2-3]
  \end{tikzcd}
\end{equation*}

Our principal example of interest is of this form, being $\Poly_{\Pa}$ for a probability monad $\Pa$ on $\cat{E}$\footnote{
Ideally, $\cat{E}$ would also be locally Cartesian closed, so that $\Poly_{\Pa}$ recapitulates much of the basic structure of $\Poly_{\Cat{Set}}$ (see Remark \ref{rmk:poly-m-pi-types}): such examples include the category $\Cat{QBS}$ of quasi-Borel spaces equipped with the quasi-Borel distribution monad \parencite{Heunen2017Convenient}, or the category $\Cat{Set}$ equipped with the finitely-supported distribution monad.}.
We we consider each such category $\Poly_{\Pa}$ to be a category of \textit{polynomials with stochastic feedback}.

\begin{rmk}
  By assuming that the category $\cat{C}$ has a monoidal structure $(\otimes, I)$, its corresponding generalized category of polynomials $\Poly_{\cat{C}}$ inherits a tensor akin to that defined in Proposition \ref{prop:poly-tensor}, and which we also denote by $(\otimes, I)$:
  the definition only differs by substituting the structure $(\otimes, I)$ on $\cat{C}$ for the product $(\times, 1)$ on $\cat{E}$.
  This follows from the monoidal Grothendieck construction: $\Fun{P}$ is lax monoidal, with laxator taking $p : \Fun{P}(B)$ and $q : \Fun{P}(C)$ to $p\otimes q : \Fun{P}(B\otimes C)$.

  On the other hand, for $\Poly_{\cat{C}}$ also to have an internal hom $[q,r]$ requires each fibre of $\Fun{P}$ to be closed with respect to the monoidal structure.
  In cases of particular interest, $\Comon(\cat{C})$ will be locally Cartesian closed, and restricting $\Fun{P}$ to its self-indexing produces fibres which are thus Cartesian monoidal closed.
  In these cases, we can think of the broader fibres of $\Fun{P}$, and thus $\Poly_{\cat{C}}$ itself, as being `deterministically' closed.
  This means, for the stochastic example $\Poly_\Pa$, we get an internal hom satisfying the adjunction $\Poly_\Pa(p\otimes q, r) \cong \Poly_\Pa(p, [q,r])$ only when the backwards components of morphisms $p\otimes q \to r$ are `uncorrelated' between $p$ and $q$.
\end{rmk}

\begin{rmk} \label{rmk:poly-m-pi-types}
  For $\Poly_{\cat{C}}$ to behave faithfully like the usual category of polynomial functors, we should want the substitution functors $\Fun{P}(f) : \Fun{P}(C) \to \Fun{P}(B)$ to have right adjoints as well as left.
  As in the preceding remark, these only obtain in restricted circumstances; we will consider the case of $\Poly_M$ for a monad $M$, writing $f^*$ to denote the functor $\Fun{P}(f)$.

  Denote the putative right adjoint by $\Pi_f : \Fun{P}(B) \to \Fun{P}(C)$, and for $\iota p : E\klto B$ suppose that $(\Pi_f E)[y]$ is given by the set of `partial sections' $\sigma : f^{-1}\{y\} \to ME$ of $p$ over $f^{-1}\{y\}$ as in the commutative diagram:
  \[\begin{tikzcd}
	& {f^{-1}\{y\}} & {\{y\}} \\
	ME & B & C
	\arrow[from=1-2, to=1-3]
	\arrow[from=1-3, to=2-3]
	\arrow[from=1-2, to=2-2]
	\arrow["f", from=2-2, to=2-3]
	\arrow["\lrcorner"{anchor=center, pos=0.125}, draw=none, from=1-2, to=2-3]
	\arrow["{{\eta_B}^*p}", from=2-1, to=2-2]
	\arrow["\sigma"', curve={height=12pt}, from=1-2, to=2-1]
  \end{tikzcd}\]
  Then we would need to exhibit a natural isomorphism $\Fun{P}(B)(f^*D, E) \cong \Fun{P}(C)(D, \Pi_f E)$.
  But this will only obtain when the `backwards'  components $h^\#_y : D[y]\to M(\Pi_f E)[y]$ are in the image of $\iota$---otherwise, it is not generally possible to pull $f^{-1}\{y\}$ out of $M$.
\end{rmk}

\section{Open dynamical systems on polynomial interfaces} \label{sec:dyn}

Having constructed $\Poly_{\cat{C}}$, we are now in a position to construct, for each $p : \Poly_{\cat{C}}$, a category of open dynamical systems $\Coalg^\Tt_{\cat{C}}(p)$ with interface $p$, and we can even state the definition entirely in the language of $\Poly_{\cat{C}}$.
Here, $\Tt$ is a monoid object $(\Tt, +, 0)$ in $\Comon(\cat{C})$ that represents time, which is necessary in general to ensure that the dynamics can `flow' appropriately; slightly more formally, we will need to ensure that evolving the dynamics for time $t:\Tt$ and then $s:\Tt$ produces the same trajectory as evolving it for time $t+s$, and that evolving it for no time $0:\Tt$ induces no change.
If we choose $\cat{C} = \Kl(\Pa)$ for $\Pa$ a probability monad, we obtain categories of stochastic systems that we call \textit{open Markov processes}, although we develop the theory in a more general context (allowing for other types of transition, as as nondeterministic).

We first give a concise definition, internal to $\Poly_{\cat{C}}$, before unpacking it into a more elementary form.

\begin{defn} \label{def:poly-dyn}
  An open dynamical system with interface $p : \Poly_{\cat{C}}$, state space $S : \cat{C}$ and time $(\Tt, +, 0)$ is a polynomial morphism $\beta : Sy^S \to [\Tt y, p]$ such that, for any section $\sigma : p \to y$, the induced morphism
  \begin{gather*}
    Sy^S \xto{\beta} [\Tt y, p] \xto{[\Tt y, \sigma]} [\Tt y, y] \xto{\sim} y^\Tt
  \end{gather*}
  is a $\ltri$-comonoid homomorphism.
\end{defn}

Unpacking this definition gives us the following characterization:

\begin{prop} \label{prop:pM-coalg}
  An open dynamical system $\beta : Sy^S \to [\Tt y, p]$ in $\Poly_{\cat{C}}$ consists in a triple \((S, \beta^o, \beta^u)\) of a state space \(S : \cat{C}\) and two morphisms \(\beta^o : \Tt \times S \to p(1)\) in $\Comon(\cat{C})$ and \(\beta^u : \sum_{t:\Tt} \sum_{s:\mathbb{S}} p[\vartheta^o(t,s)] \to S\) in $\cat{C}$, such that, for any section \(\sigma : p(1) \to \sum_{i:p(1)} p[i]\) of \(p\), the morphisms \(\beta^\sigma : \Tt \times S \to S\) given by
  \[
  \Sum_{t:\Tt} S \xto{\beta^o(t)^\ast \sigma} \Sum_{t:\Tt} \Sum_{s:S} p[\beta^o(t, s)] \xto{\beta^u} S
  \]
  form an object in the functor category \(\Cat{Cat}\big(\deloop{\Tt}, \cat{C}\big)\), where \(\deloop{\Tt}\) is the delooping of \(\Tt\).
  We call the closed system \(\beta^\sigma\), induced by a section \(\sigma\) of \(p\), the   \textnormal{closure} of \(\beta\) by \(\sigma\).
  Equivalently, we can say that $\beta^\sigma : \Tt \times S \to S$ forms an action of the monoid $\Tt$ on $S$ in $\cat{C}$.
\end{prop}

Open dynamical systems on $p$ form a category, which we denote by $\Coalg_{\cat{C}}^\Tt(p)$.
We can exhibit this category abstractly, by noting that a morphism $Sy^S \to r$ of polynomials is equivalent to a morphism $S \to r(S)$ in $\cat{C}$: that is, to an $r$-coalgebra; morphisms of open dynamical systems then correspond to coalgebra homomorphisms, and this gives us a category.
For our purposes here, however, it is more illuminating to exhibit $\Coalg_{\cat{C}}^\Tt(p)$ explicitly.

\begin{prop} \label{prop:pM-coalg-cat}
  Open dynamical systems on $p$ with time \(\Tt\) form a category, denoted $\Coalg_{\cat{C}}^\Tt$.
  Its morphisms are defined as follows.
  Let \(\vartheta := (X, \vartheta^o, \vartheta^u)\) and \(\psi := (Y, \psi^o, \psi^u)\) be two such systems.
  A morphism \(f : \vartheta \to \psi\) consists in a morphism \(f : X \to Y\) in $\cat{C}$ such that, for any time \(t : \Tt\) and global section \(\sigma : p(1) \to \Sum_{i:p(1)} p[i]\) of \(p\), the following square commutes:
  \[\begin{tikzcd}
	X & {\Sum_{x:X} p[\vartheta^o(t, x)]} & X \\
	Y & {\Sum_{y:Y} p[\psi^o(t, y)]} & Y
	\arrow["{\vartheta^o(t)^\ast \sigma}", from=1-1, to=1-2]
	\arrow["{\vartheta^u(t)}", from=1-2, to=1-3]
	\arrow["f"', from=1-1, to=2-1]
	\arrow["f", from=1-3, to=2-3]
	\arrow["{\psi^o(t)^\ast \sigma}"', from=2-1, to=2-2]
	\arrow["{\psi^u(t)}"', from=2-2, to=2-3]
  \end{tikzcd}\]
  The identity morphism \(\id_\vartheta\) on \(\vartheta\) is given by the identity morphism \(\id_X\) on its state space \(X\).
  Composition of morphisms is given by composition of the morphisms of the state spaces.
\end{prop}

Since open dynamical systems on $p$ are morphisms $Sy^S \to [\Tt y, p]$ of polynomials, there is a natural covariant reindexing of systems along morphisms $p\to q$, given by postcomposing with the map $[\Tt y, p] \to [\Tt y, q]$ induced by the functor $[\Tt y, {-}]$.
This gives \(\MCoalgT{\Tt}\) the structure of an opindexed category \(\Poly_{\cat{C}} \to \Cat{Cat}\), which we spell out in the following proposition.

\begin{prop} \label{prop:pM-coalg-idx}
  \(\PMCoalgT{\Tt}\) extends to an opindexed category, \(\MCoalgT{\Tt} : \Poly_{\cat{C}} \to \Cat{Cat}\).
  Suppose \(\varphi : p \to q\) is a morphism of polynomials.
  We define a corresponding functor
  \(\pMCoalgT{\varphi}{\cat{C}}{\Tt} : \pMCoalgT{p}{\cat{C}}{\Tt} \to \pMCoalgT{q}{\cat{C}}{\Tt}\)
  as follows.
  Suppose \((X, \vartheta^o, \vartheta^u) : \pMCoalgT{p}{\cat{C}}{\Tt}\) is an object (system) in \(\pMCoalgT{p}{\cat{C}}{\Tt}\).
  Then \(\pMCoalgT{\varphi}{\cat{C}}{\Tt}(X, \vartheta^o, \vartheta^u)\) is defined as the triple \((X, \varphi_1 \circ \vartheta^o, \vartheta^u \circ {\vartheta^o}^\ast \varphi^\#) : \pMCoalgT{q}{\cat{C}}{\Tt}\), where the two maps are explicitly the following composites:
  \begin{gather*}
    \Tt \times X \xto{\vartheta^o} p(1) \xto{\varphi_1} q(1) \, ,
    \qquad
    \Sum_{t:\Tt} \Sum_{x:X} q[\varphi_1 \circ \vartheta^o(t, x)] \xto{{\vartheta^o}^\ast \varphi^\#} \Sum_{t:\Tt} \Sum_{x:X} p[\vartheta^o(t, x)] \xto{\vartheta^u} X \, .
  \end{gather*}
  On morphisms, \(\pMCoalgT{\varphi}{\cat{C}}{\Tt}(f) : \pMCoalgT{\varphi}{\cat{C}}{\Tt}(X, \vartheta^o, \vartheta^u) \to \pMCoalgT{\varphi}{\cat{C}}{\Tt}(Y, \psi^o, \psi^u)\) is given by the same underlying map \(f : X \to Y\) of state spaces.
\end{prop}

It is sometimes useful to relate dynamical systems with different time monoids---for instance, to discretize a continuous-time system, or to adjust the timescale of evolution of a system---and for these purposes we have the following proposition.

\begin{prop} \label{prop:time-change}
  Any map $f:\Tt'\to\Tt$ of monoids induces an indexed functor $\Coalg^\Tt_{\cat{C}} \to \Coalg^{\Tt'}_{\cat{C}}$.
  \begin{proof}
    We first consider the induced functor $\Coalg^\Tt_{\cat{C}}(p) \to \Coalg^{\Tt'}_{\cat{C}}(p)$, which we denote by $\Delta_f^p$.
    Note that we have a morphism $[f y,p] : [\Tt y, p] \to [\Tt' y, p]$ of polynomials by substitution (precomposition).
    A system $\beta$ in $\Coalg^\Tt_{\cat{C}}$ is a morphism $Sy^S \to [\Tt y, p]$ for some $S$, and so we define $\Delta_f^p(\beta)$ to be $[f,p]\circ\beta : Sy^S \to [\Tt y, p] \to [\Tt' y, p]$.
    To see that this satisfies the monoid action axiom, consider that the closure $\Delta_f^p(\beta)^\sigma$ for any section $\sigma : p\to y$ is given by
    \[ \Sum_{t:\Tt'} S \xto{\beta^o(f(t))^\ast \sigma} \Sum_{t:\Tt'} \Sum_{s:S} p[\beta^o(f(t), s)] \xto{\beta^u} S \]
    which is an object in the functor category $\Cat{Cat}(\deloop{\Tt'},\cat{C})$ since $f$ is a monoid homomorphism.
    On morphisms of systems, the functor $\Delta_f^p$ acts trivially.

    To see that $\Delta_f$ collects into an indexed functor, consider that it is defined on each polynomial $p$ by the contravariant action $[f,p]$ of the internal hom $[-,=]$, and that the reindexing $\Coalg^\Tt(\varphi)$ for any morphism $\varphi$ of polynomials is similarly defined by the covariant action $[\Tt y,\varphi]$.
    By the bifunctoriality of $[-,=]$, we have $[\Tt' y,\varphi]\circ[f y,p] = [f y,\varphi] = [f y,q]\circ[\Tt y,\varphi]$, and so $\Coalg^{\Tt'}_{\cat{C}}(\varphi)\circ\Delta_f^p = \Delta_f^q\circ\Coalg^\Tt_{\cat{C}}$.
  \end{proof}
\end{prop}

\begin{cor} \label{cor:coalg-disc}
  For each $k : \rr$, the canonical inclusion $\iota_k : \nn \hookrightarrow \rr : i \mapsto ki$ induces a corresponding `discretization' indexed functor $\Fun{Disc}_k := \Delta_\iota : \Coalg^\rr_{\cat{C}} \to \Coalg^{\nn}_{\cat{C}}$.
\end{cor}

Using the tensor product \(\otimes\) of polynomials, we can put systems' interfaces ``in parallel'', and it will be useful to do the same for the systems themselves.
We can do this using the corresponding lax monoidal structure of $\MCoalgT{\Tt}$.

\begin{prop}
  $\MCoalgT{\Tt}$ is lax monoidal $(\Poly_{\cat{C}}, \otimes, y) \to (\Cat{Cat}, \times, 1)$.
  The components $\lambda_{p,q} : \pMCoalgT{p}{\cat{C}}{\Tt}\times\pMCoalgT{q}{\cat{C}}{\Tt} \to \pMCoalgT{p\otimes q}{\cat{C}}{\Tt}$ of the laxator natural transformation $\lambda$ are the functors defined as follows.
  On objects, given $\beta : Xy^X \to [\Tt y, p]$ over $p$ and $\gamma : Yy^Y \to [\Tt, q]$ over $q$, the system $\lambda_{p,q}(\beta,\gamma)$ is the system
  \[
  (X\otimes Y)y^{(X\otimes Y)} \xto{=} Xy^X\otimes Yy^Y \xto{\beta\otimes\gamma} [\Tt y, p]\otimes [\Tt, q] \xto{\upsilon_{p,q}} [\Tt y, p\otimes q]
  \]
  with state space $X\times Y$. The forwards component
  \[
  \upsilon_1 : \Comon(\cat{C})\bigl(\Tt, p(1)\bigr)\times\Comon(\cat{C})\bigl(\Tt, q(1)\bigr) \to \Comon(\cat{C})\bigl(\Tt, p(1)\times q(1)\bigr)
  \]
  of $\upsilon_{p,q}$ forms the product of two trajectories, taking $f : \Tt\to p(1)$ and $g : \Tt\to q(1)$ to
  \[ \upsilon_1(f,g) := \Tt \xto{\bcopier} \Tt\otimes\Tt \xto{f\otimes g} p(1)\otimes q(1) \, . \]
  The backwards components witness simultaneous inputs; in elementwise form, we have
  \begin{align*}
    \upsilon^\#_{f,g} : \sum_{t:\Tt} p[f(t)]\otimes q[g(t)] &\to \sum_{t,t':\Tt} p[f(t)]\otimes q[g(t')] \\
    (t,a,b) &\mapsto (t,t,a,b) \, .
  \end{align*}

  \noindent
  On morphisms $\varphi : \beta \to \beta'$ and $\psi : \gamma \to \gamma'$, $\lambda_{p,q}(\varphi,\psi) : \lambda_{p,q}(\beta,\gamma)\to\lambda_{p,q}(\beta',\gamma')$ is defined by taking the product of the underlying maps of state spaces $\varphi : X\to X'$ and $\psi : Y\to Y'$.
  We will overload the notation, writing $\beta\otimes\gamma$ in place of $\lambda_{p,q}(\beta,\gamma)$, and similarly $\varphi\otimes\psi$ on morphisms.

  \noindent
  Finally, the unitor $\epsilon : \Cat{1} \to \pMCoalgT{y}{\cat{C}}{\Tt}$ is the functor taking the unique object $\star$ in the terminal category $\Cat{1}$ to the (`closed') system $(1, !^o, !^u)$ over $y$ with trivial state space, trivial output map, and trivial update map.
  It sends the unique morphism $\id_\star$ in $\Cat{1}$ to the identity map on $1$.
  \begin{proof}[Proof sketch]
    Firstly, it is straightforward to check that the functors $\lambda_{p,q}$ and $\epsilon$ return well-defined systems and morphisms, and that they are themselves well-defined as functors.
    Next, we check that the functors $\lambda_{p,q}$ collect into a natural transformation.
    This follows almost immediately from the functoriality of $[\Tt y, {-}\otimes{=}] : \Poly_{\cat{C}}\times\Poly_{\cat{C}}\to\Poly_{\cat{C}}$.
    Finally, we check that the axioms of associativity and unitality are satisfied.
    This follows from the associativity and unitality of the monoidal structure $(\otimes, y)$ on $\Poly_{\cat{C}}$.
  \end{proof}
\end{prop}

Note that $\Coalg^\Tt_{\cat{C}}$ really is \textit{lax} monoidal---the laxators are not equivalences---since not all systems over the parallel interface $p\otimes q$ factor into a system over $p$ alongside a system over $q$.

\subsection{Monoidal bicategories of hierarchical inference systems} \label{sec:dyn-hier}

Whereas it is the morphisms (1-cells) of categories of lenses and statistical games that represent open systems, it is the objects (0-cells) of the opindexed categories \(\Coalg^\Tt_{\cat{C}}\)\footnote{or, more precisely, their corresponding opfibrations $\int\Coalg^\Tt_{\cat{C}}$} that play this role; in fact, the objects of $\Coalg^\Tt_{\cat{C}}$ each represent both an open system and its (polynomial) interface.
In order to supply dynamical semantics for statistical games---functors from categories of statistical games to categories of dynamical systems---we need to cleave the dynamical systems from their interfaces, making the interfaces into 0-cells and systems into 1-cells between them, thereby letting the systems' types and composition match those of the games.

To do this, we will associate to each pair of objects $(A,S)$ and $(B,T)$ of a category of Bayesian lenses\footnote{
We will assume that these lenses are non-dependent lenses, as in \textcite{Smithe2021Compositional1}.}
a polynomial $\blhom{Ay^S,By^T}$ whose configurations correspond to lenses and whose inputs correspond to the lenses' inputs.
The categories $\Coalg^\Tt_\Pa\bigl(\blhom{Ay^S,By^T}\bigr)$ will then form the hom-categories of bicategories of hierarchical inference systems, and it is in these bicategories that we will find our dynamical semantics.

\begin{defn} \label{def:blhom}
  Let $\BLens{C}$ be the category of (non-dependent) Bayesian lenses in $\cat{C}$, with $\cat{C}$ enriched in $\Comon(\cat{C})$.
  Then for any pair of objects $(A,S)$ and $(B,T)$ in $\BLens{C}$, we define a polynomial $\blhom{Ay^S,By^T}$ in $\Poly_{\cat{C}}$ by
  \[
  \blhom{Ay^S,By^T} := \sum_{l:\BLens{C}\bigl((A,S),(B,T)\bigr)} y^{\cat{C}(I,A)\times T} \, .
  \]
\end{defn}

\begin{rmk}
  We can think of $\blhom{Ay^S,By^T}$ as an `external hom' polynomial for $\BLens{C}$, playing a role analogous to the internal hom $[p,q]$ in $\Poly_{\cat{C}}$.
  Its `bipartite' structure---with domain and codomain parts---is what enables cleaving systems from their interfaces, which are given by these parts.
  The definition, and the following construction of the monoidal bicategory, are inspired by the operad $\Cat{Org}$ introduced by \textcite{Spivak2021Learners} and generalized by \textcite{Smithe2022Open}.
\end{rmk}

\begin{rmk}
  Note that $\blhom{Ay^S,By^T}$ is strictly speaking a monomial, since it can be written in the form $Iy^J$ for $I = \BLens{C}\bigl((A,S),(B,T)\bigr)$ and $J = \cat{C}(I,A)\times T$.
  However, we have written it in polynomial form with the view to extending it in future work to dependent lenses and dependent optics \parencite{Vertechi2022Dependent,Braithwaite2021Fibre} --- where we will call systems over such external hom polynomials \textit{cilia}, as they ``control optics'' --- and these generalized external homs will in fact be true polynomials.
\end{rmk}

\begin{prop} \label{prop:blhom-funct}
  Definition \ref{def:blhom} defines a functor $\BLens{C}\op\times\BLens{C}\to\Poly_{\cat{C}}$. Suppose $c:=(c_1,c^\#):(Z,R)\lensto(A,S)$ and $d:=(d_1,d^\#):(B,T)\lensto(C,U)$ are Bayesian lenses.
  We obtain a morphism of polynomials $\blhom{c,d}:\blhom{Ay^S,By^T}\to\blhom{Zy^R,Cy^U}$ as follows.
  Since the configurations of $\blhom{Ay^S,By^T}$ are lenses $(A,S)\lensto(B,T)$, the forwards map acts by pre- and post-composition:
  \begin{align*}
    \blhom{c,d}_1 := d\lenscirc(-)\lenscirc c : \BLens{C}\bigl((A,S),(B,T)\bigr) & \to \BLens{C}\bigl((Z,R),(C,U)\bigr) \\
    l & \mapsto d\lenscirc l\lenscirc c
  \end{align*}

  For each such $l$, the backwards map $\blhom{c,d}^\#_l$ has type $\cat{C}(I,Z)\otimes U \to \cat{C}(I,A)\otimes T$ in $\cat{C}$, and is obtained by analogy with the backwards composition rule for Bayesian lenses.
  We define
  \begin{align*}
    \blhom{c,d}^\#_l &:= \cat{C}(I,Z)\otimes U \xto{{c_1}_*\otimes U} \cat{C}(I,A)\otimes U \xto{\bcopier\otimes U} \cat{C}(I,A)\otimes\cat{C}(I,A)\otimes U \cdots \\
    & \qquad \cdots \xto{\cat{C}(I,A)\otimes {l_1}_*\otimes U} \cat{C}(I,A)\otimes\cat{C}(I,B)\otimes U\xto{\cat{C}(I,A)\otimes d^\#\otimes U} \cat{C}(I,A)\otimes\cat{C}(U,T)\otimes U \cdots \\
    & \qquad \cdots \xto{\cat{C}(I,A)\otimes\Fun{ev}_{U,T}} \cat{C}(I,A)\otimes T
  \end{align*}
  where $l_1$ is the forwards part of the lens $l:(A,S)\lensto(B,T)$, and ${c_1}_* := \cat{C}(I,c_1)$ and ${l_1}_* := \cat{C}(I,l_1)$ are the push-forwards along $c_1$ and $l_1$, and $\Fun{ev}_{U,T}$ is the evaluation map induced by the enrichment of $\cat{C}$ in $\Comon(\cat{C})$.
  In the special case where $\cat{C} = \Kl(\Pa)$ and $\Comon(\cat{C}) = \cat{E}$, we can write $\blhom{c,d}^\#_l$ as the following map in $\cat{E}$, depicted as a string diagram:
  \[
  \blhom{c,d}^\#_l \quad = \quad 
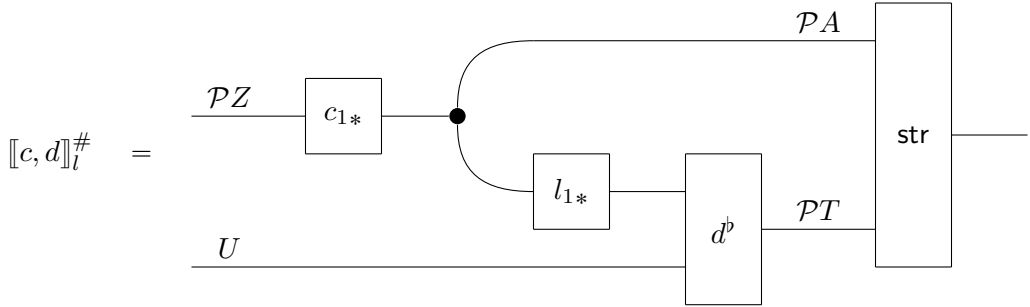
\begin{tikzpicture}
	\begin{pgfonlayer}{nodelayer}
		\node [style=none] (0) at (-3, 1) {};
		\node [style=none] (1) at (0, 2) {};
		\node [style=none] (2) at (1, 1) {${c_1}_*$};
		\node [style=none] (3) at (2, 2) {};
		\node [style=none] (4) at (0, 0) {};
		\node [style=none] (5) at (2, 0) {};
		\node [style=black copier] (6) at (4, 1) {};
		\node [style=none] (7) at (6, 3) {};
		\node [style=none] (8) at (6, -1) {};
		\node [style=none] (9) at (6, 0) {};
		\node [style=none] (10) at (8, 0) {};
		\node [style=none] (11) at (8, -2) {};
		\node [style=none] (12) at (6, -2) {};
		\node [style=none] (13) at (7, -1) {${l_1}_*$};
		\node [style=none] (14) at (10, -1) {};
		\node [style=none] (15) at (10, -3) {};
		\node [style=none] (16) at (10, -4) {};
		\node [style=none] (17) at (10, 0) {};
		\node [style=none] (18) at (11, -2) {$d^\flat$};
		\node [style=none] (19) at (12, 0) {};
		\node [style=none] (20) at (12, -4) {};
		\node [style=none] (21) at (12, -2) {};
		\node [style=none] (22) at (8, -1) {};
		\node [style=none] (24) at (2, 1) {};
		\node [style=none] (25) at (0, 1) {};
		\node [style=none] (26) at (15, -2) {};
		\node [style=none] (27) at (15, 3) {};
		\node [style=none] (28) at (-3, -3) {};
		\node [style=none] (29) at (-2, 1.5) {$\Pa Z$};
		\node [style=none] (30) at (-2, -2.5) {$U$};
		\node [style=none] (31) at (13.5, -1.5) {$\Pa T$};
		\node [style=none] (32) at (13.5, 3.5) {$\Pa A$};
		\node [style=none] (33) at (15, 4) {};
		\node [style=none] (34) at (17, 4) {};
		\node [style=none] (35) at (15, -3) {};
		\node [style=none] (36) at (17, -3) {};
		\node [style=none] (37) at (17, 0.5) {};
		\node [style=none] (38) at (16, 0.5) {$\mathsf{str}$};
		\node [style=none] (39) at (19, 0.5) {};
	\end{pgfonlayer}
	\begin{pgfonlayer}{edgelayer}
		\draw (0.center) to (25.center);
		\draw (24.center) to (6);
		\draw [in=-180, out=90, looseness=1.25] (6) to (7.center);
		\draw [in=180, out=-90, looseness=1.25] (6) to (8.center);
		\draw (28.center) to (15.center);
		\draw (22.center) to (14.center);
		\draw (7.center) to (27.center);
		\draw (21.center) to (26.center);
		\draw (1.center) to (3.center);
		\draw (1.center) to (4.center);
		\draw (4.center) to (5.center);
		\draw (3.center) to (5.center);
		\draw (9.center) to (10.center);
		\draw (9.center) to (12.center);
		\draw (12.center) to (11.center);
		\draw (10.center) to (11.center);
		\draw (17.center) to (19.center);
		\draw (17.center) to (16.center);
		\draw (16.center) to (20.center);
		\draw (19.center) to (20.center);
		\draw (33.center) to (34.center);
		\draw (33.center) to (35.center);
		\draw (35.center) to (36.center);
		\draw (34.center) to (36.center);
		\draw (37.center) to (39.center);
	\end{pgfonlayer}
\end{tikzpicture}

  \]
  Here, we have assumed that $\Kl(\Pa)(I,A) = \Pa A$, and define $d^\flat : \Pa B\times U\to\Pa T$ to be the image of $d^\# : \Pa B\to\Kl(\Pa)(U,T)$ under the Cartesian closure of $\cat{E}$, and $\mathsf{str}:\Pa A\times\Pa T\to\Pa\bigl(\Pa A\times T)$ the (right) strength of the strong monad $\Pa$.
  \begin{proof}
    We need to check that the mappings defined above respect identities and composition.
    It is easy to see that the definition preserves identities: in the forwards direction, this follows from the unitality of composition in $\BLens{C}$; in the backwards direction, because pushing forwards along the identity is again the identity, and because the backwards component of the identity Bayesian lens is the constant state-dependent morphism on the identity in $\cat{C}$.

    To check that the mapping preserves composition, we consider the contravariant and covariant parts separately. Suppose $b:=(b_1,b^\#):(Y,Q)\lensto(Z,R)$ and $e:=(e_1,e^\#):(C,U)\lensto(D,V)$ are Bayesian lenses.
    We consider the contravariant case first: we check that $\blhom{c\lenscirc b,By^T} = \blhom{b,By^T}\circ\blhom{c,By^T}$.
    The forwards direction holds by pre-composition of lenses.
    In the backwards direction, we note from the definition that only the forwards channel $c_1$ plays a role in $\blhom{c,By^T}^\#_l$, and that role is again pre-composition.
    We therefore only need to check that $(c_1\klcirc b_1)_* = {c_1}_*\circ{b_1}_*$, and this follows immediately from the functoriality of $\cat{C}(I,-)$.

    We now consider the covariant case, that $\blhom{Ay^S,e\lenscirc d} = \blhom{Ay^S,e}\circ\blhom{Ay^S,d}$.
    Once again, the forwards direction holds by composition of lenses.
    For simplicity of exposition, we consider the backwards direction in the case $\cat{C} = \Kl(\Pa)$ and reason graphically.
    In this case, the backwards map on the right-hand side is given, for a lens $l:(A,S)\lensto(B,T)$ by the following string diagram:
    
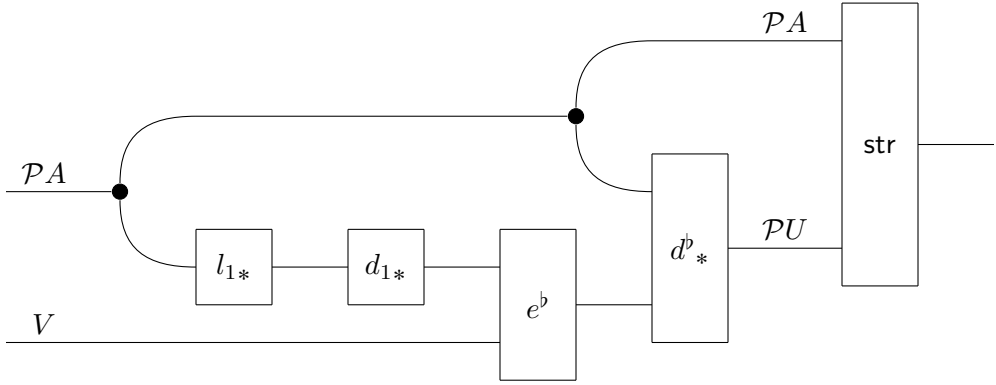
\[
    \begin{tikzpicture}
	\begin{pgfonlayer}{nodelayer}
		\node [style=none] (0) at (-12, 0) {};
		\node [style=black copier] (6) at (-9, 0) {};
		\node [style=none] (7) at (-7, 2) {};
		\node [style=none] (9) at (-7, -1) {};
		\node [style=none] (10) at (-5, -1) {};
		\node [style=none] (11) at (-5, -3) {};
		\node [style=none] (12) at (-7, -3) {};
		\node [style=none] (13) at (-6, -2) {${l_1}_*$};
		\node [style=none] (14) at (-3, -2) {};
		\node [style=none] (15) at (1, -4) {};
		\node [style=none] (16) at (1, -5) {};
		\node [style=none] (17) at (1, -1) {};
		\node [style=none] (18) at (2, -3) {$e^\flat$};
		\node [style=none] (19) at (3, -1) {};
		\node [style=none] (20) at (3, -5) {};
		\node [style=none] (21) at (3, -3) {};
		\node [style=none] (22) at (-5, -2) {};
		\node [style=none] (26) at (5, -3) {};
		\node [style=none] (28) at (-12, -4) {};
		\node [style=none] (29) at (-11, 0.5) {$\Pa A$};
		\node [style=none] (30) at (-11, -3.5) {$V$};
		\node [style=none] (40) at (-3, -1) {};
		\node [style=none] (41) at (-1, -1) {};
		\node [style=none] (42) at (-1, -2) {};
		\node [style=none] (43) at (-3, -3) {};
		\node [style=none] (44) at (-1, -3) {};
		\node [style=none] (45) at (1, -2) {};
		\node [style=black copier] (46) at (3, 2) {};
		\node [style=none] (47) at (5, 0) {};
		\node [style=none] (48) at (5, 4) {};
		\node [style=none] (49) at (5, 1) {};
		\node [style=none] (50) at (7, 1) {};
		\node [style=none] (51) at (5, -4) {};
		\node [style=none] (52) at (7, -4) {};
		\node [style=none] (53) at (7, -1.5) {};
		\node [style=none] (54) at (10, -1.5) {};
		\node [style=none] (55) at (10, 4) {};
		\node [style=none] (56) at (6, -1.5) {${d^\flat}_*$};
		\node [style=none] (57) at (-2, -2) {${d_1}_*$};
		\node [style=none] (58) at (10, 5) {};
		\node [style=none] (59) at (12, 5) {};
		\node [style=none] (60) at (10, -2.5) {};
		\node [style=none] (61) at (12, -2.5) {};
		\node [style=none] (62) at (11, 1.25) {$\mathsf{str}$};
		\node [style=none] (63) at (12, 1.25) {};
		\node [style=none] (64) at (14, 1.25) {};
		\node [style=none] (65) at (8.5, -1) {$\Pa U$};
		\node [style=none] (66) at (8.5, 4.5) {$\Pa A$};
		\node [style=none] (67) at (-7, -2) {};
	\end{pgfonlayer}
	\begin{pgfonlayer}{edgelayer}
		\draw [in=-180, out=90, looseness=1.25] (6) to (7.center);
		\draw (28.center) to (15.center);
		\draw (22.center) to (14.center);
		\draw (21.center) to (26.center);
		\draw (9.center) to (10.center);
		\draw (9.center) to (12.center);
		\draw (12.center) to (11.center);
		\draw (10.center) to (11.center);
		\draw (17.center) to (19.center);
		\draw (17.center) to (16.center);
		\draw (16.center) to (20.center);
		\draw (19.center) to (20.center);
		\draw (0.center) to (6);
		\draw (7.center) to (46);
		\draw [in=-180, out=90, looseness=1.25] (46) to (48.center);
		\draw [in=180, out=-90, looseness=1.25] (46) to (47.center);
		\draw (49.center) to (50.center);
		\draw (49.center) to (51.center);
		\draw (51.center) to (52.center);
		\draw (50.center) to (52.center);
		\draw (40.center) to (41.center);
		\draw (40.center) to (43.center);
		\draw (43.center) to (44.center);
		\draw (41.center) to (44.center);
		\draw (42.center) to (45.center);
		\draw (48.center) to (55.center);
		\draw (53.center) to (54.center);
		\draw (58.center) to (59.center);
		\draw (58.center) to (60.center);
		\draw (60.center) to (61.center);
		\draw (59.center) to (61.center);
		\draw (63.center) to (64.center);
		\draw [in=180, out=-90, looseness=1.25] (6) to (67.center);
	\end{pgfonlayer}
\end{tikzpicture}

    \]
    It is easy to verify that the composition of backwards channels here is precisely the backwards channel given by $e\lenscirc d$---compare \textcite[Theorem 3.14]{Smithe2021Compositional1} or \parencite[Theorem 5.2]{Smithe2020Bayesian}---which establishes the result.
    The case for general $\cat{C}$ is directly analogous, on the other side of the tensor-hom adjunction.
  \end{proof}
\end{prop}

Now that we have an `external hom', we might expect also to have a corresponding `external composition', represented by a family of morphisms of polynomials; we establish such a family now, and it will be important in our bicategorical construction.

\begin{defn} \label{def:blhom-comp}
  We define an `external composition' natural transformation $\Fun{c}$, with components
  \[ \blhom{Ay^S,By^T}\otimes\blhom{By^T,Cy^U}\to\blhom{Ay^S,Cy^U} \]
  given in the forwards direction by composition of Bayesian lenses.
  In the backwards direction, for each pair of lenses $c:(A,S)\lensto(B,T)$ and $d:(B,T)\lensto(C,U)$, we need a map
  \[ \Fun{c}^\#_{c,d} : \cat{C}(I,A)\otimes U \to \cat{C}(I,A)\otimes T\otimes \cat{C}(I,B)\otimes U\bigr)\]
  which we define as follows:
  \begin{align*}
    \Fun{c}^\#_{c,d} &:= \cat{C}(I,A) \otimes U \xto{\bcopier\otimes\bcopier} \cat{C}(I,A) \otimes \cat{C}(I,A) \otimes U \otimes U \cdots \\
    &\qquad\cdots \xto{\cat{C}(I,A)\otimes{c_1}_*\otimes U\otimes U} \cat{C}(I,A) \otimes \cat{C}(I,B) \otimes U \otimes U \cdots \\
    &\qquad\cdots \xto{\cat{C}(I,A)\otimes\bcopier\otimes\cat{C}(I,B)\otimes U\otimes U} \cat{C}(I,A) \otimes \cat{C}(I,B) \otimes \cat{C}(I,B) \otimes U \otimes U \\
    &\qquad\cdots \xto{\cat{C}(I,A)\otimes\cat{C}(I,B)\otimes d^\#\otimes U\otimes U} \cat{C}(I,A) \otimes \cat{C}(I,B) \otimes \cat{C}(U,T) \otimes Y \otimes U \\
    &\qquad\cdots \xto{\cat{C}(I,A)\otimes\cat{C}(I,B)\Fun{ev}_{U,T}\otimes U} \cat{C}(I,A) \otimes \cat{C}(I,B) \otimes T \otimes U \\
    &\qquad\cdots \xto{\cat{C}(I,A)\otimes\mathsf{swap}\otimes U} \cat{C}(I,A) \otimes T \otimes \cat{C}(I,B) \otimes U
  \end{align*}
  where ${c_1}_*$ and $\Fun{ev}_{U,T}$ are as in \ref{prop:blhom-funct}.

  In the case where $\cat{C} = \Kl(\Pa)$, we can equivalently (and more legibly) define $\Fun{c}^\#_{c,d}$ by the following string diagram:
  \[
  \Fun{c}^\#_{c,d} \quad := \quad 
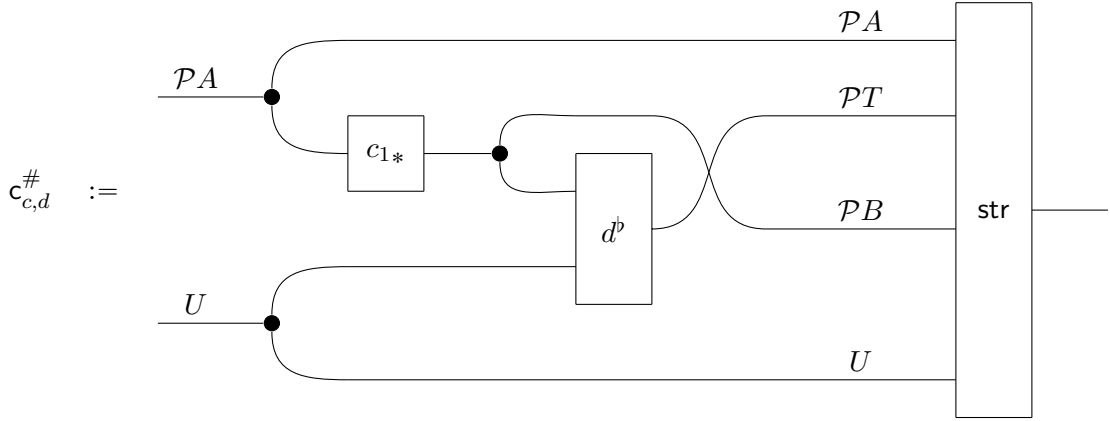
\begin{tikzpicture}
	\begin{pgfonlayer}{nodelayer}
		\node [style=none] (0) at (0, -1) {$d^\flat$};
		\node [style=none] (1) at (-1, 1) {};
		\node [style=none] (2) at (1, 1) {};
		\node [style=none] (3) at (1, -3) {};
		\node [style=none] (4) at (-1, -3) {};
		\node [style=none] (6) at (-1, 2) {};
		\node [style=none] (7) at (1, -1) {};
		\node [style=none] (8) at (1, 2) {};
		\node [style=none] (9) at (4, 2) {};
		\node [style=none] (10) at (4, -1) {};
		\node [style=none] (11) at (-5, 2) {};
		\node [style=none] (12) at (-7, 2) {};
		\node [style=none] (13) at (-7, 0) {};
		\node [style=none] (14) at (-5, 0) {};
		\node [style=none] (15) at (-6, 1) {${c_1}_*$};
		\node [style=none] (16) at (-5, 1) {};
		\node [style=none] (17) at (-7, 1) {};
		\node [style=none] (18) at (-1, 0) {};
		\node [style=none] (19) at (1, -1) {};
		\node [style=none] (21) at (-7, 4) {};
		\node [style=none] (22) at (9, 4) {};
		\node [style=none] (23) at (-1, -2) {};
		\node [style=none] (24) at (-7, -2) {};
		\node [style=none] (26) at (-7, -5) {};
		\node [style=none] (27) at (9, -5) {};
		\node [style=black copier] (28) at (-9, 2.5) {};
		\node [style=black copier] (29) at (-3, 1) {};
		\node [style=black copier] (30) at (-9, -3.5) {};
		\node [style=none] (31) at (9, 5) {};
		\node [style=none] (32) at (11, 5) {};
		\node [style=none] (33) at (9, -6) {};
		\node [style=none] (34) at (11, -6) {};
		\node [style=none] (35) at (11, -0.5) {};
		\node [style=none] (36) at (10, -0.5) {$\mathsf{str}$};
		\node [style=none] (37) at (-12, 2.5) {};
		\node [style=none] (38) at (-12, -3.5) {};
		\node [style=none] (39) at (13, -0.5) {};
		\node [style=none] (40) at (6.5, 4.5) {$\Pa A$};
		\node [style=none] (41) at (6.5, 2.5) {$\Pa T$};
		\node [style=none] (42) at (6.5, -0.5) {$\Pa B$};
		\node [style=none] (43) at (6.5, -4.5) {$U$};
		\node [style=none] (44) at (-11, 3) {$\Pa A$};
		\node [style=none] (45) at (-11, -3) {$U$};
		\node [style=none] (46) at (9, 2) {};
		\node [style=none] (47) at (9, -1) {};
	\end{pgfonlayer}
	\begin{pgfonlayer}{edgelayer}
		\draw (6.center) to (8.center);
		\draw (21.center) to (22.center);
		\draw (19.center) to (7.center);
		\draw [in=-180, out=0, looseness=1.25] (7.center) to (9.center);
		\draw [in=180, out=0, looseness=1.25] (8.center) to (10.center);
		\draw (24.center) to (23.center);
		\draw (26.center) to (27.center);
		\draw [in=-180, out=90, looseness=1.25] (28) to (21.center);
		\draw [in=180, out=-90, looseness=1.25] (28) to (17.center);
		\draw (16.center) to (29);
		\draw [in=-180, out=90, looseness=1.25] (29) to (6.center);
		\draw [in=-180, out=-90, looseness=1.25] (29) to (18.center);
		\draw [in=-180, out=90, looseness=1.25] (30) to (24.center);
		\draw [in=180, out=-90, looseness=1.25] (30) to (26.center);
		\draw (12.center) to (11.center);
		\draw (11.center) to (14.center);
		\draw (12.center) to (13.center);
		\draw (13.center) to (14.center);
		\draw (1.center) to (4.center);
		\draw (4.center) to (3.center);
		\draw (1.center) to (2.center);
		\draw (2.center) to (3.center);
		\draw (31.center) to (32.center);
		\draw (31.center) to (33.center);
		\draw (33.center) to (34.center);
		\draw (32.center) to (34.center);
		\draw (37.center) to (28);
		\draw (38.center) to (30);
		\draw (35.center) to (39.center);
		\draw (9.center) to (46.center);
		\draw (10.center) to (47.center);
	\end{pgfonlayer}
\end{tikzpicture}

  \]
  where $d^\flat$ and $\Fun{str}$ are also as in Proposition \ref{prop:blhom-funct}.
\end{defn}

We leave to the reader the detailed proof that this definition produces a well-defined natural transformation, noting only that the argument is analogous to that of Proposition \ref{prop:blhom-funct}:
one observes that, in the forwards direction, the definition is simply composition of Bayesian lenses (which is immediately natural);
in the backwards direction, one observes that the definition again mirrors that of the backwards composition of Bayesian lenses.

Next, we establish the structure needed to make our bicategory monoidal.

\begin{defn} \label{def:blhom-tensor-dist}
  We define a distributive law $\Fun{d}$ of $\blhom{-,=}$ over $\otimes$, a natural transformation with components
  \[ \blhom{Ay^S,By^T}\otimes\blhom{A'y^{S'},B'y^{T'}}\to\blhom{Ay^S\otimes A'y^{S'},By^T\otimes B'y^{T'}} \, , \]
  noting that $Ay^S\otimes A'y^{S'} = (A\otimes A')y^{(S\otimes S')}$ and $By^T\otimes B'y^{T'} = (B\otimes B')y^{(T\otimes T')}$.
  The forwards component is given simply by taking the tensor of the corresponding Bayesian lenses, using the monoidal product (also denoted $\otimes$) in $\BLens{C}$.
  Backwards, for each pair of lenses $c:(A,S)\lensto(B,T)$ and $c':(A',S')\lensto(B',T')$, we need a map
  \[ \Fun{d}^\#_{c,c'} : \cat{C}(I, A\otimes A')\otimes T\otimes T' \to \cat{C}(I,A)\times T\times \cat{C}(I,A')\times T' \]
  for which we choose
  \begin{align*}
    & \cat{C}(I, A\otimes A')\otimes T\otimes T' \xto{\bcopier\otimes T\otimes T'} \cat{C}(I, A\otimes A')\otimes \cat{C}(I, A\otimes A')\otimes T\otimes T' \cdots \\
    & \cdots\, \xto{\cat{C}(I,\mathsf{proj}_A)\otimes\cat{C}(I,\mathsf{proj}_{A'})\otimes T\otimes T'} \cat{C}(I,A)\otimes \cat{C}(I,A')\otimes T\otimes T' \cdots \\
    & \cdots\, \xto{\cat{C}(I,A)\otimes\Fun{swap}\otimes T'} \cat{C}(I,A)\otimes T\otimes \cat{C}(I,A')\otimes T'
  \end{align*}
  where $\Fun{swap}$ is the symmetry of the tensor $\otimes$ in $\cat{C}$.
  Note that $\Fun{d}^\#_{c,c'}$ so defined does not in fact depend on either $c$ or $c'$.
\end{defn}

We now have everything we need to construct a monoidal bicategory $\Hier^\Tt_{\cat{C}}$ of dynamical hierarchical inference systems in $\cat{C}$, following the intuition outlined at the beginning of this section.

\begin{rmk}
  The notion of bicategory that we adopt is the standard one of `category weakly enriched in $\Cat{Cat}$', so that between any two 0-cells we have a category of 1-cells (and 2-cells between them), such that composition of 1-cells is associative and unital up to natural isomorphism.
\end{rmk}

\begin{defn} \label{def:hier-bicat}
  Let $\Hier^\Tt_{\cat{C}}$ denote the monoidal bicategory whose 0-cells are objects $(A,S)$ in $\BLens{C}$, and whose hom-categories $\Hier^\Tt_{\cat{C}}\bigl((A,S),(B,T)\bigr)$ are given by $\Coalg^\Tt_{\cat{C}}\bigl(\blhom{Ay^S,By^T}\bigr)$.
  The identity 1-cell $\id_{(A,S)} : (A,S)\to(A,S)$ on $(A,S)$ is given by the system with trivial state space $1$, trivial update map, and output map that constantly emits the identity Bayesian lens $(A,S)\lensto(A,S)$.
  The composition of a system $(A,S)\to(B,T)$ then a system $(B,T)\to(C,U)$ is defined by the functor
  \begin{align*}
    & \Hier^\Tt_{\cat{C}}\bigl((A,S),(B,T)\bigr)\times\Hier^\Tt_{\cat{C}}\bigl((B,T),(C,U)\bigr) \\
    & = \Coalg^\Tt_{\cat{C}}\bigl(\blhom{Ay^S,By^T}\bigr)\times\Coalg^\Tt_{\cat{C}}\bigl(\blhom{By^T,Cy^U}\bigr) \\
    & \xto{\lambda} \Coalg^\Tt_{\cat{C}}\bigl(\blhom{Ay^S,By^T}\otimes\blhom{By^T,Cy^U}\bigr) \\
    & \xto{\Coalg^\Tt_{\cat{C}}(\mathsf{c})} \Coalg^\Tt_{\cat{C}}\bigl(\blhom{Ay^S,Cy^U}\bigr)
      = \Hier^\Tt_{\cat{C}}\bigl((A,S),(C,U)\bigr)
  \end{align*}
  where $\lambda$ is the laxator and $\mathsf{c}$ is the external composition morphism of Definition \ref{def:blhom-comp}.

  \noindent
  The monoidal structure $(\otimes, y)$ on $\Hier^\Tt_{\cat{C}}$ derives from the structures on $\Poly_{\cat{C}}$ and $\BLens{C}$, justifying our overloaded notation.
  On 0-cells, $(A,S)\otimes(A',S') := (A\otimes A',S\otimes S')$.
  On 1-cells $(A,S)\to(B,T)$ and $(A',S')\to(B',T')$, the tensor is given by
  \begin{align*}
    & \Hier^\Tt_{\cat{C}}\bigl((A,S),(B,T)\bigr)\times\Hier^\Tt_{\cat{C}}\bigl((A',S'),(B',T')\bigr) \\
    & = \Coalg^\Tt_{\cat{C}}\bigl(\blhom{Ay^S,By^T}\bigr)\times\Coalg^\Tt_{\cat{C}}\bigl(\blhom{A'y^{S'},B'y^{T'}}\bigr) \\
    & \xto{\lambda} \Coalg^\Tt_{\cat{C}}\bigl(\blhom{Ay^S,By^T}\otimes\blhom{A'y^{S'},B'y^{T'}}\bigr) \\
    & \xto{\Coalg^{\Tt}_{\cat{C}}(\mathsf{d})} \Coalg^\Tt_{\cat{C}}\bigl(\blhom{Ay^S\otimes A'y^{S'},By^T\otimes B'y^{T'}}\bigr) \\
    & = \Hier^\Tt_{\cat{C}}\bigl((A,S)\otimes(A',S'),(B,T)\otimes(B',T')\bigr)
  \end{align*}
  where $\mathsf{d}$ is the distributive law of Definition \ref{def:blhom-tensor-dist}.
  The same functors
  \[ \Hier^\Tt_{\cat{C}}\bigl((A,S),(B,T)\bigr)\times\Hier^\Tt_{\cat{C}}\bigl((A',S'),(B',T')\bigr) \to \Hier^\Tt_{\cat{C}}\bigl((A,S)\otimes(A',S'),(B,T)\otimes(B',T')\bigr) \]
  induce the tensor of 2-cells; concretely, this is given on morphisms of dynamical systems by taking the product of the corresponding morphisms between state spaces.
\end{defn}

We do not give here a proof that this makes $\Hier^\Tt_{\cat{C}}$ into a well-defined monoidal bicategory;
briefly, the result follows from the facts that the external composition $\mathsf{c}$ and the tensor $\otimes$ are appropriately associative and unital, that $\Coalg^\Tt_\Pa$ is lax monoidal, that $\blhom{{-},{=}}$ is functorial in both positions, and that $\blhom{{-},{=}}$ distributes naturally over $\otimes$.

Before we move on to considering doctrines of approximate inference, it will be useful to spell out concretely the elements of a morphism $(A,S)\to(B,T)$ in $\Hier^\Tt_{\Kl(\Pa)}$.

\begin{prop} \label{prop:unpack-hier}
  Suppose $\Pa$ is a monad on a Cartesian closed category $\cat{E}$.
  Then a 1-cell $\vartheta:(A,S)\to(B,T)$ in $\Hier^\Tt_{\Kl(\Pa)}$ is given by a tuple $\vartheta := (X,\vartheta^o_1,\vartheta^o_2,\vartheta^u)$ of
  \begin{itemize}
  \item a choice of state space $X$,
  \item a forwards output map $\vartheta^o_1:\Tt\times X\times A\to\Pa B$ in $\cat{E}$,
  \item a backwards output map $\vartheta^o_2:\Tt\times X\times\Pa A\times T\to\Pa S$ in $\cat{E}$, and
  \item an update map $\vartheta^u:\Tt\times X\times\Pa A\times T\to\Pa X$ in $\cat{E}$,
  \end{itemize}
  satisfying the `flow' condition of Proposition \ref{prop:pM-coalg}.
  \begin{proof}
    The result follows immediately upon unpacking the definitions, using the Cartesian closure of $\cat{E}$.
  \end{proof}
\end{prop}

\subsection{Differential and `cybernetic' systems} \label{sec:dyn-diff}

Approximate inference doctrines describe how systems play statistical games, and are particularly of interest when one asks how systems' performance may improve during such game-playing.
One prominent method of performance improvement involves descending the gradient of the statistical game's loss function, and we will see below that this method is adopted by both the Laplace and the Hebb-Laplace doctrines.
The appearance of gradient descent prompts questions about the connections between such statistical systems and other `cybernetic' systems such as deep learners or players of economic games, both of which may also involve gradient descent \parencite{Cruttwell2022Categorical,Capucci2022Diegetic}; indeed, it has been proposed \parencite{Capucci2021Towards} that parameterized gradient descent should form the basis of a compositional account of cybernetic systems in general\footnote{
Our own view on cybernetics is somewhat more general, since not all systems that may be seen as cybernetic are explicitly structured as gradient-descenders, and nor even is explicit differential structure always apparent.
In earlier work, we suggested that statistical inference was perhaps more inherent to cybernetics \parencite{Smithe2020Cyber}, although today we believe that a better, though more informal, definition of cybernetic system is perhaps ``an intentionally-controlled open dynamical system''.
Nonetheless, we acknowledge that this notion of ``intentional control'' may generally be reducible to a stationary action principle, again indicating the importance of differential structure.
We leave the statement and proof of this general principle to future work.}.

In order to incorporate gradient descent explicitly into our own compositional framework, we follow the recipes above to define here first a category of differential systems opindexed by polynomial interfaces and then a monoidal bicategory of differential hierarchical inference systems.
We then show how we can obtain dynamical from differential systems by integration, and sketch how this induces a ``change of base'' from dynamical to differential hierarchical inference systems.

\begin{notation}
  Write $\Cat{Diff}_{\cat{C}}$ for the subcategory of compact smooth manifold objects in $\Comon(\cat{C})$ and differentiable morphisms between them.
  Write $T : \Cat{Diff}_{\cat{C}} \to \Cat{Vect}(\Cat{Diff}_{\cat{C}})$ for the corresponding tangent bundle functor, where $\Cat{Vect}(\Cat{Diff}_{\cat{C}})$ is (the total category of) the fibration of vector bundles over $\Cat{Diff}_{\cat{C}}$ and their homomorphisms.
  Write $U : \Cat{Vect}(\Cat{Diff}_{\cat{C}}) \to \Cat{Diff}_{\cat{C}}$ for the functor that forgets the bundle structure.
  Write $\Fun{T} := UT : \Cat{Diff}_{\cat{C}} \to \Cat{Diff}_{\cat{C}}$ for the induced endofunctor.
\end{notation}

Recall that morphisms $Ay^B\to p$ in $\Poly_{\cat{C}}$ correspond to morphisms $A\to pB$ in $\cat{C}$.

\begin{defn}
  For each $p : \Poly_{\cat{C}}$, define the category $\DiffSys_{\cat{C}}(p)$ as follows.
  Its objects are objects $M:\Cat{Diff}_{\cat{C}}$, each equipped with a morphism $m : My^{\Fun{T}M} \to p$ of polynomials in $\Poly_{\cat{C}}$, such that for any section $\sigma : p\to y$ of $p$, the composite morphism $\sigma\circ m : My^{\Fun{T}M}\to y$ corresponds to a section $m^\sigma : M\to\Fun{T}M$ of the tangent bundle $\Fun{T}M\to M$.
  A morphism $\alpha:(M,m)\to(M',m')$ in $\DiffSys_{\cat{C}}(p)$ is a map $\alpha : M \to M'$ in $\Cat{Diff}_{\cat{C}}$ such that the following diagram commutes:
  \[\begin{tikzcd}
	M & {p\Fun{T}M} \\
	{M'} & {p\Fun{T}M'}
	\arrow["m", from=1-1, to=1-2]
	\arrow["\alpha"', from=1-1, to=2-1]
	\arrow["{m'}"', from=2-1, to=2-2]
	\arrow["{p\Fun{T}\alpha}", from=1-2, to=2-2]
  \end{tikzcd}\]
\end{defn}

\begin{prop}
  $\DiffSys_{\cat{C}}$ defines an opindexed category $\Poly_{\cat{C}}\to\Cat{Cat}$.
  Given a morphism $\varphi:p\to q$ of polynomials, $\DiffSys_{\cat{C}}(\varphi) : \DiffSys_{\cat{C}}(p) \to \DiffSys_{\cat{C}}(q)$ acts on objects by postcomposition and trivially on morphisms.
\end{prop}

\begin{prop} \label{prop:diffsys-lax}
  The functor $\DiffSys_{\cat{C}}$ is lax monoidal $(\Poly_{\cat{C}},\otimes,y) \to (\Cat{Cat},\times,\Cat{1})$.
  \begin{proof}[Proof sketch]
    Note that $\Fun{T}$ is strong monoidal, with $\Fun{T}(1) \cong 1$ and $\Fun{T}(M)\otimes\Fun{T}(N)\cong\Fun{T}(M\otimes N)$.
    The unitor $\Cat{1}\to\DiffSys_{\cat{C}}(y)$ is given by the isomorphism $1y^{\Fun{T}1} \cong 1y^1 \cong y$ induced by the strong monoidal structure of $\Fun{T}$.
    The laxator $\lambda_{p,q} : \DiffSys_{\cat{C}}(p) \times \DiffSys_{\cat{C}}(q) \to \DiffSys_{\cat{C}}(p\otimes q)$ is similarly determined: given objects $m:My^{\Fun{T}M}\to p$ and $n:Ny^{\Fun{T}N}\to q$, take their tensor $m\otimes n:(M\otimes N)y^{\Fun{T}M\otimes\Fun{T}N}$ and precompose with the induced morphism $(M\otimes N)y^{\Fun{T}(M\otimes N)} \to (M\otimes N)y^{\Fun{T}M\otimes\Fun{T}N}$; proceed similarly on morphisms of differential systems.
    The satisfaction of the unitality and associativity laws follows from the monoidality of $\Fun{T}$.
  \end{proof}
\end{prop}

We now define a monoidal bicategory $\DiffHier_{\cat{C}}$ of differential hierarchical inference systems, following the definition of $\Hier^\Tt_{\cat{C}}$ above.

\begin{defn}
  Let $\DiffHier_{\cat{C}}$ denote the monoidal bicategory whose 0-cells are again the objects $(A,S)$ of $\BLens{C}$ and whose hom-categories $\DiffHier_{\cat{C}}\bigl((A,S),(B,T)\bigr)$ are given by $\DiffSys_{\cat{C}}\bigl(\blhom{Ay^S,By^T}\bigr)$.
  The identity 1-cell $\id_{(A,S)}:(A,S)\to(A,S)$ on $(A,S)$ is given by the differential system $y \to \blhom{Ay^S,By^T}$ with state space $1$, trivial backwards component, and forwards component that picks the identity Bayesian lens on $(A,S)$.
  The composition of differential systems $(A,S)\to(B,T)$ then $(B,T)\to(C,U)$ is defined by the functor
  \begin{align*}
    & \DiffHier_{\cat{C}}\bigl((A,S),(B,T)\bigr)\times\DiffHier_{\cat{C}}\bigl((B,T),(C,U)\bigr) \\
    & = \DiffSys_{\cat{C}}\bigl(\blhom{Ay^S,By^T}\bigr)\times\DiffSys_{\cat{C}}\bigl(\blhom{By^T,Cy^U}\bigr) \\
    & \xto{\lambda} \DiffSys_{\cat{C}}\bigl(\blhom{Ay^S,By^T}\otimes\blhom{By^T,Cy^U}\bigr) \\
    & \xto{\DiffSys_{\cat{C}}(\mathsf{c})} \DiffSys_{\cat{C}}\bigl(\blhom{Ay^S,Cy^U}\bigr)
    = \DiffHier_{\cat{C}}\bigl((A,S),(C,U)\bigr)
  \end{align*}
  where $\lambda$ is the laxator of Proposition \ref{prop:diffsys-lax} and $\mathsf{c}$ is the external composition morphism of Definition \ref{def:blhom-comp}.

  \noindent
  The monoidal structure $(\otimes,y)$ on $\DiffHier_{\cat{C}}$ is similarly defined following that of $\Hier^\Tt_{\cat{C}}$.
  On 0-cells, $(A,S)\otimes(A',S') := (A\otimes A',S\otimes S')$.
  On 1-cells $(A,S)\to(B,T)$ and $(A',S')\to(B',T')$ (and their 2-cells), the tensor is given by the functors
  \begin{align*}
    & \DiffHier_{\cat{C}}\bigl((A,S),(B,T)\bigr)\times\DiffHier_{\cat{C}}\bigl((A',S'),(B',T')\bigr) \\
    & = \DiffSys_{\cat{C}}\bigl(\blhom{Ay^S,By^T}\bigr)\times\DiffSys_{\cat{C}}\bigl(\blhom{A'y^{S'},B'y^{T'}}\bigr) \\
    & \xto{\lambda} \DiffSys_{\cat{C}}\bigl(\blhom{Ay^S,By^T}\otimes\blhom{A'y^{S'},B'y^{T'}}\bigr) \\
    & \xto{\Coalg^{\Tt}_{\cat{C}}(\mathsf{d})} \DiffSys_{\cat{C}}\bigl(\blhom{Ay^S\otimes A'y^{S'},By^T\otimes B'y^{T'}}\bigr) \\
    & = \DiffHier_{\cat{C}}\bigl((A,S)\otimes(A',S'),(B,T)\otimes(B',T')\bigr)
  \end{align*}
  where $\mathsf{d}$ is the distributive law of Definition \ref{def:blhom-tensor-dist}.
\end{defn}

Following Prop. \ref{prop:unpack-hier}, we have the following characterization of a differential hierarchical inference system $(A,S) \to (B,T)$ in $\Kl(\Pa)$, for $\Pa : \cat{E} \to \cat{E}$.

\begin{prop} \label{prop:unpack-diff}
  A 1-cell $\delta : (A,S) \to (B,T)$ in $\DiffHier_{\Kl(\Pa)}$ is given by a tuple $\delta := (X,\delta^o_1,\delta^o_2,\delta^\#)$ of
  \begin{itemize}
  \item a choice of `state space' $X:\Cat{Diff}_{\cat{E}}$;
  \item a forwards output map $\delta^o_1 : X\times A \to \Pa B$ in $\cat{E}$,
  \item a backwards output map $\delta^o_2 : X\times \Pa A \times T \to \Pa S$ in $\cat{E}$,
  \item a stochastic vector field $\delta^\# : X\times \Pa A\times T \to \Pa\Fun{T} X$ in $\cat{E}$.
  \end{itemize}
\end{prop}

We can obtain continuous-time dynamical systems from differential systems by integration, and consider how to discretize these flows to give discrete-time dynamical systems.

\begin{prop} \label{prop:diff-flow}
  Integration induces an indexed functor $\Fun{Flow} : \DiffSys_{\cat{C}} \to \Coalg^\rr_{\cat{C}}$.
  \begin{proof}
    Suppose $(M,m)$ is an object in $\DiffSys_{\cat{C}}(p)$.
    The morphism $m : My^{\Fun{T}M} \to p$ consists of a map $m_1 : M\to p(1)$ in $\Comon(\cat{C})$ along with a morphism $m^\# : \sum_{x:M} p[m_1(x)] \to \Fun{T}M$ in $\cat{C}$.
    Since, for any section $\sigma : p\to y$, the induced map $m^\sigma : M\to \Fun{T}M$ is a vector field on a compact manifold, it generates a unique global flow $\Fun{Flow}(p)(m)^\sigma : \rr\times M\to M$ \parencite[{Thm.s 12.9,12.12}]{Lee2012Smooth}, which factors as
    \[ \sum_{t:\rr} M \xto{m_1^*\sigma} \sum_{t:\rr} \sum_{x:M} p[m_1(x)] \xto{\Fun{Flow}(p)(m)^u} M \, . \]
    We therefore define the system $\Fun{Flow}(p)(m)$ to have state space $M$, output map $m_1$ (for all $t:\rr$), and update map $\Fun{Flow}(p)(m)^u$.
    Since $\Fun{Flow}(p)(m)^\sigma$ is a flow for any section $\sigma$, it immediately satisfies the monoid action condition.
    On morphisms $\alpha : m\to m'$, we define $\Fun{Flow}(p)(\alpha)$ by the same underlying map on state spaces; this is again well-defined by the condition that $\alpha$ is compatible with the tangent structure.
    Given a morphism $\varphi:p\to q$ of polynomials, both the reindexing $\DiffSys_{\cat{C}}(\varphi)$ and $\Coalg^\rr_{\cat{C}}(\varphi)$ act by postcomposition, and so it is easy to see that $\Coalg^\rr_{\cat{C}}(\varphi)\circ\Fun{Flow}(p) \cong \Fun{Flow}(q)\circ\DiffSys_{\cat{C}}(\varphi)$ naturally.
  \end{proof}
\end{prop}

\begin{rmk} \label{rmk:diff-disc}
  From Proposition \ref{prop:diff-flow} and the earlier Corollary \ref{cor:coalg-disc}, we obtain a family of composite indexed functors $\DiffSys_{\cat{C}} \xto{\Fun{Flow}} \Coalg^\rr_{\cat{C}} \xto{\Fun{Disc}_k} \Coalg^\nn_{\cat{C}}$ taking each differential system to a discrete-time dynamical system in $\cat{C}$.
  Below, we will define approximate inference doctrines in discrete time that arise from processes of (stochastic) gradient descent, and which therefore factor through differential systems, but the form in which these are given---and in which they are found in the informal literature (\textit{e.g.}, \textcite{Bogacz2017tutorial})---is not obtained via the composite $\Fun{Disc}_k \circ \Fun{Flow}$ for any $k$, even though there is a free parameter $k$ that plays the same role (intuitively, a `learning rate').
  Instead, one typically adopts the following `naïve' discretization scheme.

  Let $\CartDiffSys_{\cat{C}}$ denote the sub-indexed category of $\DiffSys_{\cat{C}}$ spanned by those systems with Cartesian state spaces $\rr^n$.
  Naive discretization induces a family of indexed functors $\Fun{Naive}_k : \CartDiffSys_{\cat{C}} \to \Coalg^\nn_{\cat{C}}$, for $k:\rr$, which we illustrate for a single system $(\rr^n,m)$ over a fixed polynomial $p$, with $m : \rr^n y^{\rr^n\times\rr^n} \to p$ (since $\Fun{T}\rr^n \cong \rr^n\times\rr^n$).
  This system is determined by a pair of morphisms $m_1 : \rr^n \to p(1)$ and $m^\# : \sum_{x:\rr^n} p[m_1(x)] \to \rr^n\times \rr^n$, and we can write the action of $m^\#$ as $(x,y) \mapsto (x, v_x(y))$.

  Using these, we define a discrete-time dynamical system $\beta$ over $p$ with state space $\rr^n$.
  This $\beta$ is given by an output map $\beta^o$, which we define to be equal to $m_1$, $\beta^o := m_1$, and an update map $\beta^u : \sum_{x:\rr^n} p[\beta^o(x)] \to \rr^n$, which we define by $(x,y) \mapsto x + k\, v_x(y)$.
  Together, these define a system in $\Coalg^\nn_{\cat{C}}(p)$, and the collection of these systems $\beta$ produces an indexed functor by the definition $\Fun{Naive}_k(p)(m) := \beta$.

  By contrast, the discrete-time system obtained via $\Fun{Disc}_k \circ \Fun{Flow}$ involves integrating a continuous-time one for $k$ units of real time for each unit of discrete time: although this in general produces a more accurate simulation of the trajectories implied by the vector field, it is computationally more arduous; to trade off simulation accuracy against computational feasibility, one may choose a more sophisticated discretization scheme than that sketched above, or at least choose a ``sufficiently small'' timescale $k$.
\end{rmk}

Finally, we can use the foregoing ideas to translate differential hierarchical inference systems to dynamical hierarchical inference systems.

\begin{cor}
  Let $\CartDiffHier_{\cat{C}}$ denote the restriction of $\DiffHier_{\cat{C}}$ to hom-categories in $\CartDiffSys_{\cat{C}}$.
  The indexed functors $\Fun{Disc}_k : \Coalg^\rr_{\cat{C}} \to \Coalg^\nn_{\cat{C}}$, $\Fun{Flow} : \DiffSys_{\cat{C}} \to \Coalg^\rr_{\cat{C}}$, and $\Fun{Naive}_k : \CartDiffSys_{\cat{C}} \to \Coalg^\nn_{\cat{C}}$ induce functors (respectively) $\H\Fun{Disc}_k : \Hier^\rr_{\cat{C}} \to \Hier^\nn_{\cat{C}}$, $\H\Fun{Flow} : \DiffHier_{\cat{C}} \to \Hier^\rr_{\cat{C}}$ and $\H\Fun{Naive}_k : \CartDiffHier_{\cat{C}} \to \Hier^\nn_{\cat{C}}$ by change of base of enrichment.
\end{cor}

\section{Approximate inference doctrines} \label{sec:doctrines}

We are now in a position to build the bridge between abstract statistical models and the dynamical systems that play them, with the categories of hierarchical dynamical systems developed in the previous section supplying the semantics.
These bridges will be
functors, which we call \textit{approximate inference doctrines}.
In general, they will be functors from categories of parameterized statistical models, whose parameters form part of the dynamical state spaces, and often we are particularly interested in only a particular class of statistical models, which typically form a subcategory of a broader category of stochastic channels.
We therefore make the following definition.

\begin{defn}
  Let $\cat{D}$ be a subcategory of $\eP\cat{C}$.
  An \textit{approximate inference doctrine} for $\cat{D}$ in time $\Tt$ is a functor $\cat{D} \to \Hier^\Tt_{\cat{C}}$.
\end{defn}

Here, $\eP\cat{C}$ denotes the external parameterization of $\cat{C}$, to the definition of which we now turn.

\subsection{External parameterization} \label{sec:doctrines-para}

In the previous instalment of this series, we considered parameterized Bayesian lenses \parencite[\S 3.4]{Smithe2021Compositional1} and statistical games \parencite[Cor. 4.14, Ex. 5.5]{Smithe2021Compositional1}, in order to treat systems with the ability to improve their statistical performance.
Approximate inference doctrines operationalize this improvement, but in this context it is preferable to consider statistical systems that are `externally' rather than `internally' parameterized:
the improvement of the performance is typically a process that is `external' to the solution of the statistical problem (\textit{e.g.}, inference) itself; for instance, learning is often assumed \parencite{Buckley2017free} to take place on a slower timescale than inference.

Technically, we can see this distinction by considering the type of an internally parameterized Bayesian lens, following \textcite[\S 3.4]{Smithe2021Compositional1}.
If $(\gamma,\rho):(A,S)\xlensto{(\Theta,\Omega)}(B,T)$ is such a lens, then its forward channel $\gamma$ has the type $\Theta\otimes A\klto B$, and the backwards channel $\rho$ has the type $\cat{C}(I, \Theta\otimes A)\to\cat{C}(T, \Omega\otimes S)$.
Notice that this means that in general the inversion $\rho$ depends on a joint prior over $\Theta\otimes A$, and produces an updated state over $\Omega\otimes S$, even though one is often interested only in a family of inversions of the type $\cat{C}(I,A)\to\cat{C}(T,S)$ parameterized by $\Omega$, with the updating of the parameters taking place in an external process that `observes' the performance of the statistical game.
We make this distinction formal using the notion of external parameterization.

\begin{defn} \label{def:ext-para}
  Given a category $\cat{C}$ enriched in $(\cat{E},\times,1)$, we define the \textit{external parameterization} $\eP\cat{C}$ of $\cat{C}$ in $\cat{E}$ as the following bicategory.
  0-cells are the objects of $\cat{C}$, and each hom-category $\eP\cat{C}(A,B)$ is given by the slice category $\cat{E}/\cat{C}(A,B)$.
  The composition of 1-cells is by composing in $\cat{C}$ after taking the product of parameters:
  given $f:\Theta\to\cat{C}(A,B)$ and $g:\Omega\to\cat{C}(B,C)$, their composite $g\circ f$ is
  \[ g\circ f := \Omega\times\Theta \xto{g\times f} \cat{C}(B,C)\times\cat{C}(A,B) \xto{\klcirc} \cat{C}(A,C) \]
  where $\klcirc$ is the composition map for $\cat{C}$ in $\cat{E}$.
  The identity 1-cells are the points on the identity morphisms in $\cat{C}$.
  For instance, the identity 1-cell on $A$ is the corresponding point $\id_A : 1\to\cat{C}(A,A)$.
  We will denote 1-cells using our earlier notation for parameterized morphisms: for instance, $f : A\xto{\Theta}B$ and $\id_A : A\xto{1}A$.
  The horizontal composition of 2-cells is given by taking their product.
\end{defn}

As an example, let us consider externally parameterized statistical games.

\begin{ex} \label{ex:ext-para-sgame}
  The category $\PSGame{\cat{C}}$ of externally parameterized statistical games in $\cat{C}$ has as 0-cells pairs of objects in $\cat{C}$ (as in the case of Bayesian lenses or plain statistical games).
  Its 1-cells $(A,S)\xto{\Theta}(B,T)$ are parameterized games, consisting in a choice of parameter space $\Theta$, an externally parameterized lens $f:\Theta\to\BLens{C}((A,S),(B,T))$, and an externally parameterized loss function $\phi:\sum_{\vartheta:\Theta}\cctx(f_\vartheta)\to\rr$.
  The identity on $(A,S)$ is given by the trivially parameterized element $\id_{(A,S)}:1\to\BLens{C}((A,S),(A,S))$, equipped with the zero loss function, as in the case of unparameterized statistical games.
  Given parameterized games $(f,\phi):(A,S)\xto{\Theta}(B,T)$ and $(g,\psi):(B,T)\xto{\Theta'}(C,U)$, we form their composite as follows.
  The composite parameterized lens is given by taking the product of the parameter spaces:
  \[
  \Theta\times\Theta'\xto{f\times g}\BLens{C}\bigl((A,S),(B,T)\bigr)\times\BLens{C}\bigl((B,T),(C,U)\bigr)\xto{\lenscirc}\BLens{C}\bigl((A,S),(C,U)\bigr)
  \]
  The composite fitness function is given accordingly:
  \[
  \sum_{\vartheta:\Theta,\vartheta':\Theta'}\cctx(g_{\vartheta'}\lenscirc f_\vartheta)\xto{\bcopier}\sum_{\vartheta,\vartheta'}\cctx(g_{\vartheta'}\lenscirc f_\vartheta)^2\xto{({g_{\vartheta'}}^*,{f_\vartheta}_*)}\sum_{\vartheta,\vartheta'}\cctx(f_\vartheta)\times\cctx(g_{\vartheta'})\xto{(\phi_\vartheta,\psi_{\vartheta'})}\rr\times\rr\xto{+}\rr
  \]
\end{ex}

For concision, when we say \textit{parameterized statistical game} or \textit{parameterized lens} in the absence of further qualification, we will henceforth mean the externally (as opposed to internally) parameterized versions.

\begin{rmk}
  In prior work, this external parameterization construction has been called `proxying' \parencite{Capucci2021Parameterized}.
  We prefer the more explicit name `external parameterization', reserving `proxying' for a slightly different double-categorical construction to appear in future work.
\end{rmk}

\begin{rmk}
  Before moving on to examples of approximate inference doctrines, let us note the similarity of the notions of external parameterization, differential system, and dynamical system:
  both of the latter can be considered as externally parameterized systems with extra structure, where the extra structure is a morphism or family of morphisms back into (an algebra of) the parameterizing object:
  in the case of differential systems, this `algebra' is the tangent bundle; for dynamical systems, it is trivial; and forgetting this extra structure returns a mere external parameterization.
  Approximate inference doctrines are thus functorial ways of equipping morphisms with this extra structure, and in this respect they are close to the current understanding of general compositional game theory \parencite{Capucci2022Diegetic,Capucci2021Towards}.
\end{rmk}

\subsection{Channels with Gaussian noise} \label{sec:doctrines-gauss}

Our motivating examples from the computational neuroscience literature are defined over a subcategory of channels between Cartesian spaces with additive Gaussian noise \parencite{Friston2007Variational,Bogacz2017tutorial,Buckley2017free}; typically one writes $x \mapsto f(x) + \omega$ for a deterministic map $f : X\to Y$ and $\omega$ sampled from a Gaussian distribution over $Y$.
This choice is made, as we will see, because it permits some simplifying assumptions which mean the resulting dynamical systems resemble known neural circuits.
In this section, we develop the categorical language in which we can express such Gaussian channels.
We begin by introducing the category of probability spaces and measure-preserving maps, which we then use to define channels of the general form $x \mapsto f(x) + \omega$, before restricting to the finite-dimensional Gaussian case.

\begin{defn}
  Let \(\Pa\Cat{\mdash Spc}\) be the category \(\Cat{Comon}\big(1/\Kl(\Pa)\big)\) of probability spaces \((M, \mu)\) with $\mu : 1\klto M$ in $\Kl(\Pa)$ (\textit{i.e.}, $1\to \Pa M$ in $\cat{E}$), and whose morphisms \(f : (M, \mu) \to (N, \nu)\) are measure-preserving maps \(f : M \to N\) (\textit{i.e.}, such that \(f \klcirc \mu = \nu\) in \(\Kl(\Pa)\)).
\end{defn}

We can think of $x \mapsto f(x) + \omega$ as a map parameterized by a noise source, and so to construct a category of such channels, we can use the $\para$ construction in its actegorical form.
We will use the monoidal-actegorical definition of $\para$ given in \textcite[{\S 2.3}]{Smithe2021Compositional1}, following \textcite{Capucci2021Towards}; for a comprehensive reference on actegory theory, see \textcite{Capucci2022Actegories}.
The first step is to spell out the actegory structure.

\begin{prop} \label{prop:probstoch-act}
  Let \(\Pa : \cat{E} \to \cat{E}\) be a probability monad on the symmetric monoidal category \((\cat{E}, \times, 1)\).
  Then there is a \(\Pa\Cat{\mdash Spc}\)-actegory structure \(\ast : \Pa\Cat{\mdash Spc} \to \Cat{Cat}(\cat{E}, \cat{E})\) on \(\cat{E}\) as follows.
  For each \((M, \mu) : \Pa\Cat{\mdash Spc}\), define \((M, \mu) \ast (-) : \cat{E} \to \cat{E}\) by \((M, \mu) \ast X := M \times X\).
  For each morphism \(f : (M, \mu) \to (M', \mu')\) in \(\Pa\Cat{\mdash Spc}\), define \(f \ast X := f \times \id_X\).
  \begin{proof}[Proof sketch]
    The action on morphisms is well-defined because each morphism \(f : M \klto N\) in \(\Cat{Comon}\big(1/\Kl(\Pa)\big)\) corresponds to a map \(f : M \to N\) in \(\cat{E}\); it is clearly functorial.
    The unitor and associator are inherited from the Cartesian monoidal structure \((\times, 1)\) on \(\cat{E}\).
\end{proof}
\end{prop}

The resulting $\para$ bicategory, $\para(\ast)$, can be thought of as a bicategory of maps each of which is equipped with an independent noise source; the composition of maps takes the product of the noise sources, and 2-cells are noise-source reparameterizations.
The actegory structure $\ast$ is symmetric monoidal, and the 1-categorical truncation $\para(\ast)_1$ \parencite[{Prop. 2.47}]{Smithe2021Compositional1} is a copy-delete category \parencite[{Def. 2.2}]{Cho2017Disintegration} (also \parencite[{Def. 2.20}]{Smithe2021Compositional1}) as we now sketch.

\begin{prop}
  Consider the actegory structure \(\ast\) of Proposition \ref{prop:probstoch-act}.
  Then \(\para(\ast)_1\) is a copy-delete category.
  \begin{proof}[Proof sketch]
    The monoidal structure is defined following Proposition 2.44 of \textcite{Smithe2021Compositional1}.
    We need to define a right costrength $\rho$ with components $(N,\nu)\ast(X\times Y) \xto{\sim} X\times((N,\nu)\ast Y)$.
    Since $\ast$ is defined by forgetting the probability structure and taking the product, the costrength is given by the associator and symmetry in $\cat{E}$:
    \[ (N,\nu)\ast(X\times Y) = N\times(X\times Y) \xto{\sim} N\times(Y\times X) \xto{\sim} (N\times Y)\times X \xto{\sim} X\times(N\times Y) = X\times((N,\nu)\ast Y) \]
    It is clear that this definition gives a natural isomorphism; the rest of the monoidal structure follows from that of the product on $\cat{E}$.

    We now need to define a symmetry natural isomorphism $\beta_{X,Y}:X\times Y \xto{\sim} Y\times X$ in $\para(\ast)$.
    This is given by the symmetry of the product in $\cat{E}$, under the embedding of $\cat{E}$ in $\para(\ast)$ that takes every map to its parameterization by the terminal probability space.

    The rest of the copy-delete structure is inherited similarly from \(\cat{E}\).
  \end{proof}
\end{prop}

If we think of $\Kl(\Pa)$ as a canonical category of stochastic channels, for $\para(\ast)_1$ to be considered as a subcategory of Gaussian channels, we need the following result.

\begin{prop} \label{prop:embed-para-ast-kl}
  There is an identity-on-objects strict monoidal embedding of \(\para(\ast)_1\) into \(\Kl(\Pa)\).
  Given a morphism \(f : X \xto{(\Omega, \mu)} Y\) in \(\para(\ast)_1\), form the composite \(f \klcirc (\mu, \id_X) : X \klto Y\) in \(\Kl(\Pa)\).
  \begin{proof}[Proof sketch]
    First, the given mapping preserves identities: the identity in $\para(\ast)$ is trivially parameterized, and is therefore taken to the identity in $\Kl(\Pa)$.
    The mapping also preserves composites, by the naturality of the unitors of the symmetric monoidal structure on $\Kl(\Pa)$.
    That is, given $f : X\xto{(\Omega,\mu)}Y$ and $g : Y\xto{(\Theta,\nu)}Z$, their composite $g\circ f:X\xto{(\Theta\otimes\Omega,\nu\otimes\mu)}Z$ is taken to
    \[ X \xklto{\sim} 1\otimes 1\otimes X \xklto{\nu\otimes\nu\otimes\id_X} \Theta\otimes\Omega\otimes X \xklto{g\circ f} Z \]
    where here $g\circ f$ is treated as a morphism in $\Kl(\Pa)$.
    Composing the images of $g$ and $f$ under the given mapping gives
    \[ X \xklto{\sim} 1\otimes X \xklto{\mu\otimes\id_X} \Omega\otimes X \xklto{f} Y \xklto{\sim} 1\otimes Y \xklto{\nu\otimes Y} \Theta\otimes Y \xklto{g} Z \]
    which is equal to
    \[ X \xklto{\sim} 1\otimes 1\otimes X \xklto{\nu\otimes\mu\otimes\id_X} \Theta\otimes\Omega\otimes X \xklto{\id_\Theta\otimes f} \Theta\otimes Y \xklto{g} Z \]
    which in turn is equal to the image of the composite above.

    The given mapping is therefore functorial.
    To show that it is an embedding is to show that it is faithful and injective on objects.
    Since $\para(\ast)$ and $\Kl(\Pa)$ have the same objects, the embedding is trivially identity-on-objects (and hence injective); it is similarly easy to see that it is faithful, as distinct morphisms in $\para(\ast)$ are mapped to distinct morphisms in $\Kl(\Pa)$.

    Finally, since the embedding is identity-on-objects and the monoidal structure on $\para(\ast)$ is inherited from that on $\Kl(\Pa)$ (producing identical objects), the embedding is strict monoidal.
  \end{proof}
\end{prop}

We now restrict our attention to Gaussian maps.

\begin{defn} \label{def:gauss}
  We say that \(f : X \klto Y\) in \(\Kl(\Pa)\) is \textbf{Gaussian} if, for any \(x : X\), the state \(f(x) : \Pa Y\) is Gaussian\footnote{We admit Dirac delta distributions, and therefore deterministic channels, as Gaussian, since delta distributions can be seen as Gaussians with infinite precision.}.
  Similarly, we say that \(f : X \xto{(\Omega, \mu)} Y\) in \(\para(\ast)\) is Gaussian if its image under the embedding \(\para(\ast)_1 \hookrightarrow \Kl(\Pa)\) is Gaussian.
  Given a category of stochastic channels \(\cat{C}\), write \(\gauss(\cat{C})\) for the subcategory generated by Gaussian morphisms and their composites in \(\cat{C}\).
  Given a separable Banach space $X$, write $\gauss(X)$ for the space of Gaussian states on $X$.
\end{defn}

\begin{ex}
  A class of examples of Gaussian morphisms in $\para(\ast)$ that will be of interest to us in section \ref{sec:doctrines-hebb-laplace} is of the form $x \mapsto f(x) + \omega$ for some map $f:X\to Y$ and $\omega$ distributed according to a Gaussian distribution over $Y$.
  Writing $\E[\omega]$ for the mean of this distribution, the resulting channel in $\Kl(\Pa)$ emits for each $x:X$ a Gaussian distribution with mean $f(x)+\E[\omega]$ and variance the same as that of $\omega$.
\end{ex}

\begin{rmk} \label{rmk:gauss-not-closed}
  In general, Gaussian morphisms are not closed under composition:
  pushing a Gaussian distribution forward along a nonlinear transformation will not generally result in another Gaussian.
  For instance, consider the Gaussian morphisms $x\mapsto f(x) + \omega$ and $y\mapsto g(y) + \omega'$.
  Their composite in $\para(\ast)$ is the morphism $x\mapsto g\bigl(f(x)+\omega)\bigr)+\omega'$; even if $g\bigl(f(x)+\omega)\bigr)$ is Gaussian-distributed, the sum of two Gaussians is in general not Gaussian, and so $g\bigl(f(x)+\omega)\bigr)+\omega'$ will not be Gaussian.
  This non-closure underlies the power of statistical models such as the variational autoencoder, which are often constructed by pushing a Gaussian forward along a learnt nonlinear transformation \parencite{Kingma2017Variational}, in order to approximate an unknown distribution;
  since sampling from Gaussians is relatively straightforward, this method of approximation can be computationally tractable.
  The $\gauss$ construction here is an abstraction of the Gaussian-preserving transformations invoked by \textcite{Shiebler2020Categorical}, and is to be distinguished from the category $\gauss$ introduced by \textcite{Fritz2019synthetic}, whose morphisms are affine transformations (which do preserve Gaussianness) and which are therefore closed under composition;
  there is nonetheless an embedding of Fritz's $\gauss$ into our $\gauss\bigl(\Kl(\Pa)\bigr)$.
\end{rmk}

\begin{prop}
  Let \(\FdCartSpc(\cat{E})\) denote the full subcategory of \(\cat{E}\) spanned by finite-dimensional Cartesian spaces \(\rr^n\), where \(n : \nn\).
  Let \(\PFdCartSpc\) denote the corresponding subcategory of \(\PSpc\).
  Let \(\star : \PFdCartSpc \to \Cat{Cat}\big(\FdCartSpc(\cat{E}), \FdCartSpc(\cat{E})\big)\) be the corresponding restriction of the monoidal action \(\ast : \Pa\Cat{\mdash Spc} \to \Cat{Cat}(\cat{E}, \cat{E})\) from Proposition \ref{prop:probstoch-act}.
  Then \(\para(\star)\) is a monoidal subbicategory of \(\para(\ast)\).
\end{prop}

We will write \(\Pa_{\mathbf{Fd}} : \FdCartSpc(\cat{E}) \to \FdCartSpc(\cat{E})\) to denote the restriction of the probability monad \(\Pa : \cat{E} \to \cat{E}\) to \(\FdCartSpc(\cat{E})\).

Finally, we give the density function representation of Gaussian channels in $\Kl(\Pa_{\mathbf{Fd}})$.

\begin{prop} \label{prop:gauss-stat-param}
  Every Gaussian channel \(c : X \klto Y\) in \(\Kl(\Pa_{\mathbf{Fd}})\) admits a density function \(p_c : Y \times X \to [0, 1]\) with respect to the Lebesgue measure on \(Y\).
  Moreover, since \(Y = \rr^n\) for some \(n : \nn\), this density function is determined by two maps: the \textit{mean} \(\mu_c : X \to \rr^n\), and the \textit{covariance} \(\Sigma_c : X \to \rr^{n\times n}\) in \(\cat{E}\).
  We call the pair \((\mu_u, \Sigma_c) : X \to \rr^n \times \rr^{n\times n}\) the \textit{statistical parameters} for \(c\).
  \begin{proof}
    The density function \(p_c : Y \times X \to [0, 1]\) satisfies
    \[
    \log p_c (y | x) = \frac{1}{2} \innerprod{\epsilon_c}{{\Sigma_c(x)}^{-1} {\epsilon_c}} - \log \sqrt{(2 \pi)^n \det \Sigma_c(x) }
    \]
    where \(\epsilon_c : Y \times X \to Y : (y, x) \mapsto y - \mu_c(x)\).
  \end{proof}
\end{prop}

\subsection{The Laplace doctrine} \label{sec:doctrines-laplace}

Our first example of a doctrine arises in the computational neuroscience literature, which has sought to explain the apparently `predictive' nature of sensory cortical circuits using ideas from the theory of approximate inference \parencite{Bastos2012Canonical}; the general name for this neuroscientific theory is \textit{predictive coding}, and the task of a predictive coding model is to define a dynamical system whose structures and behaviours mimic those observed in neural circuits \textit{in vivo}.
One way to satisfy this constraint is to describe a procedure that turns a statistical problem into a dynamical system of a form known to be simulable by a neural circuit:
that is to say, there are certain classes of dynamical systems which are known to reproduce the phenomenology of neural circuits and which are built out of parts that correspond to known biological structures, and so a ``biologically plausible'' model of predictive coding should produce an instance of such a class, given a statistical problem.

This procedure pushes the `plausibility' constraint back to the level of the statistical problem (since there are presently no known neural circuit models that can solve any inference problem in general), and one restriction that is usefully made is that all noise sources in the model are Gaussian.
This restriction allows us to make an approximation, known as the \textit{Laplace approximation}, to the loss function of an autoencoder game which in turn entails that performing stochastic gradient descent on this loss function (with respect to the mean of the posterior distribution) generates a dynamical system that is biologically plausible (up to some level of biological plausibility) \parencite{Bogacz2017tutorial,Bastos2012Canonical}.

In this section, we begin by defining the Laplace approximation and the resulting dynamical system, and go on to show both how it arises and how the procedure is functorial: that is, we show that it constitutes an approximate inference doctrine, and describe how this presentation clarifies the role of what has been called the ``mean field'' assumption in earlier literature \parencite{Friston2007Variational}.
(We leave the study of the biological plausibility of compositional dynamical systems for future work.)

\begin{lemma}[Laplace approximation] \label{lemma:laplace-approx}
Suppose:
\begin{enumerate}
\item \((\gamma, \rho, \phi) : (X,X) \to (Y,Y)\) is a simple \(D_{KL}\)-autoencoder game with Gaussian channels between finite-dimensional Cartesian spaces;
\item for all priors \(\pi : \gauss(X)\), the statistical parameters of \(\rho_\pi : Y \to \Pa X\) are denoted \((\mu_{\rho_\pi}, \Sigma_{\rho_\pi}) : Y \to \rr^{|X|} \times \rr^{|X|\times |X|}\), where \(|X|\) is the dimension of \(X\); and
\item for all \(y : Y\), the eigenvalues of \(\Sigma_{\rho_\pi}(y)\) are small.
\end{enumerate}
Then the loss function \(\phi:\cctx(\gamma,\rho)\to\rr\) can be approximated by
\[
\phi(\pi, k) = \E_{y\sim\efb{\pi}{\gamma}{k}} \big[ \Fa(y) \big]
\approx \E_{y\sim\efb{\pi}{\gamma}{k}} \big[ \Fa^L(y) \big]
\]
where
\begin{align} \label{eq:laplace-energy}
\Fa^L(y)
& = E_{(\pi,\gamma)}\left(\mu_{\rho_\pi}(y), y\right) - S_X \left[ \rho_\pi(y) \right] \\
& = -\log p_\gamma(y|\mu_{\rho_\pi}(y)) -\log p_\pi(\mu_{\rho_\pi}(y)) - S_X \left[ \rho_\pi(y) \right] \nonumber
\end{align}
where \(S_x[\rho_\pi(y)] = \E_{x \sim \rho_\pi(y)} [ -\log p_{\rho_\pi}(x|y) ]\) is the Shannon entropy of \(\rho_\pi(y)\), and \(p_\gamma : Y \times X \to [0,1]\), \(p_\pi : X \to [0,1]\), and \(p_{\rho_\pi} : X \times Y \to [0,1]\) are density functions for \(\gamma\), \(\pi\), and \(\rho_\pi\) respectively.  The approximation is valid when \(\Sigma_{\rho_\pi}\) satisfies
\begin{equation} \label{eq:laplace-sigma-rho-pi}
\Sigma_{\rho_\pi} (y) = \left(\partial_x^2 E_{(\pi,\gamma)}\right)\left( \mu_{\rho_\pi}(y), y\right)^{-1} \, .
\end{equation}
We call $\Fa^L$ the \textit{Laplacian free energy} and $E_{(\pi,\gamma)}$ the corresponding \textit{Laplacian energy}.

\begin{proof}
Following Proposition \ref{prop:gauss-stat-param}, we can write the density functions as:
\begin{gather} \label{eq:laplace-gauss}
\log p_\gamma (y | x) = \frac{1}{2} \innerprod{\epsilon_\gamma}{{\Sigma_\gamma}^{-1} {\epsilon_\gamma}} - \log \sqrt{(2 \pi)^{|Y|} \det \Sigma_\gamma } \nonumber \\
\log p_{\rho_\pi} (x | y) = \frac{1}{2} \innerprod{\epsilon_{\rho_\pi}}{{\Sigma_{\rho_\pi}}^{-1} {\epsilon_{\rho_\pi}}} - \log \sqrt{(2 \pi)^{|X|} \det \Sigma_{\rho_\pi} } \\
\log p_\pi (x) = \frac{1}{2} \innerprod{\epsilon_\pi}{{\Sigma_\pi}^{-1} {\epsilon_\pi}} - \log \sqrt{(2 \pi)^{|X|} \det \Sigma_\pi } \nonumber
\end{gather}
where for clarity we have omitted the dependence of \(\Sigma_\gamma\) on \(x\) and \(\Sigma_{\rho_\pi}\) on \(y\), and where
\begin{align} \label{eq:laplace-error}
\epsilon_\gamma : Y \times X \to Y & : (y, x) \mapsto y - \mu_\gamma(x) \, , \nonumber \\
\epsilon_{\rho_\pi} : X \times Y \to X & : (x, y) \mapsto x - \mu_{\rho_\pi}(y) \, , \\
\epsilon_\pi : X \times 1 \to X & : (x, \ast) \mapsto x - \mu_\pi \, . \nonumber
\end{align}
Then, recall from \parencite[Remark 5.12]{Smithe2021Compositional1} that we can write the free energy \(\Fa(y)\) as the difference between expected energy and entropy:
\begin{align*}
  \Fa(y)
  &= \E_{x \sim \rho_\pi(y)} \left[ \log \frac{p_{\rho_\pi}(x|y)}{p_\gamma(y|x) \cdot p_\pi(x)} \right] \\
  &= \E_{x \sim \rho_\pi(y)} \left[ - \log p_\gamma(y|x) - \log p_\pi (x) \right]
     - S_X \left[ \rho_\pi(y) \right] \\
  &= \E_{x \sim \rho_\pi(y)} \left[ E_{(\pi,\gamma)}(x,y) \right] - S_X \left[ \rho_\pi(y) \right]
\end{align*}
Next, since the eigenvalues of \(\Sigma_{\rho_\pi}(y)\) are small for all \(y : Y\), we can approximate the expected energy by its second-order Taylor expansion around the mean \(\mu_{\rho_\pi}(y)\):
\begin{align*}
\Fa(y) \approx & \;
E_{(\pi,\gamma)}(\mu_{\rho_\pi}(y), y)
+ \frac{1}{2} \innerprod{\epsilon_{\rho_\pi}\left( \mu_{\rho_\pi}(y), y\right)}{\left(\partial_x^2 E_{(\pi,\gamma)}\right)\left( \mu_{\rho_\pi}(y), y\right) \cdot \epsilon_{\rho_\pi}\left( \mu_{\rho_\pi}(y), y\right)} \\
& \;\, - S_X \big[ \rho_\pi(y) \big] \, .
\end{align*}
where \(\left(\partial_x^2 E_{(\pi,\gamma)}\right)\left( \mu_{\rho_\pi}(y), y\right)\) is the Hessian of \(E_{(\pi,\gamma)}\) with respect to \(x\) evaluated at \((\mu_{\rho_\pi}(y), y)\).

Note that
\begin{equation} \label{eq:laplace-trace-sigma}
\innerprod{\epsilon_{\rho_\pi}\left( \mu_{\rho_\pi}(y), y\right)}{\left(\partial_x^2 E_{(\pi,\gamma)}\right)\left( \mu_{\rho_\pi}(y), y\right) \cdot \epsilon_{\rho_\pi}\left( \mu_{\rho_\pi}(y), y\right)}
=
\tr \left[ \left(\partial_x^2 E_{(\pi,\gamma)}\right)\left( \mu_{\rho_\pi}(y), y\right) \, \Sigma_{\rho_\pi}(y) \right] \, ,
\end{equation}
that the entropy of a Gaussian measure depends only on its covariance,
\[
S_X \big[ \rho_\pi(y) \big]
= \frac{1}{2} \log \det \left( 2 \pi \, e \, \Sigma_{\rho_\pi}(y) \right) \, ,
\]
and that the energy \(E_{(\pi,\gamma)}(\mu_{\rho_\pi}(y), y)\) does not depend on \(\Sigma_{\rho_\pi}(y)\). We can therefore write down directly the covariance \(\Sigma_{\rho_\pi}^\ast(y)\) minimizing \(\Fa(y)\) as a function of \(y\). We have
\[
\partial_{\Sigma_{\rho_\pi}} \Fa(y) \approx
\frac{1}{2} \left(\partial_x^2 E_{(\pi,\gamma)}\right)\left( \mu_{\rho_\pi}(y), y\right)
+ \frac{1}{2} {\Sigma_{\rho_\pi}}^{-1} \, .
\]
Setting \(\partial_{\Sigma_{\rho_\pi}} \Fa(y) = 0\), we find the optimum as expressed by equation \eqref{eq:laplace-sigma-rho-pi}
\begin{equation*} \label{eq:laplace-sigma-opt}
  \Sigma_{\rho_\pi}^\ast (y) = \left(\partial_x^2 E_{(\pi,\gamma)}\right)\left( \mu_{\rho_\pi}(y), y\right)^{-1} \, .
\end{equation*}
Finally, on substituting \(\Sigma_{\rho_\pi}^\ast (y)\) in equation \eqref{eq:laplace-trace-sigma}, we obtain the desired expression of equation \eqref{eq:laplace-energy}
\[
\Fa(y) \approx E_{(\pi,\gamma)}\left(\mu_{\rho_\pi}(y), y\right) - S_X \left[ \rho_\pi(y) \right] =: \Fa^L(y) \, .
\]%
\end{proof}
\end{lemma}

\begin{rmk} \label{rmk:laplace-err}
  The terms $\epsilon_\gamma:Y\times X\to Y$ (\textit{\&c.}) of eq. \eqref{eq:laplace-error} are known as \textit{error functions}, since they encode the difference between $y:Y$ and the expected element $\mu_\gamma(x):Y$ given $x:X$.
  In applications, one often thinks of these errors as \textit{prediction errors}, interpreting $\mu_\gamma$ as the system's prediction of the expected state of $Y$.

  In this context one then also defines the \textit{precision-weighted} errors
  \begin{equation} \label{eq:precis-err}
    \eta_\gamma(y, x) := {\Sigma_\gamma(x)}^{-1} \epsilon_\gamma(y, x) : Y\times X\to Y \, ,
  \end{equation}
  noting that the inverse covariance matrix ${\Sigma_\gamma(x)}^{-1}$ can be interpreted as encoding the `precision' of a belief: roughly speaking, low variance (or `diffusivity') means high precision\footnote{
  Consider the one-dimensional case: as the variance $\sigma$ of a normal distribution tends to 0, the distribution approaches a Dirac delta distribution, which is ``infintely precise''.}.
  The log-densities of eq. \eqref{eq:laplace-gauss} are then understood as measuring the precision-weighted length of the error vectors.
\end{rmk}

\begin{defn} \label{def:laplace-doc}
  Suppose $\gamma:X\klto Y$ is a Gaussian channel in $\Kl(\Pa)$. %
  Then the discrete-time Laplace doctrine defines a system $\Laplace(\gamma):(X,X)\to(Y,Y)$ in $\Hier^\nn_{\gauss(\Kl(\Pa_{\mathbf{Fd}}))}$ as follows (using the representation of Proposition \ref{prop:unpack-hier}).
  \begin{itemize}
  \item The state space is $X$;
  \item the forwards output map $\Laplace(\gamma)^o_1:X\times X\to \gauss(Y)$ is given by $\gamma$:
    \[ \Laplace(\gamma)^o_1 := X\times X \xto{\mathsf{proj}_2} X \xto{\gamma} \gauss(Y) \]
  \item the backwards output map $\Laplace(\gamma)^o_2:X\times \gauss(X)\times Y\to \gauss(X)$ is given by:
    \begin{equation} \label{eq:laplace-doc-rho}
      \begin{aligned}
        \Laplace(\gamma)^o_2 : X\times \gauss(X) \times Y &\to \rr^{|X|}\times \rr^{|X|\times|X|} \hookrightarrow \gauss(X) \\
        (x,\pi,y)&\mapsto\bigl(x,\Sigma_\rho(x,\pi,y)\bigr)
      \end{aligned}
    \end{equation}
    where the inclusion picks the Gaussian state with the given statistical parameters, whose covariance $\Sigma_\rho(x,\pi,y) := \left(\partial_x^2 E_{(\pi,\gamma)}\right)(x, y)^{-1}$ is defined following equation \eqref{eq:laplace-sigma-rho-pi} (Lemma \ref{lemma:laplace-approx});
  \item the update map $\Laplace(\gamma)^u:X\times \gauss(X)\times Y \to \gauss(X)$ returns a point distribution on the updated mean
    \begin{align*}
      \Laplace(\gamma)^u : X \times \gauss(X) \times Y & \to \gauss(X) \\
      (x, \pi, y) & \mapsto \eta^\Pa_X\bigl(\mu_\rho(x,\pi,y)\bigr)
    \end{align*}
    where $\eta^\Pa_X:X\to\gauss(X)$ denotes the unit of the monad $\Pa$ and $\mu_\rho$ is defined by
    \[
    \mu_\rho(x,\pi,y) := x + \lambda\, \partial_x \mu_\gamma (x)^T \eta_\gamma(y,x) - \lambda\, \eta_\pi(x) \, .
    \]
    Here, the precision-weighted error terms $\eta$ are as in equation \eqref{eq:precis-err} (Remark \ref{rmk:laplace-err}), and $\lambda:\rr_+$ is some choice of `learning rate'.
  \end{itemize}
\end{defn}

\begin{rmk}
  Note that the update map $\Laplace(g)^u$ as defined here is actually deterministic, in the sense that it is defined as a deterministic map followed by the unit of the probability monad.
  However, the general stochastic setting is necessary, because the composition of system depends on the composition of Bayesian lenses, which is necessarily stochastic.
\end{rmk}

\begin{defn} \label{def:laplace-game}
  A \textit{Laplacian statistical game} is a parameterized statistical game $(\gamma,\rho,\phi):(X,X)\xto{X}(Y,Y)$ satisfying the following conditions:
  \begin{enumerate}
  \item $X$ and $Y$ are finite-dimensional Cartesian spaces;
  \item the forward channel $\gamma$ is an unparameterized Gaussian channel; %
  \item the backward channel $\rho$ is parameterized by $X$ and defined as the backwards output map of the Laplace doctrine (equation \eqref{eq:laplace-doc-rho} of Definition \ref{def:laplace-doc}); that is,
    \begin{align*}
      \rho : X\times\gauss(X)\times Y&\to\rr^{|X|}\times\rr^{|X|\times|X|}\hookrightarrow\gauss(X) \\
      (x,\pi,y)&\mapsto\bigl(x,\Sigma_\rho(x,\pi,y)\bigr)
    \end{align*}
    where the inclusion picks the Gaussian with mean $x$ and $\Sigma_\rho(x,\pi,y) = \left(\partial_x^2 E_{(\pi,\gamma)}\right)(x, y)^{-1}$;
  \item the loss function $\phi:\sum_{x:X}\cctx\bigl(\gamma,\rho_x\bigr)\to\rr$ is given for each $x:X$ by $\phi_x(\pi,k) = \E_{y\sim\efb{\pi}{\gamma}{k}} \big[ \Fa^L(y) \big]$, where $\Fa^L$ is the Laplacian free energy
    \begin{align*}
      \Fa^L(y)
      =&\; E_{(\pi,\gamma)}\left(x, y\right) - S_X \bigl[ \rho(x,\pi,y) \bigr] \\
      =& -\log p_\gamma(y|x) -\log p_\pi(x) - S_X \bigl[ \rho(x,\pi,y) \bigr]
    \end{align*}
    as defined in equation \eqref{eq:laplace-energy} of Lemma \ref{lemma:laplace-approx}.
  \end{enumerate}
  (By ``unparameterized channel'', we mean a channel parameterized by the trivial space $1$; the pair $(\gamma,\rho)$ constitutes a parameterized Bayesian lens with parameter space $X$, where the choice of $\gamma$ simply forgets the parameter, discarding it along the universal map $X\to 1$.)
\end{defn}

\begin{prop} \label{prop:laplace-sgd}
  Given a Laplacian statistical game $(\gamma,\rho,\phi):(X,X)\to(Y,Y)$, $\Laplace(\gamma)$ is obtained by stochastic gradient descent of the loss function $\phi$ with respect to the mean $x$ of the posterior $\rho(x,\pi,y)$.
  \begin{proof}
    We have $\phi_x(\pi,k) = \E_{y\sim\efb{\pi}{\gamma}{k}} \big[ \Fa^L(y) \big]$, where
    \[ \Fa^L(y) = -\log p_\gamma(y|x) -\log p_\pi(x) - S_X \bigl[ \rho(x,\pi,y) \bigr] \, . \]

    Since the entropy \(S_X \left[ \rho_\pi(y) \right]\) depends only on the variance \(\Sigma_\rho(x,\pi,y)\), to optimize the mean \(x\) it suffices to consider only the energy \(E_{(\pi,\gamma)}(x, y)\). We have
    \begin{align*}
      E_{(\pi,\gamma)}(x, y)
      &= - \log p_\gamma(y|x) - \log p_\pi(x) \\
      &= - \frac{1}{2} \innerprod{\epsilon_\gamma(y,x)}{{\Sigma_\gamma(x)}^{-1} {\epsilon_\gamma(y,x)}}
      - \frac{1}{2} \innerprod{\epsilon_\pi(x)}{{\Sigma_\pi}^{-1} {\epsilon_\pi(x)}} \\
      &\quad
      + \log \sqrt{(2 \pi)^{|Y|} \det \Sigma_\gamma(x) }
      + \log \sqrt{(2 \pi)^{|X|} \det \Sigma_\pi }
    \end{align*}
    and a straightforward computation shows that
    \[
    \partial_{x} E_{(\pi,\gamma)}(x, y)
    = - \partial_x \mu_\gamma (x)^T {\Sigma_\gamma(x)}^{-1} \epsilon_\gamma(y,x) + {\Sigma_\pi}^{-1} \epsilon_\pi(x) \, .
    \]
    We can therefore rewrite the mean parameter $\mu_\rho(x,\pi,y)$ emitted by the update map $\Laplace(\gamma)^u$ as
    \begin{align*}
      \mu_\rho(x,\pi,y) &= x + \lambda\, \partial_x \mu_\gamma (x)^T \eta_\gamma(y,x) - \lambda\, \eta_\pi(x) \\
      &= x -\lambda\, \partial_x E_{(\pi,\gamma)}(x, y) \\
      &= x -\lambda\, \partial_x \Fa^L(y)
    \end{align*}
    where the last equality holds because the entropy does not depend on $x$.
    This shows that $\Laplace(\gamma)^u$ descends the gradient of the Laplacian energy with respect to $x$.

    To see then that $\Laplace(\gamma)^u$ performs stochastic gradient descent of $\phi$, note that in the dynamical semantics, the input $y:Y$ is supplied by the context.
    In $\Hier^\Tt_{\gauss(\Kl(\Pa_{\mathbf{Fd}}))}$, the dynamics in the context are stochastic, meaning that each $y:Y$ is in general sampled from a random variable valued in $Y$.
    If we fix the context to sample $y$ from $\efb{\pi}{\gamma}{k}$ then, for a given $x:X$, the expected trajectory of $\mu_\rho$ is given by
    \begin{align*}
      &  \E_{y\sim\efb{\pi}{\gamma}{k}} \big[ \mu_\rho(x,\pi,y) \big] \\
      &= \E_{y\sim\efb{\pi}{\gamma}{k}} \big[ x -\lambda\, \partial_x \Fa^L(y) \big] \\
      &= x - \lambda\, \partial_x \E_{y\sim\efb{\pi}{\gamma}{k}} \big[ \Fa^L(y) \big]
      \quad\quad \text{by linearity of expectation} \\
      &= x - \lambda\, \partial_x \phi_x(\pi,k) \, .
    \end{align*}
    Since $\efb{\pi}{\gamma}{k}$ is just a placeholder for the random variable from which $y$ is sampled, this establishes the result.
  \end{proof}
\end{prop}

Using the preceding proposition, we obtain the following theorem, expressing the Laplacian statistical games in the image of an approximate inference doctrine.

\begin{thm} \label{thm:laplace-functorial}
  Let $\cat{G}$ denote the subcategory of $\PSGame{\Kl(\Pa_{\mathbf{Fd}})}$ generated by Laplacian statistical games $(\gamma,\rho,\phi):(X,X)\xto{X}(Y,Y)$ and by the structure morphisms of a monoidal category.

  Then $\Laplace$ extends to a strict monoidal functor $\gauss(\Kl(\Pa_{\mathbf{Fd}}))\hookrightarrow\cat{G}\to\Hier^\nn_{\gauss(\Kl(\Pa_{\mathbf{Fd}}))}$, where the first factor is the embedding taking any such $\gamma$ to the corresponding Laplacian game, and the second factor performs stochastic gradient descent of loss functions with respect to their external parameterization.
\end{thm}

It helps to separate the proof of the theorem from the proof of the following lemma.

\begin{lemma} \label{lemma:gauss-embeds}
  There is an identity-on-objects strict monoidal embedding of $\gauss(\Kl(\Pa_{\mathbf{Fd}}))$ into $\cat{G}$.
  \begin{proof}
    The structure morphisms of $\gauss(\Kl(\Pa_{\mathbf{Fd}}))$ are mapped to the (trivially parameterized) structure morphisms of $\cat{G}$, and any Gaussian channel $\gamma:X\klto Y$ is mapped to the unique Laplacian statistical game with $\gamma$ as the (unparameterized) forward channel, and the (parameterized) backward channel and loss function determined by the definition of Laplacian statistical game.
    It is clear that this definition gives a faithful functor, and thus an embedding.
    Since it preserves explicitly the monoidal structure, it is also strict monoidal.
  \end{proof}
\end{lemma}

\begin{proof}[Proof of Theorem \ref{thm:laplace-functorial}]
  Thanks to Lemma \ref{lemma:gauss-embeds}, we now turn to the functor $\cat{G}\to\Hier^\nn_{\gauss(\Kl(\Pa_{\mathbf{Fd}}))}$, which we will also denote by $\Laplace$; the composite functor is obtained by pulling this functor $\cat{G}\to\Hier^\nn_{\gauss(\Kl(\Pa_{\mathbf{Fd}}))}$ back along the embedding $\gauss(\Kl(\Pa_{\mathbf{Fd}}))\hookrightarrow\cat{G}$.

  Suppose then that $g:=(\gamma,\rho,\phi):(X,X)\xto{X}(Y,Y)$ is a Laplacian statistical game.
  Proposition \ref{prop:laplace-sgd} tells us that $\Laplace(g)$ is obtained by stochastic gradient descent of the loss function $\phi$ with respect to the mean parameter of the backwards channel $\rho$.
  By definition of $\rho$, this mean parameter is given precisely by the external parameterization, and so we have that $\Laplace(g)$ is obtained by stochastic gradient descent of $\phi$ with respect to this parameterization.

  To extend $\Laplace$ to a functor accordingly, we need to check that performing stochastic gradient descent with respect to the external parameterization preserves identities and composition.
  First we note that, following Definition \ref{def:laplace-doc}, the dynamical systems in the image of $\Laplace$ emit lenses by filling in the parameterization with the dynamical state, and by the preceding remarks, update the state by stochastic gradient descent.
  Next, note that identity parameterized lenses are trivially parameterized, so there is no parameter to `fill in', and no state to update; similarly, the loss function of an identity game is the constant function on $0$, and therefore has zero gradient.
  On identity games $(X,X)\xto{1}(X,X)$, therefore, $\Laplace$ returns the system with trivial state space $1$ that constantly outputs the identity lens $(X,X)\lensto(X,X)$: but this is just the identity on $(X,X)$ in $\Hier^\nn_{\gauss(\Kl(\Pa_{\mathbf{Fd}}))}$, so $\Laplace$ preserves identities.

  We now consider composites.
  Suppose $h:=(\delta,\sigma,\psi):(Y,Y)\xto{Y}(Z,Z)$ is another Laplacian game satisfying the hypotheses of the theorem.
  Since $\Hier^\nn_{\gauss(\Kl(\Pa_{\mathbf{Fd}}))}$ is a bicategory, we need to show that $\Laplace(h)\circ\Laplace(g) \cong \Laplace(h\circ g)$.
  In fact, we will show the stronger result that $\Laplace(h)\circ\Laplace(g) = \Laplace(h\circ g)$, which means demonstrating equalities between the state spaces, output maps, and update maps of the systems on the left- and right-hand sides.

  On state spaces, the equality obtains since the composition of externally parameterized games (Example \ref{ex:ext-para-sgame}) returns a game whose parameter space is the product of the parameter spaces of the factors.
  Similarly, composition of systems in $\Hier^\nn_{\gauss(\Kl(\Pa_{\mathbf{Fd}}))}$ (after Definition \ref{def:hier-bicat}) returns a system whose state space is the product of the state spaces of the factors.
  Finally, $\Laplace$ acts by taking parameter spaces to state spaces, and we have $X\times Y = X\times Y$.

  Next, we note that the output of a composite system in $\Hier^\nn_{\gauss(\Kl(\Pa_{\mathbf{Fd}}))}$ is given by composing the outputs of the factors.
  This is the same as the output returned by $\Laplace$ on a composite game, since outputs in the image of $\Laplace$ just fill in the external parameter using the dynamical state.
  Therefore $\bigl(\Laplace(h)\circ\Laplace(g)\bigr)^o = \Laplace(h\circ g)^o$.

  We now consider the update maps, beginning by computing $\Laplace(h\circ g)^u$.
  The state space is $X\times Y$ and $h\circ g$ has type $(X,X)\xto{X\times Y}(Z,Z)$, so $\Laplace(h\circ g)^u$ has type $X\times Y\times \gauss(X)\times Z\to\gauss(X\times Y)$.
  Following Example \ref{ex:ext-para-sgame}, the composite loss function $(\psi\phi):\sum_{\mu_\rho:X,\mu_\sigma:Y}\cctx(h_{\mu_\sigma}\lenscirc g_{\mu_\rho})\to\rr$ is given by:
  \begin{align*}
    (\psi\phi)(\mu_\rho,\mu_\sigma,\pi,k)
    = & \E_{y\sim\sigma(\mu_\sigma)_{\gamma\klcirc\pi_X}\klcirc\efb{\pi}{\gamma}{\delta^*k}}\left[\Fa^L\bigl(\rho(\mu_\rho)_{\pi_X},\gamma;\pi_X,y\bigr)\right] \\
    & + \E_{z\sim\efb{(M\otimes\gamma)\klcirc\pi}{\delta}{k}}\left[\Fa^L\bigl(\sigma(\mu_\sigma)_{\gamma\klcirc\pi_X},\delta;\gamma\klcirc\pi_X,z\bigr)\right]
  \end{align*}
  Here, $\mu_\rho$ and $\mu_\sigma$ are the parameters in $X$ and $Y$, respectively, and we write $g_{\mu_\rho}$ and $h_{\mu_\sigma}$ to indicate the corresponding lenses with those parameters.
  The context is $(\pi,k)$, with $\pi:1\klto M\otimes X$ in $\gauss(\Kl(\Pa))$ and $\pi_X$ denoting its $X$ marginal, and with continuation $k:\gauss(\Kl(\Pa))(1,M\otimes Z)\to\gauss(\Kl(\Pa))(1,N\otimes Z)$, for some choices of residual objects $M$ and $N$.
  The backwards channels $\rho$ and $\sigma$ are externally parameterized and state-dependent, so that $\rho(\mu_\rho)_{\pi_X}:Y\klto X$ is returned by $\rho(\mu_\rho)$ at $\pi_X$.
  Explicitly, $\rho$ has the type $X\to\cat{E}\bigl(\gauss(X), \gauss(\Kl(\Pa))(Y,X)\bigr)$, and $\sigma$ has the type $Y\to\cat{E}\bigl(\gauss(Y), \gauss(\Kl(\Pa))(Z,Y)\bigr)$.
  Finally, $\delta^* k$ is the function
  \[ \gauss(\Kl(\Pa))(1,M\otimes Y)\xto{\gauss(\Kl(\Pa))(1,M\otimes \delta)}\gauss(\Kl(\Pa))(1,M\otimes Z)\xto{k}\gauss(\Kl(\Pa))(1,N\otimes Z) \]
  obtained by pulling back $k$ along $\delta$.

  We therefore have $\efb{\pi}{\gamma}{\delta^*k} = \efb{(M\otimes\gamma)\klcirc\pi}{\delta}{k}$, meaning that we can rewrite the loss function as
  \[
  \E_{z\sim\efb{\pi}{\gamma}{\delta^*k}}\left[
    \Fa^L\bigl(\sigma(\mu_\sigma)_{\gamma\klcirc\pi_X},\delta;\gamma\klcirc\pi_X,z\bigr) +
    \E_{y\sim\sigma(\mu_\sigma)_{\gamma\klcirc\pi_X}(z)}\left[
      \Fa^L\bigl(\rho(\mu_\rho)_{\pi_X},\gamma;\pi_X,y\bigr)
      \right]
    \right] \, .
  \]
  In the dynamical semantics for stochastic gradient descent, $z$ and $\pi_X$ are supplied by the inputs to the dynamical system: the inputs replace the context for the game.
  Rewriting the loss accordingly gives a function
  \[
  f : (z, \pi_X, \mu_\rho, \mu_\sigma) \mapsto
  \Fa^L\bigl(\sigma(\mu_\sigma)_{\gamma\klcirc\pi_X},\delta;\gamma\klcirc\pi_X,z\bigr) +
  \E_{y\sim\sigma(\mu_\sigma)_{\gamma\klcirc\pi_X}(z)}\left[
    \Fa^L\bigl(\rho(\mu_\rho)_{\pi_X},\gamma;\pi_X,y\bigr)
    \right] \, .
  \]
  Next, we compute $\partial_{(\mu_\rho,\mu_\sigma)} f(z,\pi_X)$.
  We obtain
  \begin{align*}
    \partial_{(\mu_\rho,\mu_\sigma)} f(z,\pi_X) & = \left(
    \partial_{\mu_\rho} \, \E_{y\sim\sigma(\mu_\sigma)_{\gamma\klcirc\pi_X}(z)}\left[
      \Fa^L\bigl(\rho(\mu_\rho)_{\pi_X},\gamma;\pi_X,y\bigr)
      \right] , \,
    \partial_{\mu_\sigma} \, \Fa^L\bigl(\sigma(\mu_\sigma)_{\gamma\klcirc\pi_X},\delta;\gamma\klcirc\pi_X,z\bigr)
    \right) \\
    & = \left(
    \E_{y\sim\sigma(\mu_\sigma)_{\gamma\klcirc\pi_X}(z)}\left[ \partial_{\mu_\rho}
      \Fa^L\bigl(\rho(\mu_\rho)_{\pi_X},\gamma;\pi_X,y\bigr)
      \right] , \,
    \partial_{\mu_\sigma} \, \Fa^L\bigl(\sigma(\mu_\sigma)_{\gamma\klcirc\pi_X},\delta;\gamma\klcirc\pi_X,z\bigr)
    \right) \, .
  \end{align*}
  Now, $\Laplace(h\circ g)^u$ is defined as returning the point distribution on $(\mu_\rho,\mu_\sigma) - \lambda\, \partial_{(\mu_\rho,\mu_\sigma)} f(z,\pi_X)$:
  \begin{align*}
    & (\mu_\rho,\mu_\sigma) - \lambda\, \partial_{(\mu_\rho,\mu_\sigma)} f(z,\pi_X) \\
    & = \left(
    \E_{y\sim\sigma(\mu_\sigma)_{\gamma\klcirc\pi_X}(z)}\left[ \mu_\rho - \lambda\, \partial_{\mu_\rho}
      \Fa^L\bigl(\rho(\mu_\rho)_{\pi_X},\gamma;\pi_X,y\bigr)
      \right] , \,
    \mu_\sigma - \lambda\, \partial_{\mu_\sigma} \, \Fa^L\bigl(\sigma(\mu_\sigma)_{\gamma\klcirc\pi_X},\delta;\gamma\klcirc\pi_X,z\bigr)
    \right) \, .
  \end{align*}
  We can simplify this expression by making some auxiliary definitions
  \begin{gather*}
    \rho^u(a,\pi,y) := a - \lambda\, \partial_a \Fa^L\bigl(\rho(a)_\pi,\gamma;\pi,y\bigr) \\
    \sigma^u(b,\pi',z) := b - \lambda\, \partial_b \, \Fa^L\bigl(\sigma(b)_{\pi'},\delta;\pi',z\bigr)
  \end{gather*}
  so that
  \begin{equation} \label{eq:laplace-hg-lemma}
    (\mu_\rho,\mu_\sigma) - \lambda\, \partial_{(\mu_\rho,\mu_\sigma)} f(z,\pi_X)
    = \left(
    \E_{y\sim\sigma(\mu_\sigma)_{\gamma\klcirc\pi_X}(z)}\left[ \rho^u(\mu_\rho,\pi_X,y) \right], \,
    \sigma^u(\mu_\sigma,\gamma\klcirc\pi_X,z)
    \right) \, .
  \end{equation}
  By currying $\rho^u(\mu_\rho,\pi_X,y)$ into a function $\rho^u(\mu_\rho,\pi_X) : Y\to\Pa X$, we can simplify this still further, since
  \[
  \E_{y\sim\sigma(\mu_\sigma)_{\gamma\klcirc\pi_X}(z)}\left[ \rho^u(\mu_\rho,\pi_X,y) \right]
  = \rho^u(\mu_\rho,\pi_X)\klcirc\sigma(\mu_\sigma)_{\gamma\klcirc\pi_X}(z) \, .
  \]
  Since equation \eqref{eq:laplace-hg-lemma} defines $\Laplace(h\circ g)^u$, we have
  \begin{equation} \label{eq:laplace-hg}
    \Laplace(h\circ g)^u(\mu_\rho,\mu_\sigma,\pi,z) = \eta^\Pa_{X\times Y}\bigl(
    \rho^u(\mu_\rho,\pi)\klcirc\sigma(\mu_\sigma)_{\gamma\klcirc\pi}(z), \,
    \sigma^u(\mu_\sigma,\gamma\klcirc\pi,z)
    \bigr)
  \end{equation}
  where $\eta^\Pa_{X\times Y} : X\times Y\to\gauss(X\times Y)$ is the component of the unit of the monad $\Pa$ at $X\times Y$, which takes values in Dirac delta distributions and is therefore Gaussian.

  Next, we compute the update map of the system $\Laplace(h)\circ\Laplace(g)$, using Definitions \ref{def:blhom-comp} and \ref{def:hier-bicat} (which define composition in $\Hier^\nn_{\gauss(\Kl(\Pa_{\mathbf{Fd}}))}$).
  This update map is given by composing the `double strength'\footnote{
  The double strength is also known as the `commutativity' of the monad $\Pa$ with the product $\times$.
  It says that a pair of distributions $\pi$ on $X$ and $\chi$ on $Y$ can also be thought of as a joint distribution $(\pi,\chi)$ on $X\times Y$.
  It is Gaussian on Gaussians, as the product of two Gaussians is again Gaussian.}
  $\mathsf{dst}:\gauss(X)\times\gauss(Y)\to\gauss(X\times Y)$ after the following string diagram:
  \begin{equation} \label{diag:laplace-composite}
    \begin{tikzpicture}
	\begin{pgfonlayer}{nodelayer}
		\node [style=none] (0) at (2.75, 7) {};
		\node [style=none] (1) at (5.25, 7) {};
		\node [style=none] (2) at (2.75, 6) {};
		\node [style=none] (3) at (-4, 2.5) {};
		\node [style=none] (4) at (2.75, 2) {};
		\node [style=none] (5) at (2.75, 1) {};
		\node [style=none] (6) at (5.25, 1) {};
		\node [style=none] (7) at (5.25, 4) {};
		\node [style=none] (8) at (2.75, -5) {};
		\node [style=none] (9) at (5.25, -5) {};
		\node [style=none] (10) at (2.75, -6) {};
		\node [style=none] (11) at (2.75, -8) {};
		\node [style=none] (12) at (2.75, -10) {};
		\node [style=none] (13) at (2.75, -11) {};
		\node [style=none] (14) at (5.25, -11) {};
		\node [style=none] (15) at (5.25, -8) {};
		\node [style=none] (16) at (-5, -1) {};
		\node [style=none] (17) at (-3, -1) {};
		\node [style=none] (20) at (-5, -7) {};
		\node [style=none] (21) at (-3, -7) {};
		\node [style=none] (22) at (-3, -4) {};
		\node [style=black copier] (23) at (-8, -1.5) {};
		\node [style=none] (24) at (-5, -4) {};
		\node [style=none] (26) at (-5.5, 1) {};
		\node [style=none] (27) at (-11.5, -0.5) {};
		\node [style=none] (28) at (-9.5, -0.5) {};
		\node [style=none] (29) at (-11.5, -2.5) {};
		\node [style=none] (30) at (-9.5, -2.5) {};
		\node [style=none] (31) at (-9.5, -1.5) {};
		\node [style=none] (32) at (-11.5, -1.5) {};
		\node [style=none] (33) at (-10.5, -1.5) {$\gamma_*$};
		\node [style=none] (34) at (-4, -4) {$\sigma^\flat$};
		\node [style=none] (35) at (4, 4) {$\Laplace(g)^u$};
		\node [style=none] (36) at (4, -8) {$\Laplace(h)^u$};
		\node [style=none] (37) at (-5, -6) {};
		\node [style=black copier] (38) at (-7, -8) {};
		\node [style=none] (39) at (-5, -10) {};
		\node [style=black copier] (40) at (-13.5, 0.5) {};
		\node [style=none] (41) at (-20, -8) {};
		\node [style=none] (42) at (-20, 4) {};
		\node [style=none] (43) at (-20, 6) {};
		\node [style=none] (44) at (-5, 4) {};
		\node [style=none] (45) at (9, 4) {};
		\node [style=none] (46) at (9, -8) {};
		\node [style=none] (47) at (-11.5, 2.5) {};
		\node [style=none] (48) at (-20, 0.5) {};
		\node [style=none] (49) at (-18.25, 1) {$\gauss(X)$};
		\node [style=none] (50) at (-19.5, 4.5) {$Y$};
		\node [style=none] (51) at (-19.5, 6.5) {$X$};
		\node [style=none] (52) at (-19.5, -7.5) {$Z$};
		\node [style=none] (53) at (7.75, -7.5) {$\gauss(Y)$};
		\node [style=none] (54) at (7.75, 4.5) {$\gauss(X)$};
		\node [style=none] (55) at (2.75, 4) {};
		\node [style=none] (56) at (-5, -2) {};
		\node [style=none] (60) at (-10.25, -5) {};
		\node [style=black copier] (61) at (-15.5, 4) {};
	\end{pgfonlayer}
	\begin{pgfonlayer}{edgelayer}
		\draw (0.center) to (1.center);
		\draw (1.center) to (6.center);
		\draw (0.center) to (5.center);
		\draw (5.center) to (6.center);
		\draw (8.center) to (9.center);
		\draw (9.center) to (14.center);
		\draw (8.center) to (13.center);
		\draw (13.center) to (14.center);
		\draw (16.center) to (17.center);
		\draw (16.center) to (20.center);
		\draw (20.center) to (21.center);
		\draw (17.center) to (21.center);
		\draw [in=180, out=0, looseness=1.25] (22.center) to (4.center);
		\draw [in=180, out=0, looseness=1.50] (26.center) to (11.center);
		\draw [in=180, out=-90, looseness=1.25] (23) to (24.center);
		\draw (31.center) to (23);
		\draw (27.center) to (28.center);
		\draw (28.center) to (30.center);
		\draw (27.center) to (29.center);
		\draw (29.center) to (30.center);
		\draw (39.center) to (12.center);
		\draw [in=-180, out=90, looseness=1.25] (38) to (37.center);
		\draw [in=180, out=-90, looseness=1.25] (38) to (39.center);
		\draw (41.center) to (38);
		\draw [in=180, out=-90, looseness=1.25] (40) to (32.center);
		\draw [in=180, out=0, looseness=0.50] (43.center) to (2.center);
		\draw [in=180, out=0] (42.center) to (44.center);
		\draw [in=180, out=0, looseness=1.50] (44.center) to (10.center);
		\draw (7.center) to (45.center);
		\draw (15.center) to (46.center);
		\draw [in=180, out=90, looseness=1.25] (40) to (47.center);
		\draw [in=180, out=0, looseness=1.25] (47.center) to (3.center);
		\draw (48.center) to (40);
		\draw [in=-180, out=90, looseness=1.25] (23) to (26.center);
		\draw [in=-180, out=0, looseness=1.25] (3.center) to (55.center);
		\draw [in=180, out=-90, looseness=1.25] (61) to (60.center);
		\draw [in=-180, out=0, looseness=1.25] (60.center) to (56.center);
	\end{pgfonlayer}
\end{tikzpicture}

  \end{equation}
  Here, $\sigma^\flat$ denotes the uncurrying of the parameterized state-dependent channel $\sigma : Y\to\Fun{Stat}(Y)(Z,Y)$: we can equivalently write the type of $\sigma$ as $Y\to\cat{E}\bigl(\gauss(Y), \gauss(\Kl(\Pa))(Z,Y)\bigr)$, which we can uncurry twice to give the type $Y\times\gauss(Y)\times Z\to\gauss(Y)$.

  Observe now that we can write $\Laplace(g)^u$ and $\Laplace(h)^u$ as
  \begin{gather*}
    \Laplace(g)^u(a,\pi,y) = \eta^\Pa_X\bigl(\rho^u(a,\pi,y)\bigr) \\
    \Laplace(h)^u(b,\pi',z) = \eta^\Pa_Y\bigl(\sigma^u(b,\pi',z)\bigr)
  \end{gather*}
  and that $\eta^\Pa_{X\times Y} = \mathsf{dst}(\eta^\Pa_X,\eta^\Pa_Y)$.
  Reading the string diagram and applying this equality, we find that it represents $\bigl(\Laplace(h)\circ\Laplace(g)\bigr)^u(\mu_\rho,\mu_\sigma,\pi,z)$ as
  \[ \eta^\Pa_{X\times Y}\bigl(
  \rho^u(\mu_\rho,\pi)\klcirc\sigma(\mu_\sigma)_{\gamma\klcirc\pi}(z), \,
  \sigma^u(\mu_\sigma,\gamma\klcirc\pi,z)
  \bigr) \]
  which is precisely the same as the definition of $\Laplace(h\circ g)^u$ in equation \eqref{eq:laplace-hg}.

  Therefore, as required, $\bigl(\Laplace(h)\circ\Laplace(g)\bigr)^u = \Laplace(h\circ g)^u$.

  Finally, because the functor $\Laplace$ is identity-on-objects, the unit and multiplication of its monoidal structure are easily seen to be given by identity morphisms, and so $\Laplace$ is strict monoidal: $\Laplace$ maps the structure morphisms to constant dynamical systems emitting the structure morphisms of $\Hier^\nn_{\gauss(\Kl(\Pa_{\mathbf{Fd}}))}$, and so the associativity and unitality conditions are satisfied.
\end{proof}

\begin{rmk} \label{rmk:mean-field}
  From the diagram \eqref{diag:laplace-composite}, we can refine our understanding of what is known in the literature as the \textit{mean field} approximation \parencite[{around eq.39}]{Friston2007Variational}, in which the posterior over $X\otimes Y$ is assumed at each instant of time to have independent marginals.
  We note that, even though the backwards output maps emit posterior distributions with means determined entirely by their local parameterization, and even though these parameters are updated by the tensor $\Laplace(g)^u\otimes\Laplace(h)^u$, the resulting dynamical states are correlated across time by the composition rule:
  this is made very clear by the wiring of diagram \eqref{diag:laplace-composite}, since both factors $\Laplace(g)^u$ and $\Laplace(h)^u$ have common inputs.
  We also note that, even if the means of the emitted posteriors are entirely parameter-determined, this is not true of their covariances, which are functions of both the prior and the observation.
  The operational result of these observations is that the functorial (and pictorial) approach advocated here (as opposed to writing down a complete, and complex, joint distribution for each model of interest and proceeding from there) helps us understand the structural properties of complex systems---where it is otherwise easy to get lost in the weeds.
\end{rmk}

\begin{rmk}
  Above we exhibited the Laplace doctrine directly as a functor \[ \gauss(\Kl(\Pa_{\mathbf{Fd}}))\hookrightarrow\cat{G}\to\Hier^\nn_{\gauss(\Kl(\Pa_{\mathbf{Fd}}))} \, . \]
  In fact, Proposition \ref{prop:laplace-sgd} implies that it factors further, as
  \[ \gauss(\Kl(\Pa_{\mathbf{Fd}})) \hookrightarrow \cat{G} \xto{\nabla} \DiffHier_{\gauss(\Kl(\Pa_{\mathbf{Fd}}))} \xto{\H\Fun{Naive}_{k}} \Hier^\nn_{\gauss(\Kl(\Pa_{\mathbf{Fd}}))} \]
  where $\nabla : \cat{G} \to \DiffHier_{\gauss(\Kl(\Pa_{\mathbf{Fd}}))}$ takes an externally parameterized statistical game and returns a differential system that performs gradient descent on its loss function with respect to its parameterization.
  We leave the precise exhibition of this factorisation for future work.
\end{rmk}

\subsection{The Hebb-Laplace doctrine} \label{sec:doctrines-hebb-laplace}

The Laplace doctrine constructs dynamical systems that produce progressively better posterior approximations given a fixed forwards channel, but natural adaptive systems do more than this:
they also refine the forwards channels themselves, in order to produce better predictions.
In doing so, these systems better realize the abstract nature of autoencoder games, for which improving performance means improving both prediction as well as inversion.
To be able to improve the forwards channel requires allowing some freedom in its choice, which means giving it a nontrivial parameterization.

The Hebb-Laplace doctrine that we introduce in this section therefore modifies the Laplace doctrine by fixing a class of parameterized forwards channels and performing stochastic gradient descent with respect to both these parameters as well as the posterior means;
we call it the \textit{Hebb}-Laplace doctrine as the particular choice of forwards channels results in their parameter-updates resembling the `local' Hebbian plasticity known from neuroscience, in which the strength of the connection between two neurons is adjusted according to their correlation. (Here, we could think of the `neurons' as encoding the level of activity along a basis vector.)

We begin by defining the category of these parameterized forwards channels, after which we introduce Hebbian-Laplacian games and the resulting Hebb-Laplace doctrine, which is derived similarly to the Laplace doctrine above.
Recall from Definition \ref{def:ext-para} that we write $\eP\cat{C}$ to denote the external parameterization of $\cat{C}$ in its base of enrichment $\cat{E}$.

\begin{defn}
  Let $\cat{H}$ denote the subcategory of $\eP\gauss(\para(\star))$ generated by the structure morphisms of the symmetric monoidal category $\gauss(\para(\star))$ (trivially parameterized), and by morphisms $X\to Y$ of the form (written in $\cat{E}$)
  \begin{align*}
    \Theta_X & \to \gauss(\para(\star))(X,Y) \\
    \theta\;\; & \mapsto \;\; \Bigl( x \mapsto \theta\, h(x) + \omega \Bigr)
  \end{align*}
  where $h$ is a differentiable map $X\to Y$, $\Theta_X$ is the vector space of square matrices on $X$, and $\omega$ is sampled from a Gaussian distribution on $Y$.
\end{defn}

Note that there is a canonical embedding of $\eP\gauss(\para(\star))$ into $\eP\Kl(\Pa_{\mathbf{Fd}})$, obtained in the image of Proposition \ref{prop:embed-para-ast-kl} under the external parameterization $\eP$.

\begin{defn}
  A \textit{Hebbian-Laplacian statistical game} is a parameterized statistical game $(\gamma,\rho,\phi):(X,X)\xto{\Theta_X\times X}(Y,Y)$ satisfying the following conditions:
  \begin{enumerate}
  \item $X$ and $Y$ are finite-dimensional Cartesian spaces;
  \item the forward channel $\gamma$ is a morphism in $\cat{H}$ (\textit{i.e.}, of the form $x\mapsto\theta\, h(x) + \omega$);
  \item the backward channel is as for a Laplacian statistical game (Definition \ref{def:laplace-game});
  \item the loss function is as for a Laplacian statistical game, with the substitution $\gamma\mapsto\gamma(\theta)$ for parameter $\theta:\Theta_X$.
  \end{enumerate}
  We will write $\cat{G}_\Ha$ to denote the subcategory of $\PSGame{}$ generated by Hebbian-Laplacian statistical games and by the structure morphisms of a monoidal category.
\end{defn}

\begin{defn}
  Suppose $\gamma:X\to Y$ is a morphism in $\cat{H}$.
  Then the discrete-time Hebb-Laplace doctrine defines a system $\Hebb(\gamma):(X,X)\to(Y,Y)$ in $\Hier^\nn_{\gauss(\Kl(\Pa_{\mathbf{Fd}}))}$ as follows (using the representation of Proposition \ref{prop:unpack-hier}).
  \begin{itemize}
  \item The state space is $\Theta_X\times X$ (where $\Theta_X$ is again the vector space of square matrices on $X$);
  \item the forwards output map $\Hebb(\gamma)^o_1:\Theta_X\times X\times X \to \gauss(Y)$ is given by $\gamma$:
    \[ \Hebb(\gamma)^o_1 := \Theta_X\times X\times X \xto{\mathsf{proj}_{1,3}} \Theta_X\times X \xto{\gamma^\flat} \gauss(Y) \]
    where $\gamma^\flat$ is the uncurried form of the morphism $\gamma:\Theta_X\to\gauss(\para(\star))(X,Y)$ in the image of the embedding of $\cat{H}$ in $\eP\Kl(\Pa)$;
  \item the backwards output map $\Hebb(\gamma)^o_2:\Theta_X\times X\times\gauss(X)\times Y\to \gauss(X)$ is given by:
    \begin{align*}
      \Hebb(\gamma)^o_2 : \Theta_X\times X\times\gauss(X)\times Y&\to\rr^{|X|}\times\rr^{|X|\times|X|}\hookrightarrow\gauss(X) \\
      (\theta,x,\pi,y)&\mapsto\bigl(x,\Sigma_\rho(\theta,x,\pi,y)\bigr)
    \end{align*}
    where the inclusion picks the Gaussian state with the given statistical parameters, whose covariance $\Sigma_\rho(\theta,x,\pi,y) := \left(\partial_x^2 E_{(\pi,\gamma(\theta))}\right)(x, y)^{-1}$ is defined following equation \eqref{eq:laplace-sigma-rho-pi} (Lemma \ref{lemma:laplace-approx});
  \item the update map $\Hebb(\gamma)^u:\Theta_X\times X\times\gauss(X)\times Y\to\gauss(\Theta_X\times X)$ optimizes the parameter for $\gamma$ as well as the mean of the posterior (as in the Laplace doctrine):
    \begin{align*}
      \Hebb(\gamma)^u : \Theta_X\times X \times \Pa X \times Y & \to \Pa(\Theta_X\times X) \\
      (\theta, x, \pi, y) & \mapsto \eta^\Pa_{\Theta_X\times X}\bigl(\theta^u(\theta,x,y), \mu_\rho(\theta,x,\pi,y)\bigr)
    \end{align*}
    where $\eta^\Pa$ denotes the unit of the monad $\Pa$, and $\theta^u$ and $\mu_\rho$ are defined by
    \begin{gather*}
    \theta^u(\theta,x,y) := \theta - \lambda_\theta\, \eta_{\gamma(\theta)}(y,x)\, h(x)^T \\
    \mu_\rho(\theta,x,\pi,y) := x + \lambda_\rho\, \partial_x h(x)^T \theta^T \eta_{\gamma(\theta)}(y,x) - \lambda_\rho\, \eta_\pi(x) \, .
    \end{gather*}
    Here, $\lambda_\theta,\lambda_\rho:\rr_+$ are chosen learning rates, and the precision-weighted error terms $\eta$ are again as in equation \eqref{eq:precis-err} (Remark \ref{rmk:laplace-err}).
  \end{itemize}
\end{defn}

\begin{rmk}
  The `Hebbian' part of the Hebb-Laplace doctrine enters in the forwards-parameter update map, $\theta^u(\theta,x,y) = \theta - \lambda_\theta\, \eta_{\gamma(\theta)}(y,x)\, h(x)^T$, since the change in parameters is proportional to something resembling the correlation between `pre-synaptic' and `post-synaptic' activity.
  Here, the post-synaptic activity is represented by the term $h(x)$: we may think of the components of the vector $x$ as each representing the ``internal activity'' of a single neuron, and the ``activation function'' $h$ as returning the corresponding firing rates; these are `post-synaptic' as the firing is emitted down a neuron's axon, which occurs computationally `after' the neuron's synaptic inputs.
  The synaptic inputs (generating the pre-synaptic activity) are then thought to be represented by the error term $\eta_{\gamma(\theta)}(y,x)$, so that expected trajectory of the outer product $\eta_{\gamma(\theta)}(y,x)\, h(x)^T$ computes the correlation between pre- and post-synaptic acivity.

  Note that this means that typically one assumes that $\lambda_\theta < \lambda_\rho$, because the neural activity $x$ itself must change on a faster timescale than the synaptic weights $\theta$, in order for $\theta$ to learn these correlations.
\end{rmk}

Given the foregoing definition, we obtain the following theorem.

\begin{thm} \label{thm:hebb-laplace}
  The Hebb-Laplace doctrine $\Hebb$ defines an identity-on-objects strict monoidal functor $\cat{H}\hookrightarrow\cat{G}_\Ha\to\Hier^\nn_{\gauss(\Kl(\Pa_{\mathbf{Fd}}))}$.
\end{thm}

This theorem follows in the same way as the corresponding result for the Laplace doctrine; and so we begin with a small lemma, and subsequently show that the doctrine arises by stochastic gradient descent, before putting the pieces together to prove the theorem itself.

\begin{lemma} \label{lemma:h-embeds}
  There is an identity-on-objects strict monoidal embedding $\cat{H}\hookrightarrow\cat{G}_\Ha$.
  \begin{proof}[Proof sketch]
    The proof proceeds much as the proof of Lemma \ref{lemma:gauss-embeds}, except that the forwards channels of games in the image of the embedding are given by the parameterized morphisms of $\cat{H}$.
  \end{proof}
\end{lemma}

\begin{prop} \label{prop:hebb-sgd}
  Given a Hebbian-Laplacian statistical game $(\gamma,\rho,\phi):(X,X)\xto{\Theta_X\times X}(Y,Y)$, $\Hebb(\gamma)$ is obtained by stochastic gradient descent of the loss function $\phi$ with respect to the weight matrix $\theta:\Theta_X$ of the channel $\gamma$ and the mean $x:X$ of the posterior $\rho$.
  \begin{proof}
    The proof proceeds much as the proof of Proposition \ref{prop:laplace-sgd}, except now the forwards channel $\gamma$ is parameterized: this gives us another factor against which to perform gradient descent, and furthermore means that $\gamma(\theta)$ must be substituted for $\gamma$ in expressions in the derivation of $\mu_\rho$.

    The first such expression is the definition of the loss function $\phi : \sum_{(\theta,x):\Theta_X\times X} \cctx\bigl(\gamma(\theta),\rho(x)\bigr)\to\rr$;
    we will write $\phi_{(\theta,x)}$ for the component of $\phi$ at $(\theta,x)$ with the corresponding type $\cctx\bigl(\gamma(\theta),\rho(x)\bigr)\to\rr$.
    We have $\phi_{(\theta,x)}(\pi,k) = \E_{y\sim\efb{\pi}{\gamma(\theta)}{k}}\bigl[\Fa^L(y)\bigr]$, where now
    \[
    \Fa^L(y) = {}- \log p_{\gamma(\theta)}(y|x) - \log p_\pi(x) - S_X\bigl[\rho(x,\pi,y)\bigr] \, .
    \]
    We find
    \begin{align*}
      \partial_x \Fa^L(y) &= \partial_x E_{(\pi,\gamma(\theta))} \\
      &= {}- \partial_x \mu_{\gamma(\theta)} (x)^T {\Sigma_{\gamma(\theta)}(x)}^{-1} \epsilon_{\gamma(\theta)}(y,x) + {\Sigma_\pi}^{-1} \epsilon_\pi(x) \\
      &= {}- \partial_x h(x)^T \theta^T \eta_{\gamma(\theta)}(y,x) + \eta_\pi(x)
    \end{align*}
    and
    \begin{align*}
      \partial_\theta \Fa^L(y) &= \partial_\theta E_{(\pi,\gamma(\theta))} \\
      &= {}- \frac{\partial_\theta}{2} \innerprod{\epsilon_{\gamma(\theta)}(y,x)}{{\Sigma_{\gamma(\theta)}(x)}^{-1} {\epsilon_{\gamma(\theta)}(y,x)}} \\
      &= {}- \frac{\partial_\theta}{2} \innerprod{y - \theta h(x)}{{\Sigma_{\gamma(\theta)}(x)}^{-1} {\bigl(y - \theta h(x)\bigr)}} \\
      &= {\Sigma_{\gamma(\theta)}(x)}^{-1} {\bigl(y - \theta h(x)\bigr)} \, h(x)^T \\
      &= {\Sigma_{\gamma(\theta)}(x)}^{-1} {\epsilon_{\gamma(\theta)}(y,x)} \, h(x)^T \\
      &= \eta_{\gamma(\theta)}(y,x)\, h(x)^T \, .
    \end{align*}
    Consequently, we have
    \begin{align*}
      \mu_\rho(\theta,x,\pi,y) &= x + \lambda_\rho\, \partial_x h(x)^T \theta^T \eta_{\gamma(\theta)}(y,x) - \lambda_\rho\, \eta_\pi(x) \\
      &= x - \lambda_\rho\, \partial_x \Fa^L(y)
    \end{align*}
    and
    \begin{align*}
      \theta^u(\theta,x,y) &= \theta - \lambda_\theta\, \eta_{\gamma(\theta)}(y,x)\, h(x)^T \\
      &= \theta - \lambda_\theta\, \partial_\theta \Fa^L(y) \, ,
    \end{align*}
    and this means that we can write
    \begin{align*}
      \Hebb(\gamma)^u(\theta,x,\pi,y) &= \eta^\Pa_{\Theta_X\times X} \circ \Bigl((\theta,x) - (\lambda_\theta,\lambda_\rho)\, \partial_{(\theta,x)} \Fa^L(y)\Bigr) \\
      &= \eta^\Pa_{\Theta_X\times X} \circ \Bigl(p - \lambda\, \partial_p \Fa^L(y)\Bigr)
    \end{align*}
    where $p := (\theta,x)$ and $\lambda := (\lambda_\theta,\lambda_\rho)$, which establishes that $\Hebb(\gamma)^u$ descends the gradient of the free energy with respect to the parameterization $p$.

    Finally, with $y$ sampled from a fixed context, we can see that the expected trajectory of $\Hebb(\gamma)$ follows
    \begin{align*}
      & \E_{y\sim\efb{\pi}{\gamma(\theta)}{k}} \, \Bigl(p - \lambda\, \partial_p\, \Fa^L(y) \Bigr) \\
      &= \Bigl(p - \lambda\, \partial_p\, \E_{y\sim\efb{\pi}{\gamma(\theta)}{k}} \left[ \Fa^L(y) \right]\Bigr) \\
      &= \Bigl(p - \lambda\, \partial_p\, \phi_p(\pi,k) \Bigr)
    \end{align*}
    which demonstrates that $\Hebb(\gamma)$ performs stochastic gradient descent of the loss function.
  \end{proof}
\end{prop}

\begin{proof}[Proof of Theorem \ref{thm:hebb-laplace}]
  Lemma \ref{lemma:h-embeds} gives us the first factor $\cat{H}\hookrightarrow\cat{G}_\Ha$, so we only need to establish that the Hebb-Laplace doctrine obtains by pulling a functor $\cat{G}_\Ha\to\Hier^\nn_{\gauss(\Kl(\Pa_{\mathbf{Fd}}))}$ back along this inclusion.
  We now turn to establishing that stochastic gradient descent returns the desired identity-on-objects functor $\cat{G}_\Ha\to\Hier^\nn_{\gauss(\Kl(\Pa_{\mathbf{Fd}}))}$.
  Proposition \ref{prop:hebb-sgd} shows that $\Hebb$ is obtained by applying stochastic gradient descent to morphisms in $\cat{G}_\Ha$, so we need to show that the resulting mapping is functorial.

  As in the case of Theorem \ref{thm:laplace-functorial}, the structure morphisms are preserved trivially:
  they have trivial parameterization, and so stochastic gradient descent returns the trivial systems constantly emitting the corresponding lenses;
  in particular, this means that stochastic gradient descent preserves identities.

  We now show that, for composable games $h$ and $g$, $\Hebb(h)\circ\Hebb(g) = \Hebb(h\circ g)$.
  This means demonstrating equalities between state spaces, output maps, and update maps.
  As for Theorem \ref{thm:laplace-functorial}, the state spaces are given by the external parameterization, and the parameterization of the composite game $h\circ g$ and the state space of the composite system $\Hebb(h)\circ\Hebb(g)$ are both given by taking the product of the factors, and so the state spaces on the left- and right-hand sides of the desired equation are equal.

  The proof that the equality holds for output maps is also as in the proof of Theorem \ref{thm:laplace-functorial}:
  the output of a composite system is given by composing the output lenses of the factors, which is the same as the output returned by $\Hebb$ on a composite game, since outputs in the image of $\Hebb$ are obtained by filling in the external parameter.

  We now turn to the update maps, for which we need to show that $\bigl(\Hebb(h)\circ\Hebb(g)\bigr)^u = \Hebb(h\circ g)^u$.
  Suppose $g:=(\gamma,\rho,\phi):(X,X)\to(Y,Y)$ and $h:=(\sigma,\delta,\psi):(Y,Y)\to(Z,Z)$ are Hebbian-Laplacian statistical games; we will denote the corresponding parameters by $(\theta_\gamma,\mu_\rho)$ and $(\theta_\delta,\mu_\sigma)$ respectively.
  Following the proof of Theorem \ref{thm:laplace-functorial}, we can write the loss function of the composite game $(\sigma,\delta,\psi)\circ(\gamma,\rho,\phi)$ as
  \begin{align*}
  & \E_{z\sim\efb{\pi}{\gamma(\theta_\gamma)}{\delta(\theta_\delta)^*k}}\Big[
    \Fa^L\bigl(\sigma(\mu_\sigma)_{\gamma(\theta_\gamma)\klcirc\pi_X},\delta(\theta_\delta);\gamma(\theta_\gamma)\klcirc\pi_X,z\bigr) \\
  & \qquad\quad +
    \E_{y\sim\sigma(\mu_\sigma)_{\gamma(\theta_\gamma)\klcirc\pi_X}(z)}\left[
      \Fa^L\bigl(\rho(\mu_\rho)_{\pi_X},\gamma(\theta_\gamma);\pi_X,y\bigr)
      \right]
    \Big] \, .
  \end{align*}
  (This expression is obtained by making the substitutions $\gamma\mapsto\gamma(\theta_\gamma)$ and $\delta\mapsto\delta(\theta_\delta)$ in the corresponding expression in the proof of Theorem \ref{thm:laplace-functorial}.)

  As before, $z$ and $\pi_X$ are supplied by the inputs to the dynamical system, and so we obtain a function
  \begin{align*}
    f : (z,\pi_X,\theta_\gamma,\mu_\rho,\theta_\delta,\mu_\sigma) & \mapsto \Fa^L\bigl(\sigma(\mu_\sigma)_{\gamma(\theta_\gamma)\klcirc\pi_X},\delta(\theta_\delta);\gamma(\theta_\gamma)\klcirc\pi_X,z\bigr) \\
  & \qquad\quad +
    \E_{y\sim\sigma(\mu_\sigma)_{\gamma(\theta_\gamma)\klcirc\pi_X}(z)}\left[
      \Fa^L\bigl(\rho(\mu_\rho)_{\pi_X},\gamma(\theta_\gamma);\pi_X,y\bigr)
      \right] \, .
  \end{align*}
  If we write $p:=(\theta_\gamma,\mu_\rho)$ and $q:=(\theta_\delta,\mu_\sigma)$, then $(p,q)$ denotes the parameter for $h\circ g$.
  Since $\Hebb$ performs stochastic gradient descent with respect to the parameterization, $\Hebb(h\circ g)^u$ is therefore defined as returning the point distribution on $(p,q) - \lambda\, \partial_{(p,q)} f(z,\pi_X)$, where $\lambda := (\lambda_p, \lambda_q)$, and $\lambda_p = (\lambda_\gamma,\lambda_\rho)$ and $\lambda_q = (\lambda_\delta,\lambda_\sigma)$.

  We have $\partial_{(p,q)} f = \bigl(\partial_p f, \partial_q f)$ and so
  \[
  (p,q) - \lambda\, \partial_{(p,q)} f(z,\pi_X)
  = \bigl(p - \lambda_p\, \partial_p f(z,\pi_X), q - \lambda_q\, \partial_q f(z,\pi_X)\bigr) \, .
  \]
  We make some auxiliary definitions
  \begin{gather*}
    g^u(\theta_\gamma,\mu_\rho,\pi,y) := (\theta_\gamma,\mu_\rho) - \lambda_p\, \partial_{(\theta_\gamma,\mu_\rho)} \Fa^L\bigl(\rho(\mu_\rho)_\pi,\gamma(\theta_\gamma);\pi,y\bigr) \\
    h^u(\theta_\delta,\mu_\sigma,\pi',z) := (\theta_\delta,\mu_\sigma) - \lambda_q\, \partial_{(\theta_\delta,\mu_\sigma)} \Fa^L\bigl(\sigma(\mu_\sigma)_{\pi'},\delta(\theta_\delta);\pi',z\bigr)
  \end{gather*}
  and find that
  \begin{align*}
    &  (p,q) - \lambda\, \partial_{(p,q)} f(z,\pi_X) \\
    &= (\theta_\gamma,\mu_\rho,\theta_\delta,\mu_\sigma) - \lambda\, \partial_{(\theta_\gamma,\mu_\rho,\theta_\delta,\mu_\sigma)} f(z,\pi_X) \\
    &= \bigl((\theta_\gamma,\mu_\rho) - \lambda_p\, \partial_{(\theta_\gamma,\mu_\rho)} f(z,\pi_X), (\theta_\delta,\mu_\sigma) - \lambda_q\, \partial_{(\theta_\delta,\mu_\sigma)} f(z,\pi_X)\bigr) \\
    &= \left( \E_{y\sim\sigma(\mu_\sigma)_{\gamma(\theta_\gamma)\klcirc\pi_X}(z)} \bigl[ g^u(\theta_\gamma,\mu_\rho,\pi_X,y) \bigr],\, h^u\bigl(\theta_\delta,\mu_\sigma,\gamma(\theta_\gamma)\klcirc\pi_X,z\bigr) \right) \\
    &= \Bigl( g^u(\theta_\gamma,\mu_\rho,\pi_X) \klcirc \sigma(\mu_\sigma)_{\gamma(\theta_\gamma)\klcirc\pi_X}(z),\, h^u\bigl(\theta_\delta,\mu_\sigma,\gamma(\theta_\gamma)\klcirc\pi_X,z\bigr) \Bigr) \, .
  \end{align*}
  Writing $PQ$ to denote the composite parameter space $\Theta_X\times X\times\Theta_Y\times Y$, the foregoing computation defines $\Hebb(h\circ g)^u : PQ\times\gauss(X)\times Z \to \gauss(PQ)$ as
  \begin{equation} \label{eq:hebb-hg-u}
    \Hebb(h\circ g)^u(\theta_\gamma,\mu_\rho,\theta_\delta,\mu_\sigma,\pi,z)
    = \eta^\Pa_{PQ} \Bigl( g^u(\theta_\gamma,\mu_\rho,\pi_X) \klcirc \sigma(\mu_\sigma)_{\gamma(\theta_\gamma)\klcirc\pi_X}(z),\, h^u\bigl(\theta_\delta,\mu_\sigma,\gamma(\theta_\gamma)\klcirc\pi_X,z\bigr) \Bigr) \, .
  \end{equation}

  The update map of the composite system $\bigl(\Hebb(h)\circ\Hebb(g)\bigr)^u$ is given by composing the double strength $\Fun{dst}:\gauss(P)\times\gauss(Q)\to\gauss(P\times Q)$ after the string diagram
  
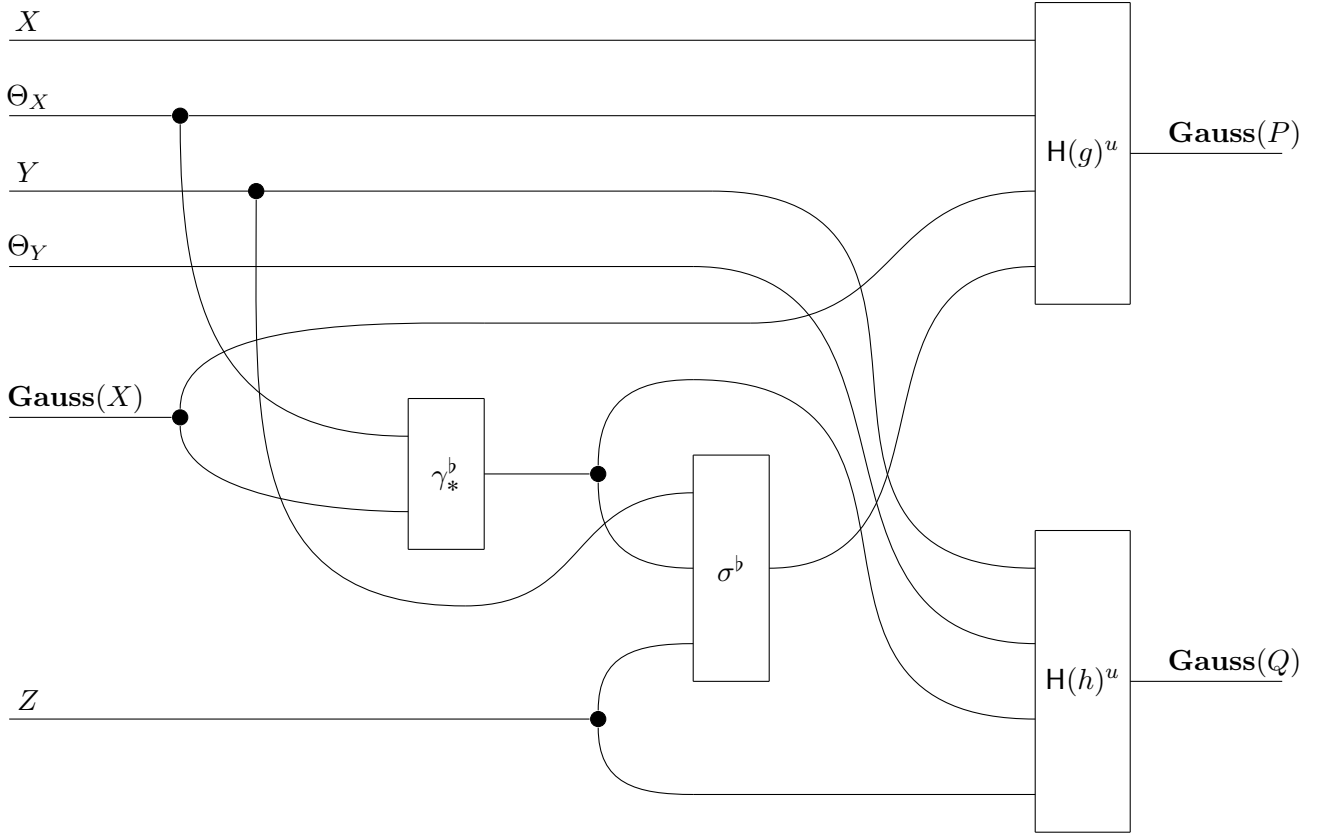
\[ \begin{tikzpicture}
	\begin{pgfonlayer}{nodelayer}
		\node [style=none] (0) at (3.5, 11) {};
		\node [style=none] (1) at (6, 11) {};
		\node [style=none] (2) at (3.5, 10) {};
		\node [style=none] (3) at (-4, 2.5) {};
		\node [style=none] (4) at (3.5, 4) {};
		\node [style=none] (5) at (3.5, 3) {};
		\node [style=none] (6) at (6, 3) {};
		\node [style=none] (7) at (6, 7) {};
		\node [style=none] (8) at (3.5, -3) {};
		\node [style=none] (9) at (6, -3) {};
		\node [style=none] (10) at (3.5, -4) {};
		\node [style=none] (11) at (3.5, -8) {};
		\node [style=none] (12) at (3.5, -10) {};
		\node [style=none] (13) at (3.5, -11) {};
		\node [style=none] (14) at (6, -11) {};
		\node [style=none] (15) at (6, -7) {};
		\node [style=none] (16) at (-5.5, -1) {};
		\node [style=none] (17) at (-3.5, -1) {};
		\node [style=none] (20) at (-5.5, -7) {};
		\node [style=none] (21) at (-3.5, -7) {};
		\node [style=none] (22) at (-3.5, -4) {};
		\node [style=black copier] (23) at (-8, -1.5) {};
		\node [style=none] (24) at (-5.5, -4) {};
		\node [style=none] (26) at (-5.5, 1) {};
		\node [style=none] (27) at (-13, 0.5) {};
		\node [style=none] (28) at (-11, 0.5) {};
		\node [style=none] (29) at (-13, -3.5) {};
		\node [style=none] (30) at (-11, -3.5) {};
		\node [style=none] (31) at (-11, -1.5) {};
		\node [style=none] (32) at (-13, -2.5) {};
		\node [style=none] (33) at (-12, -1.5) {$\gamma^\flat_*$};
		\node [style=none] (34) at (-4.5, -4) {$\sigma^\flat$};
		\node [style=none] (35) at (4.75, 7) {$\Hebb(g)^u$};
		\node [style=none] (36) at (4.75, -7) {$\Hebb(h)^u$};
		\node [style=none] (37) at (-5.5, -6) {};
		\node [style=black copier] (38) at (-8, -8) {};
		\node [style=none] (39) at (-5.5, -10) {};
		\node [style=black copier] (40) at (-19, 0) {};
		\node [style=none] (41) at (-23.5, -8) {};
		\node [style=none] (42) at (-23.5, 6) {};
		\node [style=none] (43) at (-23.5, 10) {};
		\node [style=none] (44) at (-5, 6) {};
		\node [style=none] (45) at (10, 7) {};
		\node [style=none] (46) at (10, -7) {};
		\node [style=none] (47) at (-11, 2.5) {};
		\node [style=none] (48) at (-23.5, 0) {};
		\node [style=none] (49) at (-21.75, 0.5) {$\gauss(X)$};
		\node [style=none] (50) at (-23, 6.5) {$Y$};
		\node [style=none] (51) at (-23, 10.5) {$X$};
		\node [style=none] (52) at (-23, -7.5) {$Z$};
		\node [style=none] (53) at (8.75, -6.5) {$\gauss(Q)$};
		\node [style=none] (54) at (8.75, 7.5) {$\gauss(P)$};
		\node [style=none] (55) at (3.5, 6) {};
		\node [style=none] (56) at (-5.5, -2) {};
		\node [style=none] (60) at (-11.5, -5) {};
		\node [style=black copier] (61) at (-17, 6) {};
		\node [style=none] (62) at (-23.5, 8) {};
		\node [style=none] (63) at (-23, 8.5) {$\Theta_X$};
		\node [style=none] (64) at (-23.5, 4) {};
		\node [style=none] (65) at (-23, 4.5) {$\Theta_Y$};
		\node [style=none] (66) at (-13, -0.5) {};
		\node [style=none] (67) at (3.5, 8) {};
		\node [style=black copier] (68) at (-19, 8) {};
		\node [style=none] (69) at (3.5, -6) {};
		\node [style=none] (70) at (-5.5, 4) {};
	\end{pgfonlayer}
	\begin{pgfonlayer}{edgelayer}
		\draw (0.center) to (1.center);
		\draw (1.center) to (6.center);
		\draw (0.center) to (5.center);
		\draw (5.center) to (6.center);
		\draw (8.center) to (9.center);
		\draw (9.center) to (14.center);
		\draw (8.center) to (13.center);
		\draw (13.center) to (14.center);
		\draw (16.center) to (17.center);
		\draw (16.center) to (20.center);
		\draw (20.center) to (21.center);
		\draw (17.center) to (21.center);
		\draw [in=180, out=0, looseness=1.25] (22.center) to (4.center);
		\draw [in=180, out=0, looseness=1.50] (26.center) to (11.center);
		\draw [in=180, out=-90, looseness=1.25] (23) to (24.center);
		\draw (31.center) to (23);
		\draw (27.center) to (28.center);
		\draw (28.center) to (30.center);
		\draw (27.center) to (29.center);
		\draw (29.center) to (30.center);
		\draw (39.center) to (12.center);
		\draw [in=-180, out=90, looseness=1.25] (38) to (37.center);
		\draw [in=180, out=-90, looseness=1.25] (38) to (39.center);
		\draw (41.center) to (38);
		\draw [in=180, out=-90, looseness=0.75] (40) to (32.center);
		\draw [in=180, out=0, looseness=1.25] (43.center) to (2.center);
		\draw [in=180, out=0] (42.center) to (44.center);
		\draw [in=180, out=0, looseness=1.50] (44.center) to (10.center);
		\draw (7.center) to (45.center);
		\draw (15.center) to (46.center);
		\draw [in=180, out=90, looseness=0.75] (40) to (47.center);
		\draw [in=180, out=0, looseness=1.25] (47.center) to (3.center);
		\draw (48.center) to (40);
		\draw [in=-180, out=90, looseness=1.25] (23) to (26.center);
		\draw [in=-180, out=0, looseness=1.25] (3.center) to (55.center);
		\draw [in=180, out=-90, looseness=1.25] (61) to (60.center);
		\draw [in=-180, out=0, looseness=1.25] (60.center) to (56.center);
		\draw [in=180, out=-90, looseness=1.25] (68) to (66.center);
		\draw (62.center) to (68);
		\draw [in=180, out=0, looseness=1.25] (68) to (67.center);
		\draw (64.center) to (70.center);
		\draw [in=180, out=0, looseness=1.25] (70.center) to (69.center);
	\end{pgfonlayer}
\end{tikzpicture}
 \]
  where $\gamma^\flat_*$ indicates the uncurrying of the pushforwards of the parameterized forwards channel $\gamma$:
  \begin{align*}
    & \gamma : \Theta_X \to \gauss\bigl(\para(\star)\bigr)(X,Y) \\
    \xmapsto{\text{embeds}} \qquad&\qquad \Theta_X \to \gauss(\Kl(\Pa))(X,Y) \\
    \xmapsto{(-)_*} \qquad&\qquad \Theta_X \to \cat{E}\bigl(\gauss(\Kl(\Pa))(1,X),\gauss(\Kl(\Pa))(1,Y)\bigr) \\
    \xmapsto{\sim} \qquad&\qquad \Theta_X \to \cat{E}(\gauss(X),\gauss(Y)) \\
    \xmapsto{(-)^\flat} \qquad&\;\, \gamma^\flat_* : \Theta_X \times \gauss(X) \to \gauss(Y) \, .
  \end{align*}
  Next, note that we can write $\Hebb(g)^u$ and $\Hebb(h)^u$ as
  \begin{gather*}
    \Hebb(g)^u(\theta_\gamma,\mu_\rho,\pi,y) = \eta^\Pa_P\bigl( g^u(\theta_\gamma,\mu_\rho,\pi,y) \bigr) \\
    \Hebb(h)^u(\theta_\delta,\mu_\sigma,\pi',z) = \eta^\Pa_Q\bigl( h^u\bigl(\theta_\delta,\mu_\sigma,\pi',z\bigr) \bigr)
  \end{gather*}
  where $P := \Theta_X\times X$ and $Q := \Theta_Y\times Y$, and that $\eta^\Pa_{PQ} = \Fun{dst}(\eta^\Pa_P,\eta^\Pa_Q)$.
  Reading the string diagram and comparing with equation \eqref{eq:hebb-hg-u}, we therefore find that $\bigl(\Hebb(h)\circ\Hebb(g)\bigr)^u = \Hebb(h\circ g)^u$.

  Finally, the proof that $\Hebb$ is strict monoidal is precisely analogous to the proof that $\Laplace$ is strict monoidal: $\Hebb$ is identity-on-objects and maps structure morphisms to structure morphisms, so that the associativity and unitality conditions are immediately satisfied.
\end{proof}

%

\section{References}

\printbibliography[heading=none]

\end{document}